\documentclass{article}%
\usepackage{amsmath}
\usepackage{amsfonts}
\usepackage{amssymb}
\usepackage{graphicx}
\usepackage[body={15.5cm,21cm}, top=3cm]{geometry}
\usepackage{tikz-cd}
\usepackage{setspace}
\usepackage{hyperref}
\usepackage{comment}
\usepackage{alphalph}
\usepackage[titletoc]{appendix}
\usepackage{titlesec}
\usepackage{titletoc}
\usepackage{xcolor}
\usepackage{caption}%
\providecommand{\U}[1]{\protect \rule{.1in}{.1in}}
\renewcommand{\emph}[1]{\textcolor{black}{#1}}

\numberwithin{equation}{section}
\newtheorem{theorem}{Theorem}[section]

\newtheorem{condition}[theorem]{Assumption}

\newtheorem{definition}[theorem]{Definition}

\newtheorem{lemma}[theorem]{Lemma}
\newtheorem{notation}[theorem]{Notation}

\newtheorem{proposition}[theorem]{Proposition}
\newtheorem{remark}[theorem]{Remark}

\newenvironment{proof}[1][Proof]{\noindent \textbf{#1.} }{\  \rule{0.5em}{0.5em}}
\hypersetup{
colorlinks=true,
linkcolor=blue,
anchorcolor=blue,
citecolor=blue}

\usetikzlibrary{calc}

\begin{document}

\title{Representation of forward performance criteria with random endowment via FBSDE
and its application to forward optimized certainty equivalent\thanks{Author
order is alphabetical. All authors are co-first authors of this paper. The
authors would like to express their cordial gratitude to the editor and the
referees for their exceptionally careful reading and insightful comments.}}
\author{Gechun Liang\thanks{Department of Statistics, University of Warwick, Coventry,
CV4 7AL, U.K. Email: \texttt{g.liang@warwick.ac.uk}. Research partially
supported by National Natural Science Foundation of China (No. 12171169) and
Laboratory of Mathematics for Nonlinear Science, Fudan University.}
\and Yifan Sun\thanks{Department of Applied Mathematics, The Hong Kong Polytechnic
University, Hong Kong, P.R.C. Email: \texttt{yifan-11.sun@polyu.edu.hk}.}
\and Thaleia Zariphopoulou\thanks{Departments of Mathematics and IROM, The
University of Texas at Austin, Austin, Texas 78712, U.S.A. and Oxford-Man
Institute, University of Oxford, Oxford, OX2 6ED, U.K. Email:
\texttt{zariphop@math.utexas.edu}. The author would like to thank the
Mathematical Institute, University of Oxford and the Institute for
Mathematical and Statistical Innovation at the University of Chicago for their
hospitality, a substantial part of this work was completed while visiting
there.}}
\date{December 2023 (first version) \\ October 2025 (this version)}
\maketitle

\begin{abstract}
We extend the notion of forward performance criteria to settings with random
endowment in incomplete markets. Building on these results, we introduce and
develop the novel concept of \textit{forward optimized certainty equivalent
(forward OCE)}, which offers a genuinely dynamic valuation mechanism that
accommodates progressively adaptive market model updates, stochastic risk
preferences, and incoming claims with arbitrary maturities.

In parallel, we develop a new methodology to analyze the emerging stochastic
optimization problems by directly studying the candidate optimal control
processes for both the primal and dual problems. Specifically, we derive two
new systems of forward-backward stochastic differential equations (FBSDEs) and
establish necessary and sufficient conditions for optimality, and various
equivalences between the two problems. This new approach is general and
complements the existing one for forward performance criteria with random
endowment based on backward stochastic partial differential equations
(backward SPDEs) for the related value functions. We, also, consider
representative examples for both forward performance criteria with random
endowment and for forward OCE. Furthermore, for the case of exponential
criteria, we investigate the connection between forward OCE and forward
entropic risk measures.

\end{abstract}

\newpage

\section{Introduction}

This work contributes to the theory of forward performance criteria in
incomplete markets. It studies forward performance processes in the presence
of random endowment and, building on this, introduces the novel concept of
\textit{forward optimized certainty equivalent (forward OCE)}. In parallel, it
develops a new methodological approach to study the emerging forward
stochastic optimization problems through new, interesting on their own right,
\textit{forward-backward stochastic differential equations (FBSDEs)} satisfied
by the optimal control processes of the primal and dual problems.

Random endowments are a central object of study in expected utility
maximization and, furthermore, play an important role in indifference
valuation/certainty equivalent, where they model the payoffs (or liabilities)
to be priced and hedged. The aim herein is twofold: to develop a general
framework to incorporate them within the broader class of forward performance
processes and analyze the related optimal control processes, and to
investigate how they can be used to generalize the widely-used notion of
optimized certainty equivalent.

The motivation for the plan of study herein stems from various shortcomings of
the classical (backward) setting, as highlighted next and further explained
later in the paper. Recalling the standard paradigm in the random endowment
literature, one pre-specifies at initial time $t=0$ a quadruple consisting of
i) the (longest) horizon $[0,T]$ within which random endowments will arrive,
ii) the utility function $U$ at the end of this horizon, iii) the underlying
market model $\mathcal{M}_{[0,T]}$ and iv) the upcoming random endowments. In
other words, these modeling ingredients are chosen statically, once and for
all, at initiation. However, \emph{more flexibility is frequently needed as
new random endowments might arrive at times not known before or arrive at
random times (see \cite{Wang}), the market model might be also revised (see
\cite{Angoshtari2020}) and, furthermore, the risk preferences themselves could
be modified (see \cite{Strub2021})}. Indeed, let us consider the following
simple representative case: for simplicity, it is assumed that there is a
random endowment, given by a random variable $P_{T}$ specified at $t=0$.
However, at some future time $\tau$, with $0<\tau<T$, the utility maximizer
learns that an additional payoff $P_{T_{1}}$ is expected at time $T_{1}<T$.
One now sees that the solution of the problem in $[0,\tau]$ has not considered
this updated information, yielding \textit{a posteriori} time-inconsistent
solutions. The situation becomes even more complex if this new random
endowment actually arrives at time $T_{1}>T$, for the utility $U$ was
pre-defined only at $T$, and not beyond this time. Therefore, a modified
utility maximization problem in $[T,T_{1}]$ needs first to be defined in order
to accommodate the new random endowment. Additional considerations arise if at
an intermediate time, say $\tau^{\prime}<T$, the market model is updated to
$\mathcal{M}_{[\tau^{\prime},T]}^{\prime}$, which will also yield \textit{a
posteriori} time-inconsistent solutions, even if the assumptions on the random
endowments remain the same.

Naturally, these limitations also have undesirable consequences on
indifference prices. Indeed, the expected utility maximization framework does
not allow to price in a time-consistent manner claims arriving at times not
known at initiation, and especially when these new claims mature at instances
beyond the pre-specified horizon, in which case the underlying model is not
even well defined. It, also, fails to price, in a time-consistent manner,
claims when the market model is being dynamically revised, or when the risk
preferences themselves evolve stochastically. These limitations were one of
the main motivations for the third author and M. Musiela to develop the theory
of forward performance criteria in the early 2000s (see, among others,
\cite{MZ0, MZ-Kurtz, MZ1, MZ2, MZ2010}).

Since then, the effort in the forward approach has mainly focused on, from the
one hand, developing a probabilistic characterization of forward performance
processes and, from the other, studying forward indifference prices for the
special class of exponential criteria, see, for example, the stochastic
partial differential equation (SPDE) approach initiated in \cite{MZ2010}, and
also studied in \cite{El_Karoui_2018, El_Karoui_2022, EM2014, El_Karoui_2021}.
When the forward performance process is homothetic, like the exponential case,
a new class of ergodic backward stochastic differential equations (BSDEs) have
been proposed to characterize this class (see \cite{CHLZ2019, LZ2017}). When
the forward performance process is time-monotonic, Widder's theorem has been
applied for their characterization (\cite{Avanesyan_2018, Nadtochiy_2017}). In
addition, the corresponding discrete-time theory has been recently explored
extensively (\cite{Angoshtari, Angoshtari2020, Liang2023, Strub2021}). The
applications of forward performance processes have extended to various
domains, such as relative performance criteria (\cite{Anthropelos2022}),
general semimartingale models (\cite{Bo_2023}), insurance (\cite{Chong_2018}),
duality theory (\cite{Choulli_2007, Z2009}), behavioral finance
(\cite{He_2021}), regime-switching models (\cite{Hu2020}), intertemporal
consumption (\cite{Kallblad_2016}), model uncertainty (\cite{Kallblad_2018}),
mean-field games (\cite{Zariphopoulou_2024}), and maturity-independent risk
measures (\cite{ZZ2010}).

The only studies of the forward performance process with random endowment were
the ones for exponential criteria in the context of exponential forward
indifference prices (see, among others, \cite{LSZ,MZ0,MZ-Kurtz}). The most
general result can be found in Chong et al. \cite{CHLZ2019} who built on the
work of forward entropic risk measures, initially developed in \cite{ZZ2010},
and proposed a BSDE, coupled with an ergodic BSDE, representation of
exponential forward indifference prices. However, how this BSDE/ergodic BSDE
approach may be extended beyond the exponential case remains an open problem.

\subsection{Main contributions}

\textit{FBSDE for forward performance criteria with random endowment}%
\smallskip

Herein, we depart from both the specific class of exponential forward
performance processes and the SPDE approach for general forward performance
processes without random endowment. We take an entirely different approach by
working directly with the optimal policies and the corresponding optimal state
price density processes via the solutions of two FBSDEs, one referred to as
the\textit{ primal FBSDE (\ref{FBSDEsystem})} and the other as the
\textit{dual FBSDE (\ref{DFBSDEsystem})}. These two FBSDEs form a convex dual
relationship akin to the primal and the dual problems of the forward
performance process with random endowment. We demonstrate that both FBSDEs
offer necessary and sufficient conditions for the forward performance process
with random endowment and its convex dual.

We introduce the new methodology, establish necessary and sufficient
optimality conditions via the new FBSDEs and explore equivalences between the
primal and dual problems. \emph{To solve the associated FBSDEs, we employ the
recently developed decoupling field method (\cite{Fthesis, FI2020}), which
enables us to address general forward performance criteria under specific
conditions in a single stochastic factor model. Additionally, the exponential
case allows for the decoupling of the forward and backward equations of the
FBSDE, while the complete markets scenario provides explicit solutions.}

The derivation of the primal and dual FBSDEs draws inspiration from Horst et
al. \cite{HHIRZ}, who considered the FBSDE characterization of utility
maximization with random endowment but only from the primal perspective.
However, our work differs from \cite{HHIRZ} as the static utility functions
therein are being replaced by general forward performance processes. Their
stochastic nature combined with their inherent martingale optimality leads to
simpler and more interpretable terms in the obtained FBSDEs. Furthermore, we
explore the convex dual of the primal FBSDE, presenting a novel approach to
studying the convex duality in the forward setting. The results herein also
provide a new perspective for forward performance processes in the absence of
random endowment, yielding the self-generation property, first studied in
\cite{Z2009} for the exponential case.

For the reader's convenience, we next provide a roadmap of the FBSDE approach
and the underlying equivalences and dual relationships. For completeness, the
roadmap also includes information about the formal derivation of the
corresponding primal and dual backward SPDEs discussed in Section
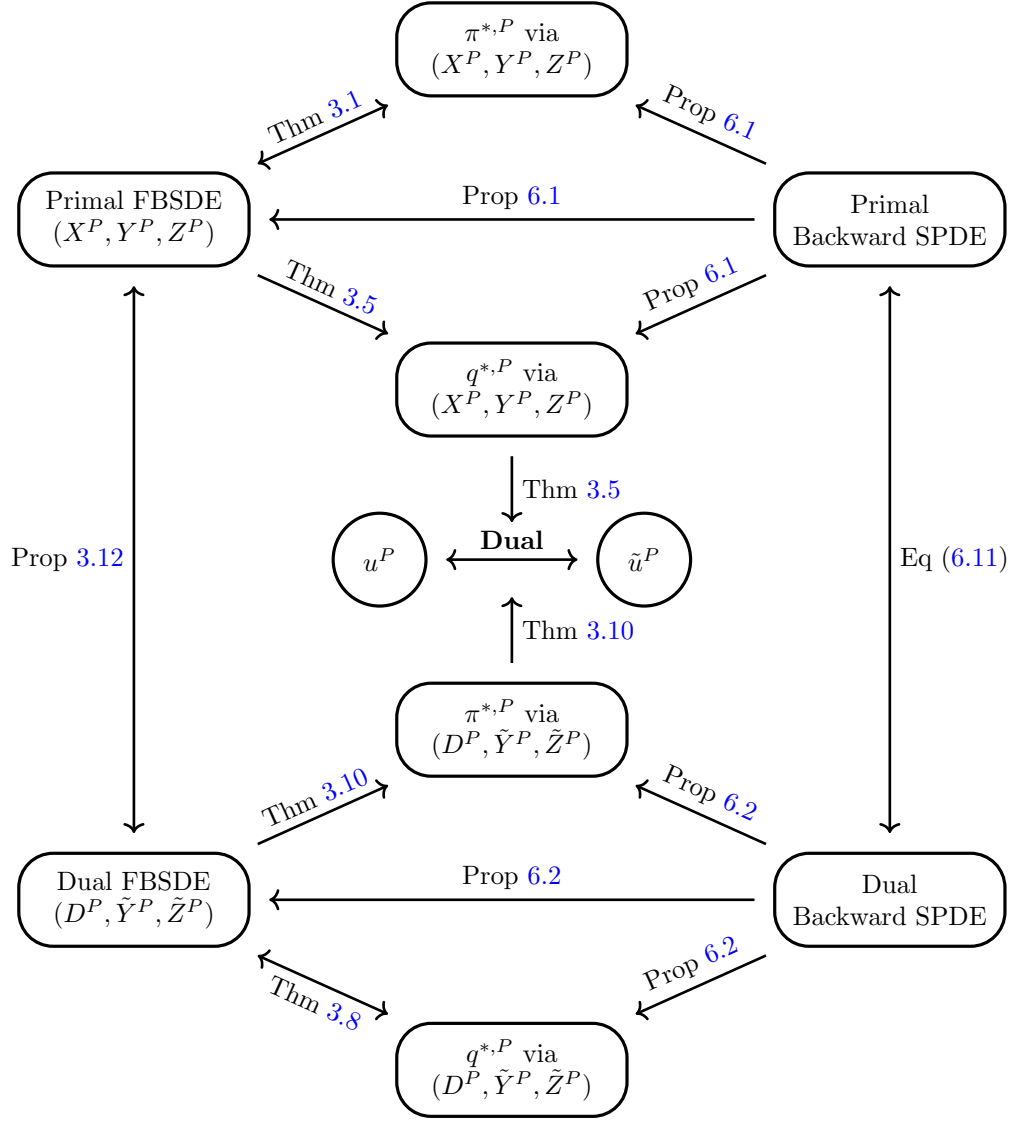
\ref{Relation between FBSDE and SPDE}.

\begin{figure}[ptb]
\begin{center}
\begin{tikzpicture}[others/.style={draw=black, very thick,fill=white, circle,text width=2.5em,
text centered, minimum height=1em},
FBSDE/.style={draw=black,fill=white,very thick, text width=8em,
text centered, minimum height=3.5em,rounded corners=3ex},
SPDE/.style={draw=black,fill=white,very thick, text width=8em,
text centered, minimum height=2.5em,rounded corners=3ex},
relation/.style={draw=black, very thick,fill=white, circle,text width=6em,
text centered, minimum height=2.5em},
connections/.style={<->,draw=black,line width=1pt,shorten <=7pt,shorten >=7pt},
outpt/.style={->,draw=black,line width=1pt,shorten <=7pt,shorten >=7pt},node distance=2cm,auto
]
\node (pistar) at (0,0) [FBSDE]  {$\pi^{\ast,P}$ via $(X^{P},Y^{P},Z^{P})$};
\node (PFBSDE) at (-5,-2.25) [FBSDE]  {Primal FBSDE\\ $(X^{P},Y^{P},Z^{P})$};
\node (DFBSDE) at (-5,-11.25) [FBSDE]  {Dual FBSDE\\ $(D^{P},\tilde{Y}^{P},\tilde{Z}^{P})$};
\node (PSPDE) at (5,-2.25)[FBSDE]  {Primal \\ Backward SPDE};
\node (DSPDE) at (5,-11.25) [FBSDE]  {Dual \\ Backward SPDE};
\node (qstar) at (0,-4.5) [FBSDE]  {$q^{\ast,P}$ via $(X^{P},Y^{P},Z^{P})$};
\node (PP) at (-1.75,-6.75) [others]  {$u^{P}$};
\node (DP) at (1.75,-6.75) [others]  {$\tilde{u}^{P}$};
\node (Dpistar) at (0,-9) [FBSDE]  {$\pi^{\ast,P}$ via $(D^{P},\tilde{Y}^{P},\tilde{Z}^{P})$};
\node (Dqstar) at (0,-13.5) [FBSDE]  {$q^{\ast,P}$ via $(D^{P},\tilde{Y}^{P},\tilde{Z}^{P})$};
\draw[connections] (pistar) to [above,sloped] node {Thm \ref{FBSDE1}}(PFBSDE);
\draw[outpt] (PFBSDE) to [above,sloped] node {Thm \ref{DO1}}(qstar);
\draw[outpt] (PSPDE) to [above,sloped] node {Prop \ref{PB}}(qstar);
\draw[outpt] (PSPDE) to [above,sloped] node {Prop \ref{PB}}(pistar);
\draw[outpt] (PSPDE) to [above] node {Prop \ref{PB}}(PFBSDE);
\draw[connections] (PFBSDE) to [left] node {Prop \ref{R1}}(DFBSDE);
\draw[outpt] (DSPDE) to [above] node {Prop \ref{DPB}}(DFBSDE);
\draw[outpt] (qstar) to node {Thm \ref{DO1}}(0,-6.5);
\draw[outpt] (Dpistar) to [right] node {Thm \ref{DOP1}}(0,-7);
\draw[connections] (PP) to node {\textbf{Dual}}(DP);
\draw[connections] (PSPDE) to node {Eq (\ref{rSPDE})}(DSPDE);
\draw[outpt] (DFBSDE) to [above,sloped] node {Thm \ref{DOP1}}(Dpistar);
\draw[outpt] (DSPDE) to [above,sloped] node {Prop \ref{DPB}}(Dpistar);
\draw[connections] (Dqstar) to [below,sloped] node {Thm \ref{DFBSDE1}}(DFBSDE);
\draw[outpt] (DSPDE) to [above,sloped] node {Prop \ref{DPB}}(Dqstar);
\end{tikzpicture}
\par
\vspace{0.3in} {\small {Glossary of notation. \vspace{-0.08in}
\[
\linespread{0.6}\selectfont
\begin{array}
[c]{lll}%
P &  & \text{random endowment}\\
u^{P} &  & \text{value function of the primal problem}\\
\tilde{u}^{P} &  & \text{value function of the dual problem}\\
\pi^{\ast,P} &  & \text{optimal control of the primal problem}\\
q^{\ast,P} &  & \text{optimal control of the dual problem}%
\end{array}
\]
} }
\end{center}
\caption{The FBSDE approach road map of main results.}%
\label{road_map}%
\end{figure}

\newpage

\noindent \textit{Forward optimized certainty equivalent}\smallskip

Building on the results on forward performance processes with random
endowment, we introduce the novel concept of the \textit{forward optimized
certainty equivalent (forward OCE)}. Its static counterpart, introduced by
Ben-Tal and Teboulle in \cite{BT1986} and \cite{BT2007} (see also (\ref{OCE})
here), is a decision-making criterion based on expected utility theory and
represents the outcome of an optimal fund allocation where an investor can
choose to allocate a portion of the money from the random endowment to spend.
However, due to the static nature of the utility function, extending the
existing OCE notion to a dynamic setting while maintaining time consistency
poses conceptual challenges. Recent studies by Backhoff-Veraguas et al.
\cite{BT2020} and \cite{BRT2022} began with the convex dual representation of
OCE and generalized it to a dynamic version by introducing an additional
variable to ensure time consistency. While \cite{BT2020} and \cite{BRT2022}
generalize the static case of \cite{BT1986} and \cite{BT2007}, the dynamic OCE
therein is directly tied down to the pre-chosen horizon $T$ even though they
are dynamically time-consistent within $[0,T]$. This is a direct consequence
of the fact that the underlying utility function is by nature tied down to
$T$. Similar horizon dependence is also observed in the classical indifference prices.

One then poses the question if there is a way to construct an OCE-type
valuation mechanism that is \textit{horizon invariant}, or \textit{maturity
independent}, since we typically align the horizon with the longest maturity.
We prefer to use the terminology maturity independent to align it with the
studies in \cite{ZZ2010} and \cite{CHLZ2019}.

We stress that we do not expect that the OCE of a given claim would not depend
on the maturity of it, an obviously wrong result. Rather we seek a valuation
mechanism, like the classical conditional expectation, the exponential forward
indifference pricing or the forward entropic risk measure, that does not
depend \textit{per se} on a specific horizon or maturity. This is what we
develop herein within the OCE framework.

\textit{The key idea of the forward OCE notion:} We present a novel
perspective on the static OCE by interpreting it as the convex dual of the
value function for a utility maximization problem with the random endowment
within an auxiliary financial market whose driving noise is orthogonal to the
random endowment. We, then, replace the static utility function with the
forward performance process within this auxiliary financial market and, in
turn, define a forward OCE as the convex dual of the value function for the
forward performance process with the random endowment. This forward OCE
represents the result of an optimal dynamic allocation of funds, where an
investor faces the choice of saving a portion of their current wealth
alongside the random endowment to maximize the forward performance process.
The auxiliary variable introduced in \cite{BT2020} and \cite{BRT2022} has a
natural interpretation as a stochastic deflator, converting nominal values
into real values. The forward OCE mirrors the optimal balance between
maximizing the forward performance process from the saving alongside the
random endowment and reducing spending due to this saving choice.

Unlike the framework proposed in \cite{BT2020} and \cite{BRT2022}, the newly
defined forward OCE is both time consistent and maturity independent.
Furthermore, in the forward OCE framework, we introduce an auxiliary financial
market where the underlying asset can be utilized to partially hedge the risk
associated with the random endowment. In contrast, in the dynamic OCE
framework presented in \cite{BT2020} and \cite{BRT2022}, the driving noise in
this auxiliary market is assumed to be orthogonal to the random endowment,
resulting in null hedging effects. Finally, we apply the new definition and
the related results to the exponential case and we demonstrate that the
forward OCE aligns with the negative of the forward entropic risk measure
proposed in \cite{CHLZ2019}.\medskip

The paper is organized as follows. In section \ref{Preliminaries}, we
introduce the market and recall the notion of forward performance criteria and
their convex dual. In section \ref{FBSDE for FPP with RE}, we present the
primal and dual problems in the presence of random endowment, and the main
results on the new FBSDEs, \emph{while in section \ref{Example}, we solve the
FBSDEs for the complete market case, the exponential forward performance case,
and the more general Markovian case, respectively.} In section
\ref{Application}, we introduce the notion of forward OCE and investigate the
connection of exponential forward OCE with forward entropic risk measure. In
section \ref{Relation between FBSDE and SPDE}, we discuss the backward SPDE
approach and compare it with the FBSDE one proposed herein. We conclude in
section \ref{sec: conclusion}. For the reader's convenience, all proofs are
deferred to the appendix.

\section{Background results\label{Preliminaries}}

We introduce the incomplete market and review fundamental concepts and results
in forward performance criteria, recalling their definition and the ill-posed
SPDE that governs their evolution. We also review their convex dual and the
corresponding SPDE.

\subsection{The incomplete market\label{Market model}}

Let $W=(W^{1},W^{2})$ be a two-dimensional Brownian motion on a complete
probability space $(\Omega,\mathcal{F},\mathbb{P})$ with the Brownian
filtration $\mathbb{F}=(\mathcal{F}_{t})_{t\geq0}$. The market consists of one
riskless asset, taken to be numeraire and offering with zero interest rate,
and a stock whose (discounted) price solves%
\[
\frac{dS_{t}}{S_{t}}=\mu_{t}dt+\sigma_{t}dW_{t}^{1},\quad t\geq0,\quad S
_{0}=S>0.
\]
The coefficients $\mu \in \mathbb{R}$ and $\sigma>0$ are $\mathbb{F}%
$-progressively measurable processes and the market price of risk
$\theta:=\frac{\mu}{\sigma}$ is assumed to be uniformly bounded.

The market is considered incomplete because the dimension of noise is higher
than the number of stocks. Models with such incompleteness can be readily
extended to multidimensional cases\ for both the traded assets and the sources
of imcompleteness (see, for example, \cite{Angoshtari,CHLZ2019,HHIRZ,LZ2017}),
but for simplicity, we keep the dimensionality low.

Let process $\tilde{\pi}$ be the amount invested in the stock. Assume that it
is self-financing and let $\pi:=\tilde{\pi}\sigma.$ Then, the wealth SDE is
given by
\begin{equation}
dX_{t}^{\pi}=\pi_{t}\left(  \theta_{t}dt+dW_{t}^{1}\right)  ,\quad
t\geq0,\quad X_{0}^{\pi}=x\in \mathbb{R}. \label{wealth}%
\end{equation}

For each $T>0$, we denote by $\mathcal{A}_{[t,T]}$ the set of admissible
trading strategies in horizon $[t,T]$, $0\leq t\leq T$, defined as%
\[
\mathcal{A}_{[t,T]}:=\left \{  \pi \in \mathbb{L}_{BMO}^{2}\left[  t,T\right]
\text{, }\pi_{s}\in \mathbb{R}\text{, }t\leq s\leq T\right \}  ,
\]
where%
\begin{align*}
\mathbb{L}_{BMO}^{2}\left[  t,T\right]  :=\Bigg \{  &  \left(  \pi_{s}\right)
_{t\leq s\leq T}:\pi \text{ is }\mathbb{F}\text{-progressively measurable and
}\mathbb{E}\left[  \left.  \int_{\tau}^{T}\left \vert \pi_{r}\right \vert
^{2}dr\right \vert \mathcal{F}_{\tau}\right]  \leq C\text{, }\mathbb{P}%
\text{-a.s.,}\\
&  \text{for some constant }C>0\text{ and all }\mathbb{F}\text{-stopping times
}\tau \text{ valued in }[t,T]\Bigg \}.
\end{align*}
We denote by $\mathcal{A}:=\cup_{t\geq0}\mathcal{A}_{\left[  0,t\right]  }$
the set of admissible trading strategies for all $t\geq0$.


This market model admits multiple state price densities (or equivalent
martingale measures), which we parameterize via a process $q$ on $[t,T]$,
$0\leq t\leq T<\infty$. Specifically, these price densities are represented
as
\[
\emph{M_{s}^{t,q}}=\exp \left(  -\int_{t}^{s}\left(  \theta_{r}dW_{r}^{1}%
+q_{r}dW_{r}^{2}\right)  -\frac{1}{2}\int_{t}^{s}\left(  \left \vert \theta
_{r}\right \vert ^{2}+\left \vert q_{r}\right \vert ^{2}\right)  dr\right)
,\quad0\leq t\leq s\leq T,
\]
where $q\in \mathcal{Q}_{[t,T]}$ with%
\[
\mathcal{Q}_{[t,T]}:=\left \{  q\in \mathbb{L}_{BMO}^{2}\left[  t,T\right]
,q_{s}\in \mathbb{R}\text{, }t\leq s\leq T\right \}  ,\quad0\leq t\leq T.
\]
Therefore, for $q\in \mathcal{Q}_{[t,T]}$, the corresponding state price
density process satisfies the SDE%
\begin{equation}
dM_{s}^{t,q}=-M_{s}^{t,q}\left(  \theta_{s}dW_{s}^{1}+q_{s}dW_{s}^{2}\right)
,\quad0\leq t\leq s\leq T, \label{price_density}%
\end{equation}
with initial condition $M_{t}^{t,q}=1,$ $0\leq t\leq T$. Moreover,
$M_{s}^{t,q}$ is a true $\mathbb{F}$-martingale since $\theta$ is uniformly
bounded and $q\in \mathcal{Q}_{[t,T]}$ (see \cite{KBOOK}). We also have the
following property, which implies that the wealth process is a true
$\mathbb{F}$-martingale under each probability measure transformed by
$q\in \mathcal{Q}_{[t,T]}$.

\begin{lemma}
\label{XZmart}For each $T>0$, $0\leq t\leq T$, $\pi \in \mathcal{A}_{[t,T]}$ and
$q\in \mathcal{Q}_{[t,T]}$, the process
\[
M_{s}^{t,q}\cdot \int_{t}^{s}\pi_{r}\left(  \theta_{r}dr+dW_{r}^{1}\right)
,\quad t\leq s\leq T,
\]
is a true $\mathbb{F}$-martingale on $[t,T]$.
\end{lemma}

\subsection{Forward performance criteria and their convex dual
representation\label{FPP and its convex conjugate}}

We recall the definition of forward performance processes and their SPDE
representation introduced by Musiela and Zariphopoulou in \cite{MZ0, MZ-Kurtz}
and established in \cite{MZ2010}, respectively. We also recall their convex
dual representation introduced by El Karoui and Mrad in \cite{EM2014}.

For $t\geq0$, let $L^{0}(\mathcal{F}_{t})$ denote the space of $\mathcal{F}%
_{t}$-measurable random variables. Similarly,\ let $L^{p}(\mathcal{F}_{t})$,
$p\geq1$, be the space of $L^{p}$-integrable $\mathcal{F}_{t}$-measurable
random variables and $L^{\infty}(\mathcal{F}_{t})$ the space of bounded
$\mathcal{F}_{t}$-measurable random variables.

\begin{definition}
\label{defFPP}Let $\mathbb{D}=[0,\infty)\times \mathbb{R}$. A process $U(t,x)$,
$(t,x)\in \mathbb{D}$, is called a forward performance process if

\begin{enumerate}
\item[(i)] for each $x\in \mathbb{R}$, $U(t,x)$ is $\mathbb{F}$-progressively measurable,

\item[(ii)] for each $t\geq0$, the mapping $x\mapsto U(t,x)$ is strictly
increasing and strictly concave,

\item[(iii)] \emph{for any bounded $\mathbb{F}$-stopping time $\tau$ and
$0\leq t\leq \tau$,
\begin{equation}
U\left(  t,\xi \right)  =\operatorname*{esssup}_{\pi \in \mathcal{A}%
_{[t,\emph{\tau}]}}\mathbb{E}\left[  \left.  U\left(  \emph{\tau},\xi+\int
_{t}^{\emph{\tau}}\pi_{r}\left(  \theta_{r}dr+dW_{r}^{1}\right)  \right)
\right \vert \mathcal{F}_{t}\right]  \text{ for any}\  \xi \in \cap_{p\geq1}%
L^{p}(\mathcal{F}_{t}).\label{FPP}%
\end{equation}
}
\end{enumerate}
\end{definition}

\emph{Here we choose $\xi \in \cap_{p\geq1}L^{p}(\mathcal{F}_{t})$ to ensure
that all its moments are well-defined, a property needed for the related BMO
processes and admissible trading strategies we will be using herein.} We will
frequently refer to (\ref{FPP}) as the \textit{self-generation} property (see
\cite{Z2009}).

In \cite{MZ2010}, the forward performance process $U$ was shown to be
associated with the so-called forward performance SPDE
\begin{equation}
dU\left(  t,x\right)  =\beta \left(  t,x\right)  dt+\alpha^{\top}\left(
t,x\right)  dW_{t},\quad(t,x)\in \mathbb{D}, \label{FSPDE}%
\end{equation}
with the drift $\beta$ given by
\[
\beta \left(  t,x\right)  =\frac{1}{2}\frac{\left \vert U_{x}\left(  t,x\right)
\theta_{t}+\alpha_{x}^{1}\left(  t,x\right)  \right \vert ^{2}}{U_{xx}\left(
t,x\right)  },
\]
and the volatility $\alpha=(\alpha^{1},\alpha^{2})$ being an $\mathbb{R}^{2}%
$-valued $\mathbb{F}$-progressively measurable process. \emph{We recall that,
contrary to the classical framework of maximal expected utility where the
volatility is endogenously determined from the dynamics in the It\^{o}-Ventzel
expansion of the value function process, the forward volatility coefficient
$\alpha(t,x)$ in the forward approach is an independent modeling input.}

The optimal control process is, in turn, expressed in terms of the solution to
the above SPDE. Specifically, if there exists a strong solution to
\begin{equation}
dX_{t}^{\ast}=-\frac{U_{x}\left(  t,X_{t}^{\ast}\right)  \theta_{t}+\alpha
_{x}^{1}\left(  t,X_{t}^{\ast}\right)  }{U_{xx}\left(  t,X_{t}^{\ast}\right)
}\left(  \theta_{t}dt+dW_{t}^{1}\right)  ,\quad t\geq0, \label{SDE}%
\end{equation}
and the feedback control
\begin{equation}
\pi_{t}^{\ast}:=-\frac{U_{x}\left(  t,X_{t}^{\ast}\right)  \theta_{t}%
+\alpha_{x}^{1}\left(  t,X_{t}^{\ast}\right)  }{U_{xx}\left(  t,X_{t}^{\ast
}\right)  },\quad t\geq0, \label{FPPpistar}%
\end{equation}
is in the admissible set $\mathcal{A}$, then $\pi^{\ast}$ is optimal.

Throughout, we assume that $U(t,x)$ solves the SPDE (\ref{FSPDE}) for a
volatility process $\alpha(t,x)$, and SDE
(\ref{SDE}) admits a strong solution $X^{\ast}$. Furthermore, we will assume
the following conditions hold.

\begin{condition}
\label{Assump}

\begin{enumerate}
\item[(i)] \emph{The progressive random field $U(t,x)$, $(t,x)\in
\mathbb{D}$, is a $\mathcal{K}_{loc}^{2,\delta}$-semimartingale for
some $\delta \in(0,1]$ and
$U\in \mathcal{C}^{3}$ (see Appendix
\ref{I-V}).}

\item[(ii)] There exist positive constants $C_{l},C_{u}$ and $C_{\alpha}$ such
that, a.s.,
\[
0<C_{l}\leq-\frac{U_{x}\left(  t,x\right)  }{U_{xx}\left(  t,x\right)  }\leq
C_{u}\quad \text{and}\quad \left \vert \frac{\alpha_{x}^{i}\left(  t,x\right)
}{U_{xx}\left(  t,x\right)  }\right \vert \leq C_{\alpha},\text{ }%
(t,x)\in \mathbb{D},\quad i=1,2.
\]

\end{enumerate}
\end{condition}

\begin{remark}
\emph{In Appendix \ref{I-V}, we provide the differential rule and
It\^{o}-Ventzel formula as well as the spaces of regular random fields.
According to Theorem \ref{DI}, Assumption \ref{Assump} (i) guarantees that the
differential rule for $U$ holds, namely,
\begin{equation}
dU_{x}\left(  t,x\right)  =\beta_{x}\left(  t,x\right)  dt+\alpha_{x}^{\top
}\left(  t,x\right)  dW_{t},\quad \emph{(t,x)\in \mathbb{D}},\label{DU}%
\end{equation}
where%
\begin{align}
\beta_{x}\left(  t,x\right)  =  &  \left(  U_{x}\left(  t,x\right)  \theta
_{t}+\alpha_{x}^{1}\left(  t,x\right)  \right)  \theta_{t}+\frac{\left(
U_{x}\left(  t,x\right)  \theta_{t}+\alpha_{x}^{1}\left(  t,x\right)  \right)
\alpha_{xx}^{1}\left(  t,x\right)  }{U_{xx}\left(  t,x\right)  }\label{DB}\\
&  -\frac{1}{2}\frac{\left \vert U_{x}\left(  t,x\right)  \theta_{t}+\alpha
_{x}^{1}\left(  t,x\right)  \right \vert ^{2}U_{xxx}\left(  t,x\right)
}{\left \vert U_{xx}\left(  t,x\right)  \right \vert ^{2}}.\nonumber
\end{align}
Furthermore, the It\^{o}-Ventzel formula can be applied to both $U(t,X_{t}%
^{\pi})$ and $U_{x}(t,X_{t}^{\pi})$.}
\end{remark}

\begin{remark}
In Assumption \ref{Assump} (ii), $C_{u}$ provides an upper bound for the risk
tolerance $-U_{x}/U_{xx}$, $1/C_{l}$ provides an upper bound for the risk
aversion $-U_{xx}/U_{x}$. \emph{We emphasize that Assumption \ref{Assump} (ii)
is necessary to ensure that the optimal control processes remain within the
admissible set. However, this condition can be further relaxed, for instance,
by allowing the optimal control processes to belong to $\mathbb{L}_{BMO}%
^{2}[t,T]$.}
\end{remark}

Next, we recall the convex dual of the process $U(t,x)$, introduced in
\cite{EM2014}. To this end, for a given forward performance process $U(t,x)$,
its convex conjugate $\tilde{U}(t,z)$ is defined as the Fenchel-Legendre
transform of $-U(t,-x)$, namely,
\begin{equation}
\tilde{U}\left(  t,z\right)  :=\sup_{x\in \mathbb{R}}\left(  U\left(
t,x\right)  -xz\right)  ,\quad(t,z)\in \mathbb{D}^{+}\text{,} \label{dr}%
\end{equation}
where $\mathbb{D}^{+}=[0,\infty)\times \mathbb{R}^{+}$.

One may readily verify that $\tilde{U}(t,z)$ satisfies:

\begin{enumerate}
\item[(i)] For each $t\geq0$, the mapping $z\mapsto \tilde{U}(t,z)$ is strictly
convex and $\tilde{U}\in \mathcal{C}^{3}$.

\item[(ii)] The inverse of the marginal utility $U_{x}$ is the negative of the
marginal of the conjugate utility $\tilde{U}_{z}$:
\begin{equation}%
\begin{array}
[c]{ll}%
-\tilde{U}_{z}\left(  t,U_{x}\left(  t,x\right)  \right)  =x, & (t,x)\in
\mathbb{D}\text{,}\\
U_{x}\left(  t,-\tilde{U}_{z}\left(  t,z\right)  \right)  =z,\text{\  \ } &
(t,z)\in \mathbb{D}^{+}.
\end{array}
\label{dr1}%
\end{equation}

\item[(iii)] The bidual relation
\begin{equation}
U\left(  t,x\right)  =\inf_{z>0}\left(  \tilde{U}\left(  t,z\right)
+xz\right)  ,\quad(t,x)\in \mathbb{D}\text{,} \label{bidr}%
\end{equation}
holds.

\item[(iv)] For $(t,x)\in \mathbb{D}$ and $(t,z)\in \mathbb{D}^{+}$,%
\begin{equation}
U\left(  t,x\right)  =\tilde{U}\left(  t,U_{x}\left(  t,x\right)  \right)
+xU_{x}\left(  t,x\right)  ,\quad \tilde{U}\left(  t,z\right)  =U\left(
t,-\tilde{U}_{z}\left(  t,z\right)  \right)  +z\tilde{U}_{z}\left(
t,z\right)  , \label{dr0}%
\end{equation}%
\begin{equation}
\tilde{U}_{zz}\left(  t,U_{x}\left(  t,x\right)  \right)  =-\frac{1}%
{U_{xx}\left(  t,x\right)  },\quad U_{xx}\left(  t,-\tilde{U}_{z}\left(
t,z\right)  \right)  =-\frac{1}{\tilde{U}_{zz}\left(  t,z\right)  },
\label{dr5}%
\end{equation}
and%
\begin{equation}
\tilde{U}_{zzz}\left(  t,U_{x}\left(  t,x\right)  \right)  =\frac
{U_{xxx}\left(  t,x\right)  }{\left(  U_{xx}\left(  t,x\right)  \right)  ^{3}%
},\quad U_{xxx}\left(  t,-\tilde{U}_{z}\left(  t,z\right)  \right)
=-\frac{\tilde{U}_{zzz}\left(  t,z\right)  }{\left(  \tilde{U}_{zz}\left(
t,z\right)  \right)  ^{3}}. \label{dr5'}%
\end{equation}

\end{enumerate}

As the analysis herein shows, the process $\tilde{U}_{z}(t,z)$ is a key
quantity in establishing the FBSDE representations in Section
\ref{FBSDE for FPP with RE}. It turns out that stronger regularity assumptions
are needed for $\tilde{U}$ and $\tilde{U}_{z}$, introduced below; see
\cite{El_Karoui_2018,EM2014} for similar assumptions.

\begin{condition}
\emph{The progressive random field $\tilde{U}(t,z)$, $(t,z)\in
\mathbb{D}^{+}$, is a $\mathcal{K}_{loc}^{2,\delta}$-semimartingale for some $\delta \in (0,1]$.}
\end{condition}

Recalling Theorem 1.5 in \cite{EM2014}, we obtain the dynamics of $\tilde
{U}_{z}$,%
\begin{equation}
d\tilde{U}_{z}\left(  t,z\right)  =\tilde{\beta}_{z}\left(  t,z\right)
dt+\tilde{\alpha}_{z}^{\top}\left(  t,z\right)  dW_{t},\quad(t,z)\in
\mathbb{D}^{+}, \label{DUSPDE}%
\end{equation}
where%
\begin{align}
\tilde{\beta}_{z}\left(  t,z\right)  =  &  \frac{\beta_{x}\left(  t,-\tilde
{U}_{z}\left(  t,z\right)  \right)  }{U_{xx}\left(  t,-\tilde{U}_{z}\left(
t,z\right)  \right)  }+\frac{1}{2}\frac{U_{xxx}\left(  t,-\tilde{U}_{z}\left(
t,z\right)  \right)  \left \vert \alpha_{x}\left(  t,-\tilde{U}_{z}\left(
t,z\right)  \right)  \right \vert ^{2}}{\left(  U_{xx}\left(  t,-\tilde{U}%
_{z}\left(  t,z\right)  \right)  \right)  ^{3}}\nonumber \\
&  -\frac{\alpha_{x}^{\top}\left(  t,-\tilde{U}_{z}\left(  t,z\right)
\right)  \alpha_{xx}\left(  t,-\tilde{U}_{z}\left(  t,z\right)  \right)
}{\left \vert U_{xx}\left(  t,-\tilde{U}_{z}\left(  t,z\right)  \right)
\right \vert ^{2}}\nonumber \\
=  &  -\tilde{U}_{zz}\left(  t,z\right)  z\left \vert \theta_{t}\right \vert
^{2}+\tilde{\alpha}_{zz}^{1}\left(  t,z\right)  z\theta_{t}+\tilde{\alpha}%
_{z}^{1}\left(  t,z\right)  \theta_{t}-\frac{1}{2}\tilde{U}_{zzz}\left(
t,z\right)  \left \vert z\theta_{t}\right \vert ^{2}\nonumber \\
&  -\frac{1}{2}\frac{\tilde{U}_{zzz}\left(  t,z\right)  |\tilde{\alpha}%
_{z}^{2}\left(  t,z\right)  |^{2}}{|\tilde{U}_{zz}\left(  t,z\right)  |^{2}%
}+\frac{\tilde{\alpha}_{z}^{2}\left(  t,z\right)  \tilde{\alpha}_{zz}%
^{2}\left(  t,z\right)  }{\tilde{U}_{zz}\left(  t,z\right)  }, \label{B2}%
\end{align}
and%
\begin{equation}
\tilde{\alpha}_{z}\left(  t,z\right)  =\frac{\alpha_{x}\left(  t,-\tilde
{U}_{z}\left(  t,z\right)  \right)  }{U_{xx}\left(  t,-\tilde{U}_{z}\left(
t,z\right)  \right)  }=-\tilde{U}_{zz}\left(  t,z\right)  \alpha_{x}\left(
t,-\tilde{U}_{z}\left(  t,z\right)  \right)  . \label{DDA}%
\end{equation}

We conclude by stressing that it is not immediately clear whether the convex
conjugate $\tilde{U}(t,z)$ satisfies a self-generation type property like
(\ref{FPP}) for the primal forward process $U(t,x)$. This question was studied
in \cite{Z2009}. Therein, the author proved the equivalence between the
self-generation properties of $U(t,x)$ and $\tilde{U}(t,z)$ and, in turn,
provided a dual characterization of\ the forward criterion. In this paper, the
FBSDE approach we develop can be also used to establish similar results for
forward processes, see Proposition \ref{SG}.

\section{\emph{Forward performance criteria with random endowment: an FBSDE
approach}\label{FBSDE for FPP with RE}}

We start by choosing an arbitrary time $T>0$ at which we introduce a random
endowment, represented as a random variable $P\in L^{\infty}(\mathcal{F}_{T}%
)$. We then introduce the \textit{primal forward problem with random
endowment}
\begin{equation}
u^{P}\left(  t,\xi;T\right)  =\operatorname*{esssup}_{\pi \in \mathcal{A}%
_{[t,T]}}\mathbb{E}\left[  \left.  U\left(  T,\xi+\int_{t}^{T}\pi_{r}\left(
\theta_{r}dr+dW_{r}^{1}\right)  +P\right)  \right \vert \mathcal{F}_{t}\right]
,\quad0\leq t\leq T\text{, }\emph{\xi \in \cap_{p\geq1}L^{p}(\mathcal{F}_{t})}.
\label{P}%
\end{equation}

The associated \textit{dual forward problem with random endowment} is then
given by
\begin{equation}
\tilde{u}^{P}\left(  t,\eta;T\right)  =\operatorname*{essinf}_{q\in
\mathcal{Q}_{[t,T]}}\mathbb{E}\left[  \left.  \tilde{U}\left(  T,\eta
M_{T}^{t,q}\right)  +\eta M_{T}^{t,q}P\right \vert \mathcal{F}_{t}\right]
,\quad0\leq t\leq T\text{, }\eta \in L^{0,+}\left(  \mathcal{F}_{t}\right)  ,
\label{CP}%
\end{equation}
where $\tilde{U}$ is the convex conjugate of $U$ as (\ref{dr}), and
$L^{0,+}(\mathcal{F}_{t})$ denotes the space of positive $\mathcal{F}_{t}%
$-measurable random variables.

A conventional approach to solving the primal and dual problems (\ref{P}) and
(\ref{CP}) is to first characterize their value functions, which are, in
general, expected to solve related backward SPDEs due to the underlying
non-Markovian model. The corresponding optimal controls for (\ref{P}) and
(\ref{CP}) are then expressed in terms of the solutions to these SPDEs.
However, as we demonstrate in Section \ref{Relation between FBSDE and SPDE},
their derivation is formal and their solvability is far from clear.

In this section, we take a different approach by directly characterizing the
optimal control processes for both the primal and dual problems using an FBSDE
approach. By substituting the optimal controls into (\ref{P}) and (\ref{CP}),
we can in turn obtain the value functions. Our main contributions are the
derivation of two FBSDEs: the primal FBSDE (\ref{FBSDEsystem}) and the dual
FBSDE (\ref{DFBSDEsystem}), both of which provide a means to characterize the
solutions for both the primal and dual problems. These two FBSDEs are closely
connected, forming a convex dual relationship, similar to what (\ref{P}) and
(\ref{CP}) represent.

\subsection{Optimal policy characterization using the primal
FBSDE\label{Primal FBSDE}}

\emph{The first main result is the derivation of both necessary and sufficient
conditions for the optimal control process for the primal problem (\ref{P})
through the FBSDE below, which we will refer as the \textit{primal FBSDE}.}
Specifically, consider the FBSDE, for $0\leq t\leq s\leq T$,
\begin{equation}
\left \{
\begin{array}
[c]{l}%
X_{s}^{P}=\xi-%
{\displaystyle \int_{t}^{s}}
\left(  \dfrac{U_{x}\left(  r,X_{r}^{P}+Y_{r}^{P}\right)  \theta_{r}%
+\alpha_{x}^{1}\left(  r,X_{r}^{P}+Y_{r}^{P}\right)  }{U_{xx}\left(
r,X_{r}^{P}+Y_{r}^{P}\right)  }+Z_{r}^{P,1}\right)  \left(  \theta
_{r}dr+dW_{r}^{1}\right)  ,\\
Y_{s}^{P}=P+%
{\displaystyle \int_{s}^{T}}
\left(  -Z_{r}^{P,1}\theta_{r}+\dfrac{1}{2}\dfrac{U_{xxx}\left(  r,X_{r}%
^{P}+Y_{r}^{P}\right)  }{U_{xx}\left(  r,X_{r}^{P}+Y_{r}^{P}\right)
}\left \vert Z_{r}^{P,2}\right \vert ^{2}+\dfrac{\alpha_{xx}^{2}\left(
r,X_{r}^{P}+Y_{r}^{P}\right)  }{U_{xx}\left(  r,X_{r}^{P}+Y_{r}^{P}\right)
}Z_{r}^{P,2}\right)  dr\\
\text{ \  \  \  \  \  \  \ }-%
{\displaystyle \int_{s}^{T}}
\left(  Z_{r}^{P}\right)  ^{\top}dW_{r},
\end{array}
\right.  \label{FBSDEsystem}%
\end{equation}
with initial-terminal condition $(\xi,P)$ for $\xi \in \cap_{p\geq1}%
L^{p}(\mathcal{F}_{t})$ and $P\in L^{\infty}(\mathcal{F}_{T})$. \emph{Here we
use the superscript $P$ to highlight the dependence of the random endowment.}

\begin{theorem}
\label{FBSDE1}Let $T>0$, $P\in L^{\infty}(\mathcal{F}_{T})$ and $\xi \in
\cap_{p\geq1}L^{p}(\mathcal{F}_{t})$, $0\leq t\leq T$.

\begin{enumerate}
\item[(i)] Assume $\pi^{\ast,P}\in \mathcal{A}_{[t,T]}$ is an optimal control
process of the primal problem (\ref{P}) and denote
\begin{equation}
X_{s}^{\ast,P}:=\xi+\int_{t}^{s}\pi_{r}^{\ast,P}(\theta_{r}dr+dW_{r}%
^{1}),\quad t\leq s\leq T. \label{WS}%
\end{equation}
Furthermore, assume that%
\begin{equation}
\mathbb{E}\mathcal{[}|U_{x}(T,X_{T}^{\ast,P}+P)|^{2}]<\infty, \label{pa}%
\end{equation}
and
\begin{equation}%
\begin{array}
[c]{c}%
\frac{1}{\varepsilon}\left \vert U\left(  T,\xi+%
{\textstyle \int_{t}^{T}}
\left(  \pi_{s}^{\ast,P}+\varepsilon \pi_{s}\right)  \left(  \theta
_{s}ds+dW_{s}^{1}\right)  +P\right)  -U\left(  T,X_{T}^{\ast,P}+P\right)
\right \vert \\
\text{is uniformly integrable in }\varepsilon \in(0,1)\text{ for any bounded
}\pi \in \mathcal{A}_{[t,T]}\text{,}%
\end{array}
\label{pb}%
\end{equation}
Then, there exists a triplet $(X^{P},Y^{P},Z^{P})$ satisfying FBSDE
(\ref{FBSDEsystem}). Furthermore, $\pi^{\ast,P}$ can be represented as
\begin{equation}
\pi_{s}^{\ast,P}=-\frac{U_{x}\left(  s,X_{s}^{P}+Y_{s}^{P}\right)  \theta
_{s}+\alpha_{x}^{1}\left(  s,X_{s}^{P}+Y_{s}^{P}\right)  }{U_{xx}\left(
s,X_{s}^{P}+Y_{s}^{P}\right)  }-Z_{s}^{P,1},\quad t\leq s\leq T\text{,}
\label{Opi}%
\end{equation}
and $X_{s}^{\ast,P}=X_{s}^{P}$, $t\leq s\leq T$.

\item[(ii)] Suppose $(X^{P},Y^{P},Z^{P})$ is a solution to FBSDE
(\ref{FBSDEsystem}) satisfying $Z^{P,i}\in \mathbb{L}_{BMO}^{2}[t,T]$, $i=1,2$.
Then, the control process $\pi^{\ast,P}$ defined by (\ref{Opi}) is optimal for
the primal problem (\ref{P}) and $X_{s}^{\ast,P}=X_{s}^{P}$, $t\leq s\leq T$.
Specifically, $\pi^{\ast,P}\in \mathcal{A}_{[t,T]}$ and
\begin{equation}
u^{P}\left(  t,\xi;T\right)  =\mathbb{E}\left[  \left.  U\left(  T,\xi
+\int_{t}^{T}\pi_{r}^{\ast,P}\left(  \theta_{r}dr+dW_{r}^{1}\right)
+P\right)  \right \vert \mathcal{F}_{t}\right]  =\mathbb{E}\left[  \left.
U\left(  T,X_{T}^{P}+P\right)  \right \vert \mathcal{F}_{t}\right]  .
\label{sss}%
\end{equation}

\end{enumerate}
\end{theorem}

\begin{remark}
\emph{Condition (\ref{pb}) is not needed when $U(t,x)$ satisfies
\[
U_{x}(t,x+y)\leq CU_{x}(t,x)(1+e^{ky}),\quad t\geq0, \quad x,y\in
\mathbb{R}\text{,}%
\]
for some constants $C>0$ and $k\in \mathbb{R}$. Indeed, under this additional
growth condition and following a similar argument as in \cite[Section
3]{HHIRZ}, one can verify that condition (\ref{pa}) implies condition
(\ref{pb}). }
\end{remark}

\begin{lemma}
\label{NM}Suppose $(X^{P},Y^{P},Z^{P})$ is a solution to FBSDE
(\ref{FBSDEsystem}) satisfying $Z^{P,i}\in \mathbb{L}_{BMO}^{2}[t,T]$, $i=1,2$.
Define
\begin{equation}
q_{s}^{\ast,P}:=-\frac{U_{xx}\left(  s,X_{s}^{P}+Y_{s}^{P}\right)  Z_{s}%
^{P,2}+\alpha_{x}^{2}\left(  s,X_{s}^{P}+Y_{s}^{P}\right)  }{U_{x}\left(
s,X_{s}^{P}+Y_{s}^{P}\right)  },\quad0\leq t\leq s\leq T. \label{qstar}%
\end{equation}
Then $q^{\ast,P}\in \mathcal{Q}_{[t,T]}$ and the corresponding state price
density is given by
\[
M_{s}^{t,q^{\ast,P}}=\frac{U_{x}\left(  s,X_{s}^{P}+Y_{s}^{P}\right)  }%
{U_{x}\left(  t,\xi+Y_{t}^{P}\right)  },\quad0\leq t\leq s\leq T.
\]
Furthermore, $M^{t,q^{\ast,P}}$ is a true $\mathbb{F}$-martingale.
\end{lemma}

The proof follows easily by Assumption \ref{Assump} (ii) and the fact that
\begin{equation}
dU_{x}\left(  s,X_{s}^{P}+Y_{s}^{P}\right)  =-U_{x}\left(  s,X_{s}^{P}%
+Y_{s}^{P}\right)  \left(  \theta_{s}dW_{s}^{1}+q_{s}^{\ast,P}dW_{s}%
^{2}\right)  ,\quad0\leq t\leq s\leq T, \label{2402}%
\end{equation}
where we use FBSDE (\ref{FBSDEsystem}) and the dynamics of $U_{x}$ in
(\ref{DU}). Comparing with the SDE for the state price density in
(\ref{price_density}), we have thus identified $M^{t,q^{\ast,P}}$ as a
candidate process for the optimal state price density.

The above lemma is used to verify that FBSDE (\ref{FBSDEsystem}) serves as a
\textit{sufficient} condition for the primal problem (\ref{P}) in Theorem
\ref{FBSDE1} (ii) above, and also for the dual problem (\ref{CP}) in Theorem
\ref{DO1} below. \emph{We also mention that process $q^{\ast,P}$ depends on
the solution of primal FBSDE (\ref{FBSDEsystem}) with initial-terminal
condition $(\xi,P)$, and so is the state price density $M^{t,q^{\ast,P}}$.}

\begin{remark}
\label{PC1}When $P\equiv0$, equation (\ref{sss}) reduces to
\begin{equation}
u^{0}\left(  t,\xi;T\right)  =\operatorname*{esssup}_{\pi \in \mathcal{A}%
_{[t,T]}}\mathbb{E}\left[  \left.  U\left(  T,\xi+\int_{t}^{T}\pi_{r}\left(
\theta_{r}dr+dW_{r}^{1}\right)  \right)  \right \vert \mathcal{F}_{t}\right]
=U\left(  t,\xi \right)  . \label{2201}%
\end{equation}
It then follows directly that the triplet $(X^{0},Y^{0},Z^{0})=(X^{\ast}%
,0,0)$, with $X^{\ast}$ satisfying (\ref{SDE}) with initial condition
$X_{t}^{\ast}=\xi$, solves FBSDE (\ref{FBSDEsystem}). Thus, Theorem
\ref{FBSDE1} (ii) yields that the control process
\[
\pi_{s}^{\ast,0}=-\frac{U_{x}\left(  s,X_{s}^{\ast}\right)  \theta_{s}%
+\alpha_{x}^{1}\left(  s,X_{s}^{\ast}\right)  }{U_{xx}\left(  s,X_{s}^{\ast
}\right)  },\quad0\leq t\leq s\leq T,
\]
is optimal and $X^{\ast,0}=X^{0}=X^{\ast}$, where $X^{\ast,0}$ solves
(\ref{wealth}) with $\pi^{\ast,0}$ above being used. As expected, $\pi
^{\ast,0}$ coincides with the optimal process derived in (\ref{FPPpistar}).
Furthermore, (\ref{2201}) aligns with the self-generation property (\ref{FPP}).
\end{remark}

\subsubsection{Primal FBSDE and the dual problem\label{Duality}}

We demonstrate that the primal FBSDE (\ref{FBSDEsystem}) provides a solution
to the dual problem (\ref{CP}).

\begin{theorem}
\label{DO1}Let $T>0$, $P\in L^{\infty}(\mathcal{F}_{T})$ and $\xi \in
\cap_{p\geq1}L^{p}(\mathcal{F}_{t})$, $0\leq t\leq T$. Suppose $(X^{P}%
,Y^{P},Z^{P})$ is a solution to FBSDE (\ref{FBSDEsystem}) on $[t,T]$ with
initial-terminal condition $(\xi,P)$, satisfying $Z^{P,i}\in \mathbb{L}%
_{BMO}^{2}[t,T]$, $i=1,2$. Let
\begin{equation}
\hat{\eta}:=U_{x}\left(  t,\xi+Y_{t}^{P}\right)  . \label{eta}%
\end{equation}
Then,

\begin{enumerate}
\item[(i)] the control process $q^{\ast,P}$ defined in (\ref{qstar}) is
optimal for the dual problem (\ref{CP}), namely, $q^{\ast,P}\in \mathcal{Q}%
_{[t,T]}$ and
\begin{align*}
\tilde{u}^{P}\left(  t,\hat{\eta};T\right)   &  =\mathbb{E}\left[  \left.
\tilde{U}\left(  T,\hat{\eta}M_{T}^{t,q^{\ast,P}}\right)  +\hat{\eta}%
M_{T}^{t,q^{\ast,P}}P\right \vert \mathcal{F}_{t}\right] \\
&  =\mathbb{E}\left[  \left.  \tilde{U}\left(  T,U_{x}\left(  T,X_{T}%
^{P}+P\right)  \right)  +U_{x}\left(  T,X_{T}^{P}+P\right)  P\right \vert
\mathcal{F}_{t}\right]  ,
\end{align*}

\item[(ii)] the bidual relation
\[
u^{P}\left(  t,\xi;T\right)  =\operatorname*{essinf}_{\eta \in L^{0,+}%
(\mathcal{F}_{t})}\left(  \tilde{u}^{P}\left(  t,\eta;T\right)  +\xi
\eta \right)  =\tilde{u}^{P}\left(  t,\hat{\eta};T\right)  +\xi \hat{\eta}%
\]
holds.
\end{enumerate}
\end{theorem}

\begin{remark}
\label{DO2}\emph{The reverse direction of Theorem \ref{DO1} amounts to the
case of the initial state $\eta$ of the density process being arbitrary, but
the initial wealth $\hat{\xi}$ depending on $\eta$. }More precisely, let
$\eta \in L^{0,+}(\mathcal{F}_{t})$. Then, if $(X^{P},Y^{P},Z^{P})$ is a
solution to the FBSDE (\ref{FBSDEsystem}) on $[t,T]$ with initial-terminal
condition $(\hat{\xi},P)$, satisfying $Z^{P,i}\in \mathbb{L}_{BMO}^{2}[t,T]$
for $i=1,2$, and
\[
\hat{\xi}=-\tilde{U}_{z}\left(  t,\eta \right)  -Y_{t}^{P},
\]
with $\hat{\xi}\in \cap_{p\geq1}L^{p}(\mathcal{F}_{t})$, the control process
$q^{\ast,P}$, defined in (\ref{qstar}) and depending on $\hat{\xi}$, is
optimal for the dual problem $\tilde{u}^{P}(t,\eta;T)$ in (\ref{CP}).
Furthermore, the bidual relation
\[
\tilde{u}^{P}(t,\eta;T)=\operatorname*{esssup}_{\xi \in \cap_{p\geq1}%
L^{p}(\mathcal{F}_{t})}\left(  u^{P}(t,\xi;T)-\xi \eta \right)  =u^{P}%
(t,\hat{\xi};T)-\hat{\xi}\eta
\]
holds. This result follows directly from Theorem \ref{DO1}, using similar arguments.
\end{remark}

Next, using the above primal FBSDE characterization, we re-establish the
self-generation property of $\tilde{U}$, first proven in \cite{Z2009}.

\begin{proposition}
\label{SG}For any $0\leq t\leq T$ and $\eta \in L^{0,+}(\mathcal{F}_{t})$
satisfying $\hat{\xi}:=-\tilde{U}_{z}(t,\eta)\in \cap_{p\geq1}L^{p}%
(\mathcal{F}_{t})$,
\begin{equation}
\tilde{U}\left(  t,\eta \right)  =\operatorname*{essinf}_{q\in \mathcal{Q}%
_{[t,T]}}\mathbb{E}\left[  \left.  \tilde{U}\left(  T,\eta M_{T}^{t,q}\right)
\right \vert \mathcal{F}_{t}\right]  , \label{SG1}%
\end{equation}
and
\begin{equation}
\tilde{U}\left(  t,\eta \right)  =\mathbb{E}\left[  \left.  \tilde{U}\left(
T,\eta M_{T}^{t,q^{\ast}}\right)  \right \vert \mathcal{F}_{t}\right]  ,
\label{SG2}%
\end{equation}
where $q^{\ast}\in \mathcal{Q}_{[t,T]}$ is given by%
\begin{equation}
q_{s}^{\ast}:=-\frac{\alpha_{x}^{2}\left(  s,X_{s}^{\ast}\right)  }%
{U_{x}\left(  s,X_{s}^{\ast}\right)  },\quad t\leq s\leq T, \label{0q}%
\end{equation}
with $X^{\ast}$ satisfying (\ref{SDE}) with initial condition $X_{t}^{\ast
}=\hat{\xi}$.
\end{proposition}

This result follows by choosing $P\equiv0$ in Remark \ref{DO2}. Specifically,
by Remark \ref{PC1}, we know that $(X^{0},Y^{0},Z^{0})=(X^{\ast},0,0)$ solves
FBSDE (\ref{FBSDEsystem}) on $[t,T]$ with initial-terminal condition
$(\hat{\xi},0)$, and $q^{\ast}$ in (\ref{0q}) is bounded according to
Assumption \ref{Assump} (ii). By Remark \ref{DO2}, if $\hat{\xi}\in \cap
_{p\geq1}L^{p}(\mathcal{F}_{t})$, then $q^{\ast}$ is optimal for $\tilde
{u}^{0}(t,\eta;T)$, i.e.,%
\[
\tilde{u}^{0}\left(  t,\eta;T\right)  =\operatorname*{essinf}_{q\in
\mathcal{Q}_{[t,T]}}\mathbb{E}\left[  \left.  \tilde{U}\left(  T,\eta
M_{T}^{t,q}\right)  \right \vert \mathcal{F}_{t}\right]  =\mathbb{E}\left[
\left.  \tilde{U}\left(  T,\eta M_{T}^{t,q^{\ast}}\right)  \right \vert
\mathcal{F}_{t}\right]  .
\]
Moreover,
\[
\tilde{u}^{0}\left(  t,\eta;T\right)  =\operatorname*{esssup}_{\xi \in
\cap_{p\geq1}L^{p}(\mathcal{F}_{t})}\left(  u^{0}\left(  t,\xi;T\right)
-\xi \eta \right)  =u^{0}\left(  t,\hat{\xi};T\right)  -\hat{\xi}\eta.
\]
Note that $u^{0}(t,\xi;T)=U(t,\xi)$ for any $\xi \in \cap_{p\geq1}%
L^{p}(\mathcal{F}_{t})$ according to Remark \ref{PC1}. It follows that%
\[
\tilde{u}^{0}\left(  t,\eta;T\right)  =U\left(  t,\hat{\xi}\right)  -\hat{\xi
}\eta=\tilde{U}\left(  t,\eta \right)  ,
\]
which proves Proposition \ref{SG}.

\subsection{Optimal state price density characterization using the dual FBSDE}

\label{Dual FBSDE}

Having formulated FBSDE (\ref{FBSDEsystem}) for the primal problem, we revert
to the dual problem and derive an analogous FBSDE, given in
(\ref{DFBSDEsystem}) below. With a slight abuse of terminology, we will be
refering to (\ref{DFBSDEsystem}) as the \textit{dual FBSDE}.

To this end, consider the FBSDE, for $0\leq t\leq s\leq T$,%
\begin{equation}
\left \{
\begin{array}
[c]{l}%
D_{s}^{P}=\eta-%
{\displaystyle \int_{t}^{s}}
\left(  D_{r}^{P}\theta_{r}dW_{r}^{1}+\dfrac{\tilde{Z}_{r}^{P,2}+\tilde
{\alpha}_{z}^{2}\left(  r,D_{r}^{P}\right)  }{\tilde{U}_{zz}\left(
r,D_{r}^{P}\right)  }dW_{r}^{2}\right)  ,\\
\tilde{Y}_{s}^{P}=P+%
{\displaystyle \int_{s}^{T}}
\left(  \dfrac{1}{2}\dfrac{\tilde{U}_{zzz}\left(  r,D_{r}^{P}\right)
\left \vert \tilde{Z}_{r}^{P,2}\right \vert ^{2}}{\left \vert \tilde{U}%
_{zz}\left(  r,D_{r}^{P}\right)  \right \vert ^{2}}-\tilde{Z}_{r}^{P,1}%
\theta_{r}-\tilde{Z}_{r}^{P,2}\left(  \dfrac{\tilde{\alpha}_{zz}^{2}\left(
r,D_{r}^{P}\right)  }{\tilde{U}_{zz}\left(  r,D_{r}^{P}\right)  }%
-\dfrac{\tilde{U}_{zzz}\left(  r,D_{r}^{P}\right)  \tilde{\alpha}_{z}%
^{2}\left(  r,D_{r}^{P}\right)  }{\left \vert \tilde{U}_{zz}\left(  r,D_{r}%
^{P}\right)  \right \vert ^{2}}\right)  \right)  dr\\
\text{ \  \  \  \ }-%
{\displaystyle \int_{s}^{T}}
\left(  \tilde{Z}_{r}^{P}\right)  ^{\top}dW_{r},
\end{array}
\right.  \label{DFBSDEsystem}%
\end{equation}
with initial-terminal condition $(\eta,P)$ for $\eta \in L^{0,+}(\mathcal{F}%
_{t})$, $P\in L^{\infty}(\mathcal{F}_{T})$.

We use this FBSDE to characterize the optimal control process of the dual
problem (\ref{CP}). \emph{Specifically, we first derive a necessary condition
for the optimal density process and, in turn, demonstrate that it also yields
a sufficient optimality condition. }

\begin{theorem}
\label{DFBSDE1}Let $T>0$, $P\in L^{\infty}(\mathcal{F}_{T})$ and $\eta \in
L^{0,+}(\mathcal{F}_{t})$, $0\leq t\leq T$.

\begin{enumerate}
\item[(i)] Assume $q^{\ast,P}\in \mathcal{Q}_{[t,T]}$ is an optimal control
process of the dual problem (\ref{CP}).\ Furthermore, assume that
\[
\mathbb{E}\mathcal{[}|(\tilde{U}_{z}(T,\eta M_{T}^{t,q^{\ast,P}})+P)\eta
M_{T}^{t,q^{\ast,P}}|^{2}]<\infty,
\]
and
\[%
\begin{array}
[c]{c}%
\frac{1}{\varepsilon}\left \vert \tilde{U}\left(  T,\eta M_{T}^{t,q^{\ast
,P}+\varepsilon q}\right)  -\tilde{U}\left(  T,\eta M_{T}^{t,q^{\ast,P}%
}\right)  \right \vert \\
\text{is uniformly integrable in }\varepsilon \in(0,1)\text{ for any bounded
}q\in \mathcal{Q}_{[t,T]}.
\end{array}
\]
Then there exists a triplet $(D^{P},\tilde{Y}^{P},\tilde{Z}^{P})$ satisfying
FBSDE (\ref{DFBSDEsystem}). Furthermore, $q^{\ast,P}$ can be represented as
\begin{equation}
q_{s}^{\ast,P}=\frac{\tilde{Z}_{s}^{P,2}+\tilde{\alpha}_{z}^{2}\left(
s,D_{s}^{P}\right)  }{D_{s}^{P}\tilde{U}_{zz}\left(  s,D_{s}^{P}\right)  },\,
\quad t\leq s\leq T, \label{Oq}%
\end{equation}
and $\eta M_{s}^{t,q^{\ast,P}}=D_{s}^{P}$, $t\leq s\leq T.$

\item[(ii)] Suppose $(D^{P},\tilde{Y}^{P},\tilde{Z}^{P})$ is a solution to
FBSDE (\ref{DFBSDEsystem}) \emph{satisfying $D^{P}>0$} and $\tilde{Z}^{P,i}%
\in \mathbb{L}_{BMO}^{2}[t,T],\,i=1,2$. Then, the control process $q^{\ast,P}$
defined by (\ref{Oq}) is optimal for the dual problem (\ref{CP}) and $\eta
M_{s}^{t,q^{\ast,P}}=D_{s}^{P}$, $t\leq s\leq T$. Specifically, $q^{\ast,P}%
\in \mathcal{Q}_{[t,T]}$ and
\[
\tilde{u}^{P}\left(  t,\eta;T\right)  =\mathbb{E}\left[  \left.  \tilde
{U}\left(  T,\eta M_{T}^{t,q^{\ast,P}}\right)  +\eta M_{T}^{t,q^{\ast,P}%
}P\right \vert \mathcal{F}_{t}\right]  =\mathbb{E}\left[  \left.  \tilde
{U}\left(  T,D_{T}^{P}\right)  +D_{T}^{P}P\right \vert \mathcal{F}_{t}\right]
.
\]

\end{enumerate}
\end{theorem}

\begin{lemma}
\label{UZ}Suppose $(D^{P},\tilde{Y}^{P},\tilde{Z}^{P})$ is a solution to FBSDE
(\ref{DFBSDEsystem}) \emph{satisfying $D^{P}>0$ a.s.} and $\tilde{Z}^{P,i}%
\in \mathbb{L}_{BMO}^{2}[t,T],\,i=1,2$. Define%
\begin{equation}
\pi_{s}^{\ast,P}:=\tilde{U}_{zz}\left(  s,D_{s}^{P}\right)  D_{s}^{P}%
\theta_{s}-\tilde{\alpha}_{z}^{1}\left(  s,D_{s}^{P}\right)  -\tilde{Z}%
_{s}^{P,1},\quad0\leq t\leq s\leq T. \label{Dpi*}%
\end{equation}
Then $\pi^{\ast,P}\in \mathcal{A}_{[t,T]}$ and%
\begin{equation}
d\left(  \tilde{U}_{z}\left(  s,D_{s}^{P}\right)  +\tilde{Y}_{s}^{P}\right)
=-\pi_{s}^{\ast,P}\left(  \theta_{s}ds+dW_{s}^{1}\right)  ,\quad0\leq t\leq
s\leq T. \label{UZpi}%
\end{equation}

\end{lemma}

The proof follows easily from Assumption \ref{Assump} (ii) and by combining
FBSDE (\ref{DFBSDEsystem}) with the dynamics of $\tilde{U}_{z}$ in
(\ref{DUSPDE}). From (\ref{UZpi}), $\pi^{\ast,P}$ can be identified as a
candidate process for the optimal policy of the primal problem.

This lemma is used to verify that FBSDE (\ref{DFBSDEsystem}) serves as a
sufficient condition for the dual problem (\ref{CP}) in Theorem \ref{DFBSDE1}
(ii) above and also for the primal problem (\ref{P}) in Theorem \ref{DOP1}
below. \emph{We note that, as expected, process $\pi^{\ast,P}$ depends on the
solution of the dual FBSDE (\ref{DFBSDEsystem}) with initial-terminal
condition $(\eta,P)$.}

\subsubsection{Dual FBSDE and the primal problem}

Corresponding to the results in Section \ref{Duality}, the dual FBSDE
(\ref{DFBSDEsystem}) can be also used to characterize the solution of the
primal problem (\ref{P}).

\begin{theorem}
\label{DOP1}Let $T>0,P\in L^{\infty}(\mathcal{F}_{T})$ and $\eta \in
L^{0,+}(\mathcal{F}_{t})$, $0\leq t\leq T$. Suppose $(D^{P},\tilde{Y}%
^{P},\tilde{Z}^{P})$ is a solution to FBSDE (\ref{DFBSDEsystem}) on $[t,T]$
with initial-terminal condition $(\eta,P)$ \emph{satisfying }$\emph{D^{P}>0}$
and $\tilde{Z}^{P,i}\in \mathbb{L}_{BMO}^{2}[t,T]$, $i=1,2$. Define%
\begin{equation}
\hat{\xi}:=-\tilde{U}_{z}\left(  t,\eta \right)  -\tilde{Y}_{t}^{P}.
\label{Dxi*}%
\end{equation}
Then, if $\hat{\xi}\in \cap_{p\geq1}L^{p}(\mathcal{F}_{t})$,

\begin{enumerate}
\item[(i)] the control process $\pi^{\ast,P}$ defined in (\ref{Dpi*}) is
optimal for the primal problem (\ref{P}), namely, $\pi^{\ast,P}\in
\mathcal{A}_{[t,T]}$ and%
\begin{align*}
u^{P}\left(  t,\hat{\xi};T\right)   &  =\mathbb{E}\left[  \left.  U\left(
T,\hat{\xi}+\int_{t}^{T}\pi_{r}^{\ast,P}\left(  \theta_{r}dr+dW_{r}%
^{1}\right)  +P\right)  \right \vert \mathcal{F}_{t}\right] \\
&  =\mathbb{E}\left[  \left.  U\left(  T,-\tilde{U}_{z}\left(  T,D_{T}%
^{P}\right)  \right)  \right \vert \mathcal{F}_{t}\right]  ,
\end{align*}

\item[(ii)] the bidual relation
\[
\tilde{u}^{P}\left(  t,\eta;T\right)  =\operatorname*{esssup}_{\xi \in
\cap_{p\geq1}L^{p}(\mathcal{F}_{t})}\left(  u^{P}\left(  t,\xi;T\right)
-\xi \eta \right)  =u^{P}\left(  t,\hat{\xi};T\right)  -\hat{\xi}\eta
\]
holds.
\end{enumerate}
\end{theorem}

\begin{remark}
\label{DOP2}Similarly to Remark \ref{DO2}, the reverse direction of Theorem
\ref{DOP1} amounts to the case of the initial wealth $\xi$ being arbitrary,
but the initial state $\hat{\eta}$ specified as a particular class of initial
states for the density process, depending on $\xi$. More precisely, for
$\xi \in \cap_{p\geq1}L^{p}(\mathcal{F}_{t})$, suppose $(D^{P},\tilde{Y}%
^{P},\tilde{Z}^{P})$ is a solution to the FBSDE (\ref{DFBSDEsystem}) on
$[t,T]$ with the initial-terminal condition $(\hat{\eta},P)$, satisfying
$D^{P}>0$, $\tilde{Z}^{P,i}\in \mathbb{L}_{BMO}^{2}[t,T]$ for $i=1,2$, and
\[
\hat{\eta}=U_{x}\left(  t,\xi+\tilde{Y}_{t}^{P}\right)  \in L^{0,+}%
(\mathcal{F}_{t}).
\]
Then, the control process $\pi^{\ast,P}$, defined in (\ref{Dpi*}) and
depending on $\hat{\eta}$, is optimal for the primal problem $u^{P}(t,\xi;T)$
in (\ref{P}). Furthermore, the bidual relation
\[
u^{P}(t,\xi;T)=\operatorname*{essinf}_{\eta \in L^{0,+}(\mathcal{F}_{t}%
)}\left(  \tilde{u}^{P}(t,\eta;T)+\xi \eta \right)  =\tilde{u}^{P}(t,\hat{\eta
};T)+\xi \hat{\eta}%
\]
holds.
\end{remark}

\subsection{Relations between the primal and dual FBSDEs}

We show that the primal FBSDE (\ref{FBSDEsystem}) and the dual FBSDE
(\ref{DFBSDEsystem}) form a convex dual pair. Their relationship is a direct
analogue to the convex duality between the primal problem (\ref{P}) and the
dual problem (\ref{CP}).

\begin{proposition}
\label{R1}Let $T>0$ and $P\in L^{\infty}(\mathcal{F}_{T})$.

\begin{enumerate}
\item[(i)] For $\eta \in L^{0,+}(\mathcal{F}_{t})$, $0\leq t\leq T$, suppose
$(D^{P},\tilde{Y}^{P},\tilde{Z}^{P})$ is a solution to FBSDE
(\ref{DFBSDEsystem}) on $[t,T]$ with initial-terminal condition $(\eta,P)$,
satisfying\ $D^{P}>0$. Let $\hat{\xi}$ be defined in (\ref{Dxi*}). Then, if
$\hat{\xi}\in \cap_{p\geq1}L^{p}(\mathcal{F}_{t})$, the triplet $(X^{P}%
,Y^{P},Z^{P})$ given by
\[
X_{s}^{P}:=-\tilde{U}_{z}\left(  s,D_{s}^{P}\right)  -\tilde{Y}_{s}^{P},\quad
Y_{s}^{P}:=\tilde{Y}_{s}^{P},\quad Z_{s}^{P}:=\tilde{Z}_{s}^{P},\quad t\leq
s\leq T,
\]
is a solution to FBSDE (\ref{FBSDEsystem}) on $[t,T]$ with initial-terminal
condition $(\hat{\xi},P)$.

\item[(ii)] For $\xi \in \cap_{p\geq1}L^{p}(\mathcal{F}_{t})$, $0\leq t\leq T$,
suppose $(X^{P},Y^{P},Z^{P})$ is a solution to FBSDE (\ref{FBSDEsystem}) on
$[t,T]$ with initial-terminal condition $(\xi,P)$. Then, for $\hat{\eta}$
defined in (\ref{eta}), the triplet $(D^{P},\tilde{Y}^{P},\tilde{Z}^{P})$
given by
\[
D_{s}^{P}:=U_{x}\left(  s,X_{s}^{P}+Y_{s}^{P}\right)  ,\quad \tilde{Y}_{s}%
^{P}:=Y_{s}^{P},\quad \tilde{Z}_{s}^{P}:=Z_{s}^{P},\quad t\leq s\leq T,
\]
is a solution to FBSDE (\ref{DFBSDEsystem}) on $[t,T]$ with initial-terminal
condition $(\hat{\eta},P)$.
\end{enumerate}
\end{proposition}

The proof follows easily by (\ref{2402}) and\ (\ref{UZpi}) using the duality
relations (\ref{dr1}), (\ref{dr5}), (\ref{dr5'}) and (\ref{DDA}).

\begin{remark}
\emph{When both the primal and dual FBSDEs both admit a unique solution, there
is a bijection between their solutions. We refer the reader to Section
\ref{Example} for some popular cases of solvable FBSDEs, for which existence
and uniqueness hold.}
\end{remark}

\subsection{Maturity independence of the value functions in forward
performance criteria with random endowment\label{Maturity Independence}}

We set the ground for the upcoming notion of forward OCE by pointing out the
fundamental maturity-independence property of the value functions $u^{P}%
(t,\xi;T)$ and $\tilde{u}^{P}(t,\eta;T)$. Recall the primal problem (\ref{P}),
rewritten below for convenience, in a slightly different form and with some
abuse of notation, $u^{P}(t,\xi;T)=u^{P_{T}}(t,\xi;T)$, where, for $0\leq
t\leq T\leq T^{\prime}$,
\begin{equation}
u^{P_{T}}\left(  t,\xi;T^{\prime}\right)  =\operatorname*{esssup}_{\pi
\in \mathcal{A}_{[t,T^{\prime}]}}\mathbb{E}\left[  \left.  U\left(  T^{\prime
},\xi+\int_{t}^{T^{\prime}}\pi_{r}\left(  \theta_{r}dr+dW_{r}^{1}\right)
+P_{T}\right)  \right \vert \mathcal{F}_{t}\right]  . \label{PP}%
\end{equation}
We essentially parameterize the value function by the time $T$ that the random
endowment $P_{T}$ arrives as well as by $T^{\prime}$ at which we set the
forward performance criterion. Problem (\ref{PP}) can be then thought as a
classical expected utility maximization problem in $[t,T^{\prime}]$ with
terminal random utility $U(T^{\prime},x)$ and random endowment $P_{T}$
received at time $T$. We now claim that $u^{P_{T}}\left(  t,\xi;T^{\prime
}\right)  $ is independent of $T^{\prime}$ in the sense that the
horizon-invariance/maturity-independence property
\[
{u}^{P_{T}}(t,\xi;T)={u}^{P_{T}}(t,\xi;T^{\prime})
\]
holds. Note that the above property cannot hold for any $T^{\prime}>T$ in the
classical setting due to the pre-chosen, arbitrary but fixed, terminal horizon
at which the (static) utility is allocated. \emph{Under the self-generation
properties of $\tilde{U}$, we have the following general result.}

\begin{proposition}
\label{MI}\emph{Let $0<T\leq T^{\prime}$, $P_{T}\in L^{\infty}(\mathcal{F}%
_{T})$ and $0\leq t\leq \emph{T}$. Then, for $\eta \in L^{0,+}(\mathcal{F}_{t}%
)$,
\begin{equation}
\tilde{u}^{P_{T}}\left(  t,\eta;T\right)  =\tilde{u}^{P_{T}}\left(
t,\eta;T^{\prime}\right)  .\label{MI1}%
\end{equation}
If, furthermore, FBSDE (\ref{FBSDEsystem}) admits a solution $(X^{P_{T}%
},Y^{P_{T}},Z^{P_{T}})$ on $[t,T]$ with initial-terminal condition $(\xi
,P_{T})$, $\xi \in \cap_{p\geq1}L^{p}(\mathcal{F}_{t})$, satisfying $Z^{P_{T}%
,i}\in \mathbb{L}_{BMO}^{2}[t,T]$, $i=1,2$, then
\begin{equation}
u^{P_{T}}\left(  t,\xi;T\right)  =u^{P_{T}}\left(  t,\xi;T^{\prime}\right)
.\label{MI2}%
\end{equation}
}
\end{proposition}

Reverting to $u^{P}\left(  t,\xi;T\right)  $, i.e., when $T=T^{\prime}$, we
see that the dependence on $T$ is generated exclusively by the fact that this
is the arrival time of the random endowment. This then motives us to introduce
the following notion of \textit{endowment maturity} which will be useful in
building a universal framework across all times. Consider random endowments in
the general space $L$ defined by%
\[
L:=\cup_{T\geq0}L^{\infty}(\mathcal{F}_{T}).
\]
For any random endowment $P\in L$, we define its \textit{maturity} by%
\begin{equation}
T_{P}:=\inf \left \{  T\geq0:P\text{ is }\mathcal{F}_{T}\text{-measurable}%
\right \}  . \label{maturity}%
\end{equation}
In turn, the maturity independent value functions $u^{P}$ and $\tilde{u}^{P}$
can be defined as%
\begin{align}
u^{P}\left(  t,\xi \right)   &  :=u^{P}\left(  t,\xi;T_{P}\right) \nonumber \\
&  =\operatorname*{esssup}_{\pi \in \mathcal{A}_{[t,T_{P}]}}\mathbb{E}\left[
\left.  U\left(  T_{P},\xi+\int_{t}^{T_{P}}\pi_{r}\left(  \theta_{r}%
dr+dW_{r}^{1}\right)  +P\right)  \right \vert \mathcal{F}_{t}\right]
,\quad0\leq t\leq T_{P}, \, \, \xi \in \cap_{p\geq1}L^{p}(\mathcal{F}%
_{t})\text{,} \label{P1}%
\end{align}
and
\begin{align}
\tilde{u}^{P}\left(  t,\eta \right)  :=  &  \  \tilde{u}^{P}\left(  t,\eta
;T_{P}\right) \nonumber \\
=  &  \operatorname*{essinf}_{q\in \mathcal{Q}_{[t,T_{P}]}}\mathbb{E}\left[
\left.  \tilde{U}\left(  T_{P},\eta M_{T_{P}}^{t,q}\right)  +\eta M_{T_{P}%
}^{t,q}P\right \vert \mathcal{F}_{t}\right]  ,\quad0\leq t\leq T_{P}, \, \,
\eta \in L^{0,+}(\mathcal{F}_{t}). \label{P2}%
\end{align}

The above property essentially allows us to define the value functions
$u^{P}(t,\xi)$ and $\tilde{u}^{P}(t,\eta)$ no matter what the maturity of the
random endowment is. We stress that this does \textit{not} imply that
$u^{P}(t,\xi)$ and $\tilde{u}^{P}(t,\eta)$ do not depend on $T_{P}$, the time
when the endowment arrives, an obviously wrong statement. Rather, it expresses
how problems (\ref{P1}) and (\ref{P2}) can be well-defined for \textit{all}
times and \textit{all} endowments using the flexible forward performance
framework. Once more, note that in the classical framework this cannot be done
as the entire optimization problem is tied down to the a priori chosen
terminal horizon. We comment further on this in Section \ref{sec: conclusion}.

This \textquotedblleft maturity-independent" construction was first developed
in \cite{ZZ2010} and, later in \cite{CHLZ2019}\ for forward entropic risk
measures. Herein, it will play a fundamental role in the new notion of forward
optimized certainty equivalent we introduce and develop in Section
\ref{Application}.

\section{\emph{Solving the primal and dual FBSDEs with random endowment}%
\label{Example}}

We consider representative examples which give rise to solvable primal and
dual FBSDEs.

\subsection{Complete market\label{Complete Market}}

\emph{We start with the complete market case, where both the primal FBSDE
(\ref{FBSDEsystem}) and the dual FBSDE (\ref{DFBSDEsystem}) simplify
considerably. In particular, the dual FBSDE (\ref{DFBSDEsystem}) admits a
closed-form solution.} To this end, assume that the Brownian motion $W$ is
one-dimensional and that $\mathbb{F}=(\mathcal{F}_{t})_{t\geq0}$ is generated
by $W$. Hence, the controlled wealth state equation (\ref{wealth}) becomes
\[
dX_{t}^{\pi}=\pi_{t}\left(  \theta_{t}dt+dW_{t}\right)  ,\quad t\geq0,\text{
}X_{0}^{\pi}=x\in \mathbb{R}\text{.}%
\]
Due to market completeness, there exists a \textit{unique} state price density
process given by
\[
\mathcal{E}\left(  -%
{\textstyle \int}
\theta dW\right)  _{t}=\exp \left(  -\tfrac{1}{2}%
{\textstyle \int_{0}^{t}}
\left \vert \theta_{r}\right \vert ^{2}dr-%
{\textstyle \int_{0}^{t}}
\theta_{r}dW_{r}\right)  ,\quad t\geq0\text{,}%
\]
where we recall that the market price of risk $\theta$ is a bounded
$\mathbb{F}$-progressively measurable process.

For the primal FBSDE (\ref{FBSDEsystem}), since $W^{2}\equiv0$, the solution
component $Z^{P,2}$ vanishes and thus (\ref{FBSDEsystem}) reduces to a
classical BSDE, namely, for $0\leq t\leq s\leq T$,%
\begin{equation}
\left \{
\begin{array}
[c]{l}%
X_{s}^{P}=\xi-%
{\displaystyle \int_{t}^{s}}
\left(  \dfrac{U_{x}\left(  r,X_{r}^{P}+Y_{r}^{P}\right)  \theta_{r}%
+\alpha_{x}\left(  r,X_{r}^{P}+Y_{r}^{P}\right)  }{U_{xx}\left(  r,X_{r}%
^{P}+Y_{r}^{P}\right)  }+Z_{r}^{P}\right)  \left(  \theta_{r}dr+dW_{r}\right)
,\\
Y_{s}^{P}=P-%
{\displaystyle \int_{s}^{T}}
Z_{r}^{P}\theta_{r}dr-%
{\displaystyle \int_{s}^{T}}
Z_{r}^{P}dW_{r}.
\end{array}
\right.  \label{comFBSDE}%
\end{equation}
We easily see that it admits a unique solution $(Y^{P},Z^{P})$ and that
$Y^{P}$ is represented, for $t\leq s\leq T$, as
\begin{equation}
Y_{s}^{P}=\mathbb{E}\left[  \left.  \mathcal{E}\left(  -%
{\textstyle \int}
\theta dW\right)  _{T}^{s}P\right \vert \mathcal{F}_{s}\right]  =\mathbb{E}%
^{\mathbb{Q}^{0}}\left[  P|\mathcal{F}_{s}\right]  \quad \text{with}%
\quad \mathcal{E}\left(  -%
{\textstyle \int}
\theta dW\right)  _{T}^{s}:=\frac{\mathcal{E}(-\int \theta dW)_{T}}%
{\mathcal{E}(-\int \theta dW)_{s}}, \label{EY}%
\end{equation}
where $\mathbb{Q}^{0}$ is the unique equivalent martingale measure given by%
\[
\left.  \frac{d\mathbb{Q}^{0}}{d\mathbb{P}}\right \vert _{\mathcal{F}_{T}%
}=\mathcal{E}\left(  -%
{\textstyle \int}
\theta dW\right)  _{T}.
\]
We also have that $Z^{P}\in \mathbb{L}_{BMO}^{2}[t,T]$, as it follows by a
priori estimates for (\ref{comFBSDE}) (see, for example, \cite[Section
7.2]{ZBOOK}). Hence, if $X^{P}$ is the solution to the SDE in (\ref{comFBSDE})
with $Y^{P}$ in (\ref{EY}) and $Z_{s}^{P}=\frac{d \langle Y^{P},W \rangle_{s}%
}{ds}$, $t\leq s\leq T$, we obtain, by Theorem \ref{FBSDE1} (ii), that the
control process
\[
\pi_{s}^{\ast,P}:=-\dfrac{U_{x}\left(  s,X_{s}^{P}+Y_{s}^{P}\right)
\theta_{s}+\alpha_{x}\left(  s,X_{s}^{P}+Y_{s}^{P}\right)  }{U_{xx}\left(
s,X_{s}^{P}+Y_{s}^{P}\right)  }-Z_{s}^{P},\quad t\leq s\leq T,
\]
is optimal for the primal problem (\ref{P}). Furthermore, by Remark \ref{DO2},
we have that $q^{\ast,P}\equiv0$ is the optimal control process for the dual
problem (\ref{CP}). \emph{Note that although the forward equation and the
backward equation in the primal FBSDE (\ref{FBSDEsystem}) become
\textit{decoupled}, as shown in (\ref{comFBSDE}), with $(Y^{P},Z^{P})$ given
explicitly, the solvability of the SDE in (\ref{comFBSDE}) is not
straightforward, This is where the dual formulation proves helpful.}

Since the market is complete, the dual problem is, as expected, much simpler.
Indeed, for the dual FBSDE (\ref{DFBSDEsystem}), the solution component
$\tilde{Z}^{P,2}$ vanishes, and thus (\ref{DFBSDEsystem}) reduces, for $0\leq
t\leq s\leq T$, to
\begin{equation}
\left \{
\begin{array}
[c]{l}%
D_{s}=\eta-%
{\displaystyle \int_{t}^{s}}
D_{r}\theta_{r}dW_{r},\\
\tilde{Y}_{s}^{P}=P-%
{\displaystyle \int_{s}^{T}}
\tilde{Z}_{r}^{P}\theta_{r}dr-%
{\displaystyle \int_{s}^{T}}
\tilde{Z}_{r}^{P}dW_{r}.
\end{array}
\right.  \label{3001}%
\end{equation}
Note here $D$ is independent of $P$. It is then straightforward to check that
there exists a unique solution given by
\[
D_{s}=\eta \mathcal{E}\left(  -%
{\textstyle \int}
\theta dW\right)  _{s}^{t},\text{ \ }\tilde{Y}_{s}^{P}=Y_{s}^{P}%
=\mathbb{E}^{\mathbb{Q}^{0}}\left[  P|\mathcal{F}_{s}\right]  ,\text{ }\quad
t\leq s\leq T,
\]
and $\tilde{Z}^{P}=Z^{P}\in \mathbb{L}_{BMO}^{2}[t,T]$. Therefore, Theorem
\ref{DFBSDE1} (ii) yields that $q^{\ast,P}\equiv0$ is the optimal control
process for the dual problem (\ref{CP}).

We now solve the SDE in (\ref{comFBSDE}) using the dual formulation. For any
$\xi \in \cap_{p\geq1}L^{p}(\mathcal{F}_{t})$, Remark \ref{DOP2} yields that the
control process
\[
\pi_{s}^{\ast,P}:=\tilde{U}_{zz}\left(  s,\hat{\eta}\mathcal{E}\left(  -%
{\textstyle \int}
\theta dW\right)  _{s}^{t}\right)  \hat{\eta}\mathcal{E}\left(  -%
{\textstyle \int}
\theta dW\right)  _{s}^{t}\theta_{s}-\tilde{\alpha}_{z}\left(  s,\hat{\eta
}\mathcal{E}\left(  -%
{\textstyle \int}
\theta dW\right)  _{s}^{t}\right)  -\tilde{Z}_{s}^{P},\quad0\leq t\leq s\leq
T,
\]
with $\hat{\eta}:=U_{x}(t,\xi+\mathbb{E}^{\mathbb{Q}^{0}}[P|\mathcal{F}_{t}%
])$, is optimal for the primal problem (\ref{P}). Furthermore, we deduce,
using Proposition \ref{R1}, that the optimal wealth $X^{\ast,P}$ is given by
\[
X_{s}^{\ast,P}=-\tilde{U}_{z}\left(  s,U_{x}\left(  t,\xi+\mathbb{E}%
^{\mathbb{Q}^{0}}\left[  P|\mathcal{F}_{t}\right]  \right)  \mathcal{E}\left(
-%
{\textstyle \int}
\theta dW\right)  _{s}^{t}\right)  -\mathbb{E}^{\mathbb{Q}^{0}}\left[
P|\mathcal{F}_{s}\right]  ,\quad0\leq t\leq s\leq T,
\]
solves the SDE in (\ref{comFBSDE}) with initial condition $X_{t}^{\ast,P}=\xi$.

\emph{Finally, we derive an explicit expression for the value function
$\tilde{u}^{P}$ of the dual problem. Recall that, $\tilde{U}$ satisfies the
self-generation property (\ref{SG1})-(\ref{SG2}), namely, $\tilde{U}%
(t,\eta)=\mathbb{E}[\tilde{U}(T,D_{T})|\mathcal{F}_{t}]$. Furthermore, we can
easily check by (\ref{3001}) that $D\tilde{Y}^{P}$ is a true $\mathbb{F}%
$-martingale. Thus, we obtain an explicit form of the value function of the
dual problem by Theorem \ref{DFBSDE1} (ii):%
\begin{equation}
\tilde{u}^{P}\left(  t,\eta;T\right)  =\mathbb{E}\left[  \tilde{U}\left(
T,D_{T}\right)  +D_{T}P|\mathcal{F}_{t}\right]  =\tilde{U}\left(
t,\eta \right)  +\eta \tilde{Y}_{t}^{P}=\tilde{U}\left(  t,\eta \right)
+\eta \mathbb{E}^{\mathbb{Q}^{0}}\left[  P|\mathcal{F}_{t}\right]
\text{.}\label{3002}%
\end{equation}
}

\subsection{Stochastic factor model and exponential forward
criteria\label{expfpp}}

We analyze the popular exponential forward performance criteria in an
incomplete market driven by a single stochastic factor. In the absence of
random endowment, this class was studied in \cite{LZ2017} where an ergodic
BSDE was developed. A by-product of the results below is recovery of the
results in \cite{LZ2017} and \cite{CHLZ2019} as special cases.

To this end, we recall the single stochastic factor model, as in \cite{LZ2017}
and \cite{CHLZ2019}, which assumes that the discounted stock price process
follows
\[
\frac{dS_{t}}{S_{t}}=\mu \left(  V_{t}\right)  dt+\sigma \left(  V_{t}\right)
dW_{t}^{1},\quad S_{0}=S>0,
\]
where $\mu$ and $\sigma>0$ are deterministic functions. The stochastic factor
solves
\begin{equation}
dV_{t}=\emph{l\left(  V_{t}\right)  }dt+\left(  \rho dW_{t}^{1}+\sqrt
{1-\rho^{2}}dW_{t}^{2}\right)  ,\quad t\geq0,\quad V_{0}=v\in \mathbb{R}%
\text{,}\quad \rho \in \left[  0,1\right]  . \label{sf}%
\end{equation}
The controlled wealth state equation (\ref{wealth}) then becomes
\begin{equation}
dX_{t}^{\pi}=\pi_{t}\left(  \theta \left(  V_{t}\right)  ds+dW_{t}^{1}\right)
\text{,}\quad t\geq0\text{, \ }X_{0}^{\pi}=x\in \mathbb{R}\text{,} \label{sfw}%
\end{equation}
with $\theta \left(  v\right)  :=\frac{\mu \left(  v\right)  }{\sigma \left(
v\right)  }$. We assume that, for $l,\theta:\mathbb{R\rightarrow R}$,

\begin{enumerate}
\item[(i)] there exists a large enough constant $C>0$ such that
\[
\left(  l\left(  v\right)  -l\left(  v^{\prime}\right)  \right)  \left(
v-v^{\prime}\right)  \leq-C\left \vert v-v^{\prime}\right \vert ^{2},
\]

\item[(ii)] $\theta$ is uniformly bounded and Lipschitz continuous.
\end{enumerate}

Due to the homothetic property of the exponential forward performance process
and the single stochastic factor Markovian setup, SPDE (\ref{FSPDE})
simplifies to an ergodic BSDE, derived in \cite{LZ2017}. Indeed, consider the
ergodic BSDE
\begin{equation}
dY_{t}^{e}=\left(  \theta \left(  V_{t}\right)  Z_{t}^{e,1}+\frac{1}%
{2}\left \vert \theta \left(  V_{t}\right)  \right \vert ^{2}-\frac{1}%
{2}\left \vert Z_{t}^{e,2}\right \vert ^{2}+\lambda t\right)  dt+\left(
Z_{t}^{e}\right)  ^{\top}dW_{t}, \label{expEBSDE}%
\end{equation}
for which we easily deduce that it admits a unique Markovian solution
$(Y_{t}^{e},Z_{t}^{e},\lambda)=(y^{e}(V_{t}),z^{e}(V_{t}),\lambda)$, $t\geq0$,
where $y^{e}:\mathbb{R}\rightarrow \mathbb{R}$ has at most linear growth and
$z^{e}:\mathbb{R}\rightarrow \mathbb{R}^{2}$ is bounded. Then, the process
given by%
\begin{equation}
U\left(  t,x\right)  =-e^{-\gamma x+Y_{t}^{e}-\lambda t},\quad(t,x)\in
\mathbb{D}\text{,\quad}\gamma>0, \label{expFPP}%
\end{equation}
is an exponential forward performance process. It provides a solution to SPDE
(\ref{FSPDE}) with drift and volatility choice as%
\begin{equation}
\beta \left(  t,x\right)  =\frac{1}{2}U\left(  t,x\right)  \left \vert
\theta \left(  V_{t}\right)  +Z_{t}^{e,1}\right \vert ^{2}\quad \text{and}%
\quad \alpha^{i}\left(  t,x\right)  =U\left(  t,x\right)  Z_{t}^{e,i},\text{
}i=1,2. \label{expFPPD}%
\end{equation}
Moreover, the optimal control process $\pi^{\ast}$ for (\ref{FPP}) is given by%
\[
\pi_{t}^{\ast}=\frac{\theta \left(  V_{t}\right)  +Z_{t}^{e,1}}{\gamma},\quad
t\geq0.
\]

It follows easily that $U$ in (\ref{expFPP}) satisfies both conditions in
Assumption \ref{Assump}. The ergodic BSDE (\ref{expEBSDE}) can be also used to
characterize the convex conjugate of $U$, which takes the form
\begin{equation}
\tilde{U}\left(  t,z\right)  =-\frac{z}{\gamma}+\frac{z}{\gamma}\ln \frac
{z}{\gamma}-\frac{z}{\gamma}\left(  Y_{t}^{e}-\lambda t\right)  ,\quad
(t,z)\in \mathbb{D}^{+}. \label{expCPU}%
\end{equation}
In turn, the primal problem (\ref{P}) and the dual problem (\ref{CP})
simplify, respectively, to, for $0\leq t\leq T$,%
\begin{equation}
u^{P}\left(  t,\xi;T\right)  =\operatorname*{esssup}_{\pi \in \mathcal{A}%
_{[t,T]}}\mathbb{E}\left[  \left.  -e^{-\gamma \left(  \xi+\int_{t}^{T}\pi
_{r}\left(  \theta \left(  V_{r}\right)  dr+dW_{r}^{1}\right)  +P\right)
+Y_{T}^{e}-\lambda T}\right \vert \mathcal{F}_{t}\right]  =e^{-\gamma \xi}%
u^{P}\left(  t,0;T\right)  , \label{expP1}%
\end{equation}
and{ }
\begin{align}
\tilde{u}^{P}\left(  t,\eta;T\right)  =  &  \operatorname*{essinf}%
_{q\in \mathcal{Q}_{[t,T]}}\mathbb{E}\left[  \left.  -\frac{\eta M_{T}^{t,q}%
}{\gamma}+\frac{\eta M_{T}^{t,q}}{\gamma}\ln \frac{\eta M_{T}^{t,q}}{\gamma
}-\frac{\eta M_{T}^{t,q}}{\gamma}\left(  Y_{T}^{e}-\lambda T\right)  +\eta
M_{T}^{t,q}P\right \vert \mathcal{F}_{t}\right] \label{expDP}\\
=  &  \eta \tilde{u}^{P}\left(  t,1;T\right)  +\frac{1}{\gamma}\eta \ln
\eta.\nonumber
\end{align}

For this class of exponential forward performance processes, the forward and
backward equations in both the primal FBSDE (\ref{FBSDEsystem}) and the dual
FBSDE (\ref{DFBSDEsystem}) become \textit{decoupled}. Specifically,
(\ref{FBSDEsystem}) reduces, for $0\leq t\leq s\leq T$, to
\begin{equation}
\left \{
\begin{array}
[c]{l}%
X_{s}^{P}=\xi+%
{\displaystyle \int_{t}^{s}}
\left(  \dfrac{\theta \left(  V_{r}\right)  +Z_{r}^{e,1}}{\gamma}-Z_{r}%
^{P,1}\right)  \left(  \theta \left(  V_{r}\right)  dr+dW_{r}^{1}\right)  ,\\
Y_{s}^{P}=P-%
{\displaystyle \int_{s}^{T}}
\left(  Z_{r}^{P,1}\theta \left(  V_{r}\right)  +\dfrac{\gamma}{2}\left \vert
Z_{r}^{P,2}\right \vert ^{2}-Z_{r}^{e,2}Z_{r}^{P,2}\right)  dr-%
{\displaystyle \int_{s}^{T}}
\left(  Z_{r}^{P}\right)  ^{\top}dW_{r}.
\end{array}
\right.  \label{BSDE}%
\end{equation}
Recalling Theorem 3.2 in \cite{CHLZ2019}, there exists a unique solution
$(Y^{P},Z^{P})$ with $Y^{P}$ being bounded and $Z^{P,i}\in \mathbb{L}_{BMO}%
^{2}[t,T],\ i=1,2$. \emph{We note that, in the exponential preference case,
unbounded random endowment $P$ can also be accommodated, provided it satisfies
$\mathbb{E}[e^{\gamma^{\prime}P}]<\infty$ for some $\gamma^{\prime}>\gamma$
(see \cite{BH2006,HLT2024}).}

By Theorem \ref{FBSDE1} (ii), the optimal control process for $u^{P}(t,\xi
;T)$, $0\leq t\leq T$, is given by%
\begin{equation}
\pi_{s}^{\ast,P}=\dfrac{\theta \left(  V_{s}\right)  +Z_{s}^{e,1}}{\gamma
}-Z_{s}^{P,1},\quad t\leq s\leq T. \label{expO}%
\end{equation}
Moreover, using Remark \ref{DO2}, we also derive the optimal control process
for the dual problem (\ref{CP}),
\begin{equation}
q_{s}^{\ast,P}=\gamma Z_{s}^{P,2}-Z_{s}^{e,2},\quad t\leq s\leq T.
\label{expDO}%
\end{equation}

The dual FBSDE (\ref{DFBSDEsystem}) reduces to the following decoupled FBSDE,
for $0\leq t\leq s\leq T$,
\[
\left \{
\begin{array}
[c]{l}%
D_{s}^{P}=\eta-%
{\displaystyle \int_{s}^{T}}
D_{r}^{P}\left(  \theta \left(  V_{r}\right)  dW_{r}^{1}+\left(  \gamma
\tilde{Z}_{r}^{P,2}-Z_{r}^{e,2}\right)  dW_{r}^{2}\right)  ,\\
\tilde{Y}_{s}^{P}=P-%
{\displaystyle \int_{s}^{T}}
\left(  \tilde{Z}_{r}^{P,1}\theta \left(  V_{r}\right)  +\dfrac{\gamma}%
{2}\left \vert \tilde{Z}_{r}^{P,2}\right \vert ^{2}-Z_{r}^{e,2}\tilde{Z}%
_{r}^{P,2}\right)  dr-%
{\displaystyle \int_{s}^{T}}
\left(  \tilde{Z}_{r}^{P}\right)  ^{\top}dW_{r}.
\end{array}
\right.
\]

Comparing with the BSDE in (\ref{BSDE}), we notice that $(\tilde{Y}^{P}%
,\tilde{Z}^{P})=(Y^{P},Z^{P})$ by the uniqueness of the solution. Then, from
Theorem \ref{DFBSDE1}, the optimal control process $q^{\ast,P}$ for (\ref{CP})
is given by
\[
q_{s}^{\ast,P}=\gamma \tilde{Z}_{s}^{P,2}-Z_{s}^{e,2},\quad0\leq t\leq s\leq
T,
\]
which aligns with (\ref{expDO}). Moreover, from Remark \ref{DOP2}, the optimal
control process $\pi^{\ast,P}$ for (\ref{P}) is given by
\[
\pi_{s}^{\ast,P}=\tfrac{\theta \left(  V_{s}\right)  +Z_{s}^{e,1}}{\gamma
}-\tilde{Z}_{s}^{P,1},\quad0\leq t\leq s\leq T,
\]
which aligns with (\ref{expO}).

Finally, we derive the explicit forms of the value functions $u^{P}\left(
t,\xi;T\right)  $ and $\tilde{u}^{P}\left(  t,\eta;T\right)  $, $0\leq t\leq
T$, using the explicit representations of $U$ in (\ref{expFPP}) and $\tilde
{U}$ in (\ref{expCPU}). For this, we first, observe that $U\left(  t,x\right)
=-\frac{1}{\gamma}U_{x}\left(  t,x\right)  $. In turn, from Theorem
\ref{FBSDE1} (ii), we have
\[
u^{P}\left(  t,\xi;T\right)  =\mathbb{E}\left[  \left.  U\left(  T,X_{T}%
^{P}+P\right)  \right \vert \mathcal{F}_{t}\right]  =-\frac{1}{\gamma
}\mathbb{E}\left[  \left.  U_{x}\left(  T,X_{T}^{P}+P\right)  \right \vert
\mathcal{F}_{t}\right]  .
\]
It follows by Lemma \ref{NM} that $U_{x}(s,X_{s}^{P}+Y_{s}^{P})$, $t\leq s\leq
T$, is a true $\mathbb{F}$-martingale, which implies that%
\[
u^{P}\left(  t,\xi;T\right)  =-\frac{1}{\gamma}U_{x}\left(  t,\xi+Y_{t}%
^{P}\right)  =U\left(  t,\xi+Y_{t}^{P}\right)  .
\]
Next, noting that $\tilde{U}\left(  t,z\right)  =-\frac{z}{\gamma}+z\tilde
{U}_{z}\left(  t,z\right)  $, we obtain by Theorem \ref{DOP1} (ii) that,
\begin{align}
\tilde{u}^{P}\left(  t,\eta;T\right)   &  =\mathbb{E}\left[  \tilde{U}\left(
T,D_{T}^{P}\right)  +D_{T}^{P}P\right]  =-\frac{\eta}{\gamma}+\mathbb{E}%
\left[  D_{T}^{P}\left(  \tilde{U}_{z}\left(  T,D_{T}^{P}\right)  +P\right)
\right] \nonumber \\
&  =-\frac{\eta}{\gamma}+\eta \left(  \tilde{U}\left(  t,\eta \right)
+\tilde{Y}_{t}^{P}\right)  =\tilde{U}\left(  t,\eta \right)  +\eta \tilde{Y}%
_{t}^{P}. \label{0304}%
\end{align}
\emph{Comparing (\ref{0304}) to (\ref{3002}), we observe that the last term in
(\ref{3002}), $\mathbb{E}^{\mathbb{Q}^{0}}\left[  P|\mathcal{F}_{t}\right]  $,
is now replaced by $\tilde{Y}_{t}^{P}$.}

\subsection{\emph{Stochastic factor model and general forward performance
criteria: decoupling field method}\label{Decoupling}}

By employing the decoupling field method used in \cite{FI2020}, we establish
the solvability of the primal FBSDE (\ref{FBSDEsystem}) for general forward
performance criteria in single stochastic factor models. The concepts of the
decoupling field are provided in Appendix \ref{DF}, and we refer to
\cite{Fthesis,FI2020} for further background on this method.

Consider the stochastic factor model introduced in previous subsection. As
shown in \cite[Section 6.2.4]{MZ2010}, the forward performance process is
associated with the (ill-posed) forward HJB equation
\begin{equation}
u_{t}-\frac{1}{2}\frac{\left \vert u_{x}\theta \left(  v\right)  +\rho
u_{xv}\right \vert ^{2}}{u_{xx}}+\frac{1}{2}u_{vv}+l\left(  v\right)
u_{v}=0,\quad t\geq0,\quad v,x\in \mathbb{R}, \label{HJB}%
\end{equation}
with a suitable initial condition $u(0,v,x)$. If $u$ solves (\ref{HJB}), then
the process $U(t,x):=u(t,V_{t},x)$ satisfies the forward performance SPDE
(\ref{FSPDE}) with the choice of forward volatility%
\[
\alpha^{1}\left(  t,x\right)  =\rho u_{v}\left(  t,V_{t},x\right)
\text{\  \ and\  \ }\alpha^{2}\left(  t,x\right)  =\sqrt{1-\rho^{2}}u_{v}\left(
t,V_{t},x\right)  .
\]
The forward HJB equation (\ref{HJB}) was investigated in
\cite{NZ2014,Nadtochiy_2017} where specific examples are provided.

For simplicity, we introduce the auxiliary notation%
\[
\varphi:=\frac{u_{x}}{u_{xx}},\text{ }\phi^{1}:=\frac{u_{xxx}}{u_{xx}},\text{
}\phi^{2}:=\frac{u_{xxv}}{u_{xx}}\text{\  \thinspace \ and \thinspace \ }%
\psi:=\dfrac{u_{x}\theta \left(  v\right)  +\rho u_{xv}}{u_{xx}}\text{.}%
\]

To align with Assumption \ref{Assump}, we introduce the following conditions
on the function $u(t,v,x)$.

\begin{condition}
\label{A.1}

\begin{enumerate}
\item[(i)] $u(t,v,x)\in C^{1,2,4}([0,\infty)\times \mathbb{R}\times \mathbb{R})$
and, for each $t\geq0$ and $v\in \mathbb{R}$, the mapping $x\mapsto u(t,v,x)$
is strictly increasing and strictly concave.

\item[(ii)] There exist positive constants $C_{l},C_{u}$ and $C_{\alpha}$ such
that $0<C_{l}\leq-\varphi \leq C_{u}$ and $|\frac{u_{xv}}{u_{xx}}|\leq
C_{\alpha}.$
\end{enumerate}
\end{condition}

Let $T>0$ be arbitrary but fixed. We consider a random endowment of the form
$P(V_{T},X_{T}^{\pi})$, where $P:\mathbb{R}\times \mathbb{R\rightarrow R}$ is
deterministic. In this case, for $0\leq s\leq T$, the primal FBSDE
(\ref{FBSDEsystem}) takes the form
\begin{equation}
\left \{
\begin{array}
[c]{l}%
X_{s}^{P}=x-%
{\displaystyle \int_{0}^{s}}
\left(  \psi \left(  r,V_{r},X_{r}^{P}+Y_{r}^{P}\right)  +Z_{r}^{P,1}\right)
\left(  \theta \left(  V_{r}\right)  dr+dW_{r}^{1}\right)  ,\\
Y_{s}^{P}=P\left(  V_{T},X_{T}^{P}\right)  -%
{\displaystyle \int_{s}^{T}}
\left(  Z_{r}^{P}\right)  ^{\top}dW_{r}\\
\quad \quad+%
{\displaystyle \int_{s}^{T}}
\left(  -Z_{r}^{P,1}\theta \left(  V_{r}\right)  +\dfrac{1}{2}\phi^{1}\left(
r,V_{r},X_{r}^{P}+Y_{r}^{P}\right)  \left \vert Z_{r}^{P,2}\right \vert
^{2}+\sqrt{1-\rho^{2}}\phi^{2}\left(  r,V_{r},X_{r}^{P}+Y_{r}^{P}\right)
Z_{r}^{P,2}\right)  dr,
\end{array}
\right.  \label{MarFBSDE}%
\end{equation}
where the superscript $P$ represents the function appearing the terminal
condition. We introduce the following additional assumptions.

\begin{condition}
\label{A.2}

\begin{enumerate}
\item[(i)] The functions $\phi^{i},\phi_{x}^{i},\phi_{v}^{i},$ $i=1,2,$ and
$\dfrac{u_{xvv}}{u_{xx}}$ are uniformly bounded, and $\phi_{x}^{1}<0$.

\item[(ii)] $P$ is uniformly bounded and Lipschitz continuous in both
components with $L_{P,x}<1$, where $L_{P,x}$ denotes the Lipschitz constant of
$P(v,x)$ with respect to the second component (see (\ref{LC})).

\item[(iii)] $\theta$ and $l$ are uniformly bounded, differentiable in $v$,
and with bounded derivatives.
\end{enumerate}
\end{condition}

\begin{remark}
We provide a simple example of $u$ satisfying the above conditions, based on
\cite[Lemma 1.2]{FI2020}. Let $\theta=0$ and consider $u(t,v,x)$ in the form
$u\left(  t,v,x\right)  =\tilde{u}\left(  x\right)  +F\left(  t,v\right)  $
with
\[
\tilde{u}\left(  x\right)  =-\int_{x}^{\infty}\int_{y}^{\infty}e^{-\kappa
\left(  z\right)  }dzdy,
\]
where $\kappa \in C^{2}(\mathbb{R})$ satisfies $0<\delta_{l}\leq \kappa^{\prime
},\kappa^{\prime \prime}\leq \delta_{u}$ for some constants $\delta_{l}%
,\delta_{u}$. Suppose that $F$ solves PDE
\[
F_{t}+\frac{1}{2}F_{vv}+l\left(  v\right)  F_{v}=0\text{,}%
\]
where $\frac{F_{v}}{F},\frac{F_{vv}}{F}$ are uniformly bounded. Then, by
(\cite[Lemma 1.2]{FI2020}), it follows that all the above assumptions on $u$
are satisfied. The solvability of the PDE for $F$ has been studied in detail
in \cite{Nadtochiy_2017} using Widder's theorem.
\end{remark}

We consider following cases:

\noindent \textbf{Case 1.} $\phi_{x}^{2}=0$ ($\phi_{x}^{2}$ denotes the partial
derivative of $\phi^{2}$ with respect to $x$),

\noindent \textbf{Case 2.} $\phi_{x}^{2}\neq0$, $\phi_{x}^{1}\leq-\delta$ for
some $\delta>0$. Furthermore, $T<\frac{1}{K}\ln \frac{2}{1+L_{P,x}}$, where
constant $K$ satisfies $\frac{1-\rho^{2}}{2}|\phi_{x}^{2}|^{2}\leq-K\phi
_{x}^{1}$.

\begin{theorem}
\label{Mar}Suppose that Assumptions \ref{A.1} and \ref{A.2} hold. Then, under
either Case 1 or Case 2, the primal FBSDE (\ref{MarFBSDE}) admits a unique
solution $(X^{P},Y^{P},Z^{P})$ with $Z^{P}$ being uniformly bounded.
\end{theorem}

Using the above result, we may, in turn, obtain the existence and uniqueness
of the associated dual FBSDE in the Markovian form, following the results in
Proposition \ref{R1}.

\section{Forward optimized certainty equivalent (Forward
OCE)\label{Application}}

The concept of optimized certainty equivalent (OCE) was first introduced by
Ben-Tal and Teboulle in \cite{BT1986} and yields a valuation criterion rooted
in expected utility theory. It is a \textit{static} criterion, defined through
the optimization problem,
\begin{equation}
S\left(  P\right)  :=\sup_{r\in \mathbb{R}}\left(  \mathbb{E}\left[  u\left(
P-r\right)  \right]  +r\right)  , \label{OCE}%
\end{equation}
where $u$ is the utility of the investor and $P$ is a claim/random endowment.
Essentially, the OCE represents the optimal split of funds if the investor has
the option to choose to spend a portion $r$ of the random endowment, and
receive the present value of this $r$ plus the expected utility of $P-r$.

\emph{Note that the optimization in (\ref{OCE}) balances the \textit{expected
utility satisfaction from the uncertain payoff}, $P - r$, with the
\textit{certainty of receiving $r$}. Since both terms are considered unitless,
their sum represents the combination of the expected utility from $P - r$
units of wealth and the present value of $r$ units of wealth. This
interpretation is consistent with the dual representation in (\ref{D1}), where
the first term corresponds to the expected value of $P$, and the second term
represents the penalty function, both of which are also unitless.}

One of the key results established by Ben-Tal and Teboulle in their followup
work \cite{BT2007} is the dual representation for OCE, which we recall next.
Let $\tilde{u}$ be the convex conjugate of $u$, $\tilde{u}\left(  z\right)
:=\sup_{x\in \mathbb{R}}(u\left(  x\right)  -xz)$. Then, the OCE of $P$ admits
the dual representation,
\begin{equation}
S\left(  P\right)  =\inf_{\mathbb{Q\in}\mathcal{Q}}\left(  \mathbb{E}%
^{\mathbb{Q}}\left[  P\right]  +I_{\tilde{u}}\left(  \mathbb{Q},\mathbb{P}%
\right)  \right)  , \label{D1}%
\end{equation}
where $\mathcal{Q}$ is the set of probability measures equivalent to
$\mathbb{P}$, and $I_{\tilde{u}}(\mathbb{Q},\mathbb{P})$ is the associated
penalty function, given by%
\[
I_{\tilde{u}}\left(  \mathbb{Q},\mathbb{P}\right)  :=\mathbb{E}\left[
\tilde{u}\left(  \frac{d\mathbb{Q}}{d\mathbb{P}}\right)  \right]  .
\]

In the context of risk measures, the quantity $\rho(P):=-S(P)$ turns out to be
a convex risk measure for $P$ and admits the dual representation%
\[
\rho \left(  P\right)  =\sup_{\mathbb{Q\in}\mathcal{Q}}\left(  \mathbb{E}%
^{\mathbb{Q}}\left[  -P\right]  -I_{\tilde{u}}\left(  \mathbb{Q}%
,\mathbb{P}\right)  \right)  ,
\]
also known as a divergence risk measure (see \cite{CK2007} amd
\cite{Rufloff2017}). See also \cite{Bartl2020, KMT2022} for its numerical computational aspects.

A challenging problem, which we attempt to address herein, is how to produce a
genuinely dynamic extension of OCE in (\ref{OCE}). Specifically, we are
interested in building a valuation framework that is viable for all claims and
all maturities and, furthermore, yields time-consistent OCE.

We face several difficulties here, both conceptual and technical. Firstly, it
is not clear how the static valuation rule (\ref{OCE}) should be modified
across times. The obvious choice to merely replace the utility $u$ by a
utility, say $U$, at a given future horizon $T$, will not work as it will
generate, to say the least, similar difficulties with the ones we face in the
random endowment and the indifference valuation settings. Among others, such a
framework will neither allow for valuation beyond $T$, nor will allow for
adaptive model revision. \emph{One may simply propose to work in an infinite
horizon to avoid horizon limitations. However, this will require to put very
stringent constraints on a pre-chosen market model in $[0,\infty)$ to maintain
the consistency as explained in the introduction.}

As mentioned in the Introduction, efforts have been made to build dynamic
extensions of OCE. Recently, Backhoff-Veraguas et al. (see \cite{BT2020} and
\cite{BRT2022}) proposed a dynamic version working with the convex dual and,
in addition, introduced an additional variable to guarantee time-consistency.
This approach can be compared with the convex dual representation of dynamic
risk measures in \cite{DPR2010}. However, the valuation machanism remains
bound to the fixed horizon $T$ as we mentioned earlier.

\subsection{Interpretation of OCE via convex dual of utility maximization with
random endowment}

Before we introduce the forward OCE, we provide a key observation which, to
the best of our knowledge, has not been employed so far. It interprets the
existing static OCE \textit{via the value function of the convex dual of a
utility maximization problem with random endowment}.

To this end, we observe that the OCE of a random endowment $P\in L^{\infty
}(\mathcal{F}_{T})$ can be rewritten as%
\begin{equation}
S\left(  P\right)  =\sup_{x\in \mathbb{R}}\left(  \mathbb{E}\left[  u\left(
x+P\right)  \right]  -x\right)  \quad \text{with}\quad x=-r. \label{2101}%
\end{equation}
\emph{This provides a new interpretation of OCE: the value $S(P)$ represents
the optimal allocation of wealth when the investor decides to deposit $x$ into
a savings account at present. Consequently, at maturity, she will receive an
additional amount $x$ along with the endowment $P$, and the corresponding
value is the expected utility of $x+P$.}

\emph{When $P $ is negative, this interpretation adjusts slightly. In this
case, $P $ can be seen as a liability or loss. The investor is effectively
considering how much to deposit into the savings account to offset this
negative endowment. The OCE still indicates the optimal allocation of wealth,
but it now reflects the investor's strategy to minimize the impact of the
negative endowment by determining how much additional capital $x $ is needed
to achieve a level of utility that compensates for the adverse effects of $P
$. }

Let us then consider the utility maximization problem with the random
endowment $P$,
\[
v^{P}\left(  x\right)  =\sup_{X_{T}\in L^{0}\left(  \mathcal{F}_{T}\right)
}\mathbb{E}\left[  u\left(  X_{T}+P\right)  \right]  ,
\]
under the constraint $\mathbb{E}^{\mathbb{Q}}\left[  X_{T}\right]  \leq x$ for
all equivalent martingale measures $\mathbb{Q}$, and $X_{T}$ being the
terminal wealth generated by trading strategies with initial wealth
$x\in \mathbb{R}$. If we assume that, for any $X_{T}\in L^{0}(\mathcal{F}%
_{T}),$
\begin{equation}
X_{T}\text{ is independent of }P\text{ and }\mathbb{E}[X_{T}]=x\text{, i.e.,
}\mathbb{P}\text{ itself is a martingale measure,} \label{a}%
\end{equation}

\noindent then, by Jensen's inequality, we have%
\[
\mathbb{E}\left[  u\left(  X_{T}+P\right)  \right]  =\mathbb{E}\left[  \left.
\mathbb{E}\left[  u\left(  X_{T}+p\right)  \right]  \right \vert _{p=P}\right]
\leq \mathbb{E}\left[  \left.  u\left(  \mathbb{E}\left[  X_{T}+p\right]
\right)  \right \vert _{p=P}\right]  =\mathbb{E}\left[  u\left(  x+P\right)
\right]  .
\]
In other words, we obtain that the optimal wealth is given by $X_{T}^{\ast}%
=x$, and the value function is given by
\[
v^{P}\left(  x\right)  =\mathbb{E}\left[  u\left(  x+P\right)  \right]  .
\]
Therefore, $S(P)$ can be expressed as the convex conjugate of $v^{P}$ at value
$1$,
\begin{equation}
S\left(  P\right)  =\sup_{x\in \mathbb{R}}\left(  v^{P}\left(  x\right)
-x\cdot1\right)  =:\left(  v^{P}\right)  ^{\ast}\left(  1\right)  , \label{MO}%
\end{equation}
where $(v^{P})^{\ast}$ represents the convex conjugate of the value function
$v^{P}$. In other words, $S(P)$ can be seen as the convex conjugate, evaluated
at value $1$, of a utility maximization problem involving $P$ as the random
endowment in an auxiliary financial market whose driving noise is orthogonal
to this random endowment. In this setting, the optimal policy involves holding
only the riskless asset, so that $X_{T}^{\ast}=x$.

\subsection{Forward OCE\label{Forward OCE}}

We are now ready to introduce the novel concept of forward OCE, which, to the
best of our knowledge, is new. The definition makes full use of the key
observation above for the representation (\ref{2101}) and (\ref{MO}) of the
OCE and employs the suitable primal and dual forward counterparts of the
static utility in (\ref{OCE}) and (\ref{D1}).

\emph{In the following definition of forward OCE, we will make use of an
auxiliary variable $\eta$, which can be naturally interpreted as a deflator,
serving to convert nominal values of current wealth into real values. A
similar concept has been explored in \cite{BT2020}, albeit without the context
of utility maximization. Denote the set of deflators at time $t\geq0$ by
\[
\mathcal{H}_{t}:=\{ \eta \in L^{0,+}(\mathcal{F}_{t})\text{, }\mathbb{E}%
\mathcal{[}\eta]=1\} \text{.}%
\]
}

\begin{definition}
\label{FOCE}Let $T>0$ be arbitrary and fixed.

\begin{enumerate}
\item[(i)] For $0\leq t\leq T$ and $\eta \in \mathcal{H}_{t}$, the forward OCE
of $P\in L^{\infty}(\mathcal{F}_{T})$ at time $t$ is defined by%
\begin{align*}
F\left(  t,P;\eta,T\right)   &  :=\operatorname*{esssup}_{\xi \in \cap_{p\geq
1}L^{p}(\mathcal{F}_{t})}\left(  u^{P}\left(  t,\xi;T\right)  -\xi \eta \right)
\\
&  =\operatorname*{esssup}_{\xi \in \cap_{p\geq1}L^{p}(\mathcal{F}_{t})}\left(
\operatorname*{esssup}_{\pi \in \mathcal{A}_{[t,T]}}\mathbb{E}\left[  \left.
U\left(  T,\xi+\int_{t}^{T}\pi_{r}\left(  \theta_{r}dr+dW_{r}^{1}\right)
+P\right)  \right \vert \mathcal{F}_{t}\right]  -\xi \eta \right)  \text{.}%
\end{align*}

\item[(ii)] In particular, at $t=0$,
\[
F\left(  0,P;1,T\right)  =\sup_{x\in \mathbb{R}}\left(  \sup_{\pi \in
\mathcal{A}_{\left[  0,T\right]  }}\mathbb{E}\left[  U\left(  T,x+\int_{0}%
^{T}\pi_{r}\left(  \theta_{r}dr+dW_{r}^{1}\right)  +P\right)  \right]
-x\right)  .
\]

\end{enumerate}
\end{definition}

\bigskip

The forward OCE represents the result of an optimal dynamic allocation of
funds. An investor faces the choice of saving a portion of their current
wealth $\xi$ at time $t$ in addition to the random endowment $P$ at time $T$.
The consequence of this decision is a sacrifice in the real value due to
reduced spending, which amounts to $\xi \eta$. Thus, the forward OCE at time
$t$ mirrors the optimal balance between maximizing the forward performance
process from the saving amount of $\xi$ at time $t$ alongside the random
endowment $P$ at time $T$ and reducing spending due to this saving choice. In
the newly defined framework for forward OCE, we introduce an associated
auxiliary financial market where the underlying asset can be utilized to
partially hedge the risk associated with $P$. In contrast, in the classical
OCE framework, the driving noise in this auxiliary market is assumed to be
orthogonal to $P$.

Similar to the classical OCE, we also obtain the following dual representation
for the forward OCE. The proof follows easily from Theorems \ref{DFBSDE1} (ii)
and \ref{DOP1}, and is, then, omitted.

\begin{theorem}
\label{dualFOCE}Let $T>0$, $P\in L^{\infty}(\mathcal{F}_{T})$ and $\eta
\in \mathcal{H}_{t}$, $0\leq t\leq T$. Suppose $(D^{P},\tilde{Y}^{P},\tilde
{Z}^{P})$ is a solution to the dual FBSDE (\ref{DFBSDEsystem}) on $[t,T]$ with
initial-terminal condition $(\eta,P)$ satisfying $D^{P}>0$ and $\tilde
{Z}^{P,i}\in \mathbb{L}_{BMO}^{2}[t,T],\ i=1,2$. Then, the forward OCE of $P$
at time $t$ admits the dual representation
\[
F\left(  t,P;\eta,T\right)  =\operatorname*{essinf}_{q\in \mathcal{Q}_{[t,T]}%
}\mathbb{E}\left[  \left.  \tilde{U}\left(  T,\eta M_{T}^{t,q}\right)  +\eta
M_{T}^{t,q}P\right \vert \mathcal{F}_{t}\right]  .
\]
The corresponding optimal control $q^{\ast,P}$ is given by%
\[
q_{s}^{\ast,P}:=\frac{\tilde{Z}_{s}^{P,2}+\tilde{\alpha}_{z}^{2}\left(
s,D_{s}^{P}\right)  }{D_{s}^{P}\tilde{U}_{zz}\left(  s,D_{s}^{P}\right)
},\quad t\leq s\leq T.
\]

\end{theorem}

In light of Proposition \ref{MI}, we obtain directly the
\textit{maturity-independence} property of the forward OCE under the
self-generation property of $\tilde{U}$.

\begin{proposition}
\label{MIOCE}Let $0<T\leq T^{\prime}$, $P\in L^{\infty}(\mathcal{F}_{T})$ and
$\eta \in \mathcal{H}_{t}$, $0\leq t\leq T$. Then
\[
F\left(  t,P;\eta,T\right)  =F\left(  t,P;\eta,T^{\prime}\right)  .
\]

\end{proposition}

\subsubsection{Connection with the classical (static) OCE}

We discuss how the forward OCE relates to the OCE. Firstly, note that the two
assumptions in (\ref{a}) play a crucial role in transitioning from the
perspective of utility maximization to the classical OCE definition. These
assumptions involve the orthogonality between the market noise and the random
endowment, as well as the exclusive use of the riskless asset as the optimal policy.

In the context of the forward performance framework, we need to consider the
counterparts of (\ref{a}), which ultimately lead to the recovery of classical
OCE. Specifically, we assume the following conditions in the forward
performance framework:

\begin{enumerate}
\item[(i)] the market price of risk $\theta \equiv0,$

\item[(ii)] the random endowment $P\in L^{\infty}(\mathcal{F}_{T}^{2})$, where
$\mathbb{F}^{2}=(\mathcal{F}_{t}^{2})_{t\geq0}$ is the natural filtration of
$W^{2},$

\item[(iii)] the volatility component $\alpha^{1}\equiv0$ in the forward
performance SPDE (\ref{FSPDE}).
\end{enumerate}

Under the above assumptions, the dynamics of the forward performance process
$U$ reduce to $dU\left(  t,x\right)  =\alpha^{2}\left(  t,x\right)  dW_{t}%
^{2}$. Additionally, the value function in (\ref{P}) takes the form:
\[
u^{P}\left(  t,\xi;T\right)  =\operatorname*{esssup}_{\pi \in \mathcal{A}%
_{[t,T]}}\mathbb{E}\left[  \left.  U\left(  T,\xi+\int_{t}^{T}\pi_{r}%
dW_{r}^{1}+P\right)  \right \vert \mathcal{F}_{t}\right]  ,\quad0\leq t\leq
T,\quad \xi \in \cap_{p\geq1}L^{p}(\mathcal{F}_{t}).
\]
From Theorem \ref{FBSDE1} (ii), we deduce that if $(Y^{P},Z^{P,2})$ is a
solution to the BSDE%
\[
Y_{s}^{P}=P+\int_{s}^{T}\left(  \dfrac{1}{2}\dfrac{U_{xxx}\left(  r,\xi
+Y_{r}^{P}\right)  }{U_{xx}\left(  r,\xi+Y_{r}^{P}\right)  }\left \vert
Z_{r}^{P,2}\right \vert ^{2}+\dfrac{\alpha_{xx}^{2}\left(  r,\xi+Y_{r}%
^{P}\right)  }{U_{xx}\left(  r,\xi+Y_{r}^{P}\right)  }Z_{r}^{P,2}\right)
dr-\int_{s}^{T}Z_{r}^{P,2}dW_{r}^{2},\quad0\leq t\leq s\leq T,
\]
namely, $(\xi,Y^{P},(0,Z^{P,2}))$ solves FBSDE (\ref{FBSDEsystem}) under
assumptions (i)-(iii), then the control process $\pi^{\ast,P}=0$ is optimal
for $u^{P}(t,\xi;T)$. Thus, $u^{P}\left(  t,\xi;T\right)  =\mathbb{E}\left[
U\left(  T,\xi+P\right)  |\mathcal{F}_{t}\right]  $, and the forward OCE of
$P$ at time $t$ is given by
\[
F\left(  t,P;\eta,T\right)  =\operatorname*{esssup}_{\xi \in \cap_{p\geq1}%
L^{p}(\mathcal{F}_{t})}\left(  \mathbb{E}\left[  \left.  U\left(
T,\xi+P\right)  \right \vert \mathcal{F}_{t}\right]  -\xi \eta \right)  .
\]
In particular, at $t=0$, we have
\[
F\left(  0,P;1,T\right)  =\sup_{x\in \mathbb{R}}\left(  \mathbb{E}\left[
U\left(  T,x+P\right)  \right]  -x\right)  =\sup_{r\in \mathbb{R}}\left(
\mathbb{E}\left[  U\left(  T,P-r\right)  \right]  +r\right)  .
\]
Thus, if, furthermore, the volatility component $\alpha^{2}\equiv0$, then
$U(t,x)\equiv u(x)$ and the forward OCE coincides with the classical OCE,
$F\left(  0,P;1,T\right)  =S\left(  P\right)  .$

Under assumptions (i)-(iii), we have a more explicit form of the dual
representation,
\[
F\left(  t,P;\eta,T\right)  =\operatorname*{essinf}_{q\in \mathcal{Q}_{[t,T]}%
}\mathbb{E}\left[  \left.  \tilde{U}\left(  T,\eta \mathcal{E}\left(  -%
{\textstyle \int}
qdW^{2}\right)  _{T}^{t}\right)  +\eta \mathcal{E}\left(  -%
{\textstyle \int}
qdW^{2}\right)  _{T}^{t}P\right \vert \mathcal{F}_{t}\right]  ,
\]
where%
\[
\mathcal{E}\left(  -%
{\textstyle \int}
qdW^{2}\right)  _{T}^{t}:=\frac{\mathcal{E}\left(  -%
{\textstyle \int}
qdW^{2}\right)  _{T}}{\mathcal{E}\left(  -%
{\textstyle \int}
qdW^{2}\right)  _{t}}\text{ with }\mathcal{E}\left(  -%
{\textstyle \int}
qdW^{2}\right)  _{t}:=\exp \left(  -\tfrac{1}{2}%
{\textstyle \int_{0}^{t}}
\left \vert q_{r}\right \vert ^{2}dr-%
{\textstyle \int_{0}^{t}}
q_{r}dW_{r}^{2}\right)  .
\]
In particular, at $t=0$, we have
\[
F\left(  0,P;1,T\right)  =\inf_{q\in \mathcal{Q}_{\left[  0,T\right]  }%
}\mathbb{E}\left[  \tilde{U}\left(  T,\mathcal{E}\left(  -%
{\textstyle \int}
qdW^{2}\right)  _{T}\right)  +\mathcal{E}\left(  -%
{\textstyle \int}
qdW^{2}\right)  _{T}P\right]  .
\]

\subsubsection{Properties of the forward OCE}

We conclude this section by presenting several key properties of the forward
OCE. The proofs follow easily from the properties of the forward performance
process. To simplify and unify the presentation, we work with a slight
modification of the forward OCE through normalization. The \textit{normalized}
forward OCE, denoted by $\tilde{F}(t,P;\eta,T)$, is defined as
\begin{equation}
\tilde{F}\left(  t,P;\eta,T\right)  :=F\left(  t,P;\eta,T\right)  -F\left(
t,0;\eta,T\right)  , \label{NFOCE}%
\end{equation}
so that $\tilde{F}\left(  t,0;\eta,T\right)  =0$.

\begin{proposition}
\label{FOCEprop}Let $T>0$ be arbitrary and fixed, and $\eta \in \mathcal{H}_{t}%
$, $0\leq t\leq T$. Then the (normalized) forward OCE has the following properties:

\begin{enumerate}
\item[(i)] Monotonicity: for $P^{i}\in L^{\infty}(\mathcal{F}_{T}),\ i=1,2$
and $P^{1}\geq P^{2}$,
\[
F(t,P^{1};\eta,T)\geq F(t,P^{2};\eta,T).
\]

\item[(ii)] Cash invariance: for $P\in L^{\infty}(\mathcal{F}_{T})$ and $c\in
L^{\infty}(\mathcal{F}_{t})$,
\[
F(t,P+c;\eta,T)=F(t,P;\eta,T)+\eta c.
\]

\item[(iii)] Concavity: for $P^{i}\in L^{\infty}(\mathcal{F}_{T}),\ i=1,2$,
and \emph{for $\lambda \in L^{\infty}(\mathcal{F}_{t})$ valued in $(0,1),$}%
\[
\lambda F\left(  t,P^{1};\eta,T\right)  +\left(  1-\lambda \right)  F\left(
t,P^{2};\eta,T\right)  \leq F\left(  t,P^{\lambda};\eta,T\right)  ,
\]
where $P^{\lambda}:=\lambda P^{1}+(1-\lambda)P^{2}$.

\item[(iv)] Replication invariance: for $P\in L^{\infty}(\mathcal{F}_{T})$ and
for any $\pi \in \mathcal{A}_{[t,T]}$,
\[
F\left(  t,P+\int_{t}^{T}\pi_{r}\left(  \theta_{r}dr+dW_{r}^{1}\right)
;\eta,T\right)  =F\left(  t,P;\eta,T\right)  .
\]

\item[(v)] Positivity: for nonnegative $P\in L^{\infty}(\mathcal{F}_{T})$,
$\tilde{F}(t,P;\eta,T)\geq0$.

\item[(vi)] Constancy: for $c\in L^{\infty}(\mathcal{F}_{t})$,
\[
\tilde{F}(t,c;\eta,T)=\eta c.
\]

\end{enumerate}
\end{proposition}

\begin{remark}
The normalized forward OCE $\tilde{F}$ defined in (\ref{NFOCE}) satisfies
Proposition \ref{MIOCE} and Proposition \ref{FOCEprop} (i)-(iv).
\end{remark}

\begin{remark}
\emph{In the case of complete markets following the discussion at the end of
Section \ref{Complete Market}, we obtain by (\ref{3002}) that%
\[
\tilde{F}\left(  t,P;\eta,T\right)  =\tilde{u}^{P}\left(  t,\eta;T\right)
-\tilde{u}^{0}\left(  t,\eta;T\right)  =\eta \mathbb{E}^{\mathbb{Q}^{0}}\left[
P|\mathcal{F}_{t}\right]  ,
\]
where $\mathbb{Q}^{0}$ is the unique equivalent martingale measure.
Particularly, when $\eta=1$, the normalized forward OCE coincides with the
complete market price, given by $\tilde{F}\left(  t,P;1,T\right)
=\mathbb{E}^{\mathbb{Q}^{0}}\left[  P|\mathcal{F}_{t}\right]  $.}
\end{remark}

\subsection{Forward OCE under exponential forward performance
criteria\label{FOCEexp}}

\emph{Building on the ergodic BSDE representation for the exponential forward
performance process introduced in \cite{LZ2017}, Chong et al. \cite{CHLZ2019}
extended this framework to study exponential forward utility maximization with
random endowment and its application to forward entropic risk measures. In
this section, based on the results in Sections \ref{expfpp} and
\ref{Forward OCE}, we focus on the \textit{normalized} forward OCE defined in
(\ref{NFOCE}), with the exponential forward performance process given by
(\ref{expFPP}). We further demonstrate its direct connection to the negative
of the forward entropic risk measure introduced in \cite{CHLZ2019}.}

For $T>0$, $P\in L^{\infty}(\mathcal{F}_{T})$ and stochastic deflator $\eta
\in \mathcal{H}_{t}$, $0\leq t\leq T$, the definition of the normalized forward
OCE in (\ref{NFOCE}) and Theorem \ref{dualFOCE} yield that%
\[
\tilde{F}\left(  t,P;\eta,T\right)  =\tilde{u}^{P}\left(  t,\eta;T\right)
-\tilde{u}^{0}\left(  t,\eta;T\right)  =\tilde{u}^{P}\left(  t,\eta;T\right)
-\tilde{U}\left(  t,\eta \right)  ,
\]
where we used that $\tilde{u}^{0}(t,\eta;T)=\tilde{U}(t,\eta)$, as it follows
from Proposition \ref{SG}.

On one hand, as shown in (\ref{0304}), the normalized exponenital forward OCE
in terms of the solution component $Y^{P}$ of FBSDE (\ref{BSDE}) is given by
\begin{equation}
\tilde{F}\left(  t,P;\eta,T\right)  =\eta Y_{t}^{P}. \label{EXPFOCE1}%
\end{equation}
By choosing a deterministic deflator $\eta=1$, we observe that (\ref{EXPFOCE1}%
) for the normalized exponential forward OCE corresponds to the negative of
the forward entropic risk measure,
\[
\tilde{F}\left(  t,P;1,T\right)  =Y_{t}^{P}=-\rho \left(  t,P;T\right)  ,
\]
where the last inequality follows from \cite[Theorem 3.2]{CHLZ2019}, and
$\rho$ is defined as the utility indifference price of $P$, namely,
\[
U\left(  t,x\right)  =u^{P}\left(  t,x+\rho \left(  t,P;T\right)  \right)  ,
\]
as stated in \cite[Definition 3.1]{CHLZ2019}. \emph{This result further
implies that, under the exponential case, the forward OCE coincides with the
forward indifference price. Notably, this is the only case in utility
indifference pricing that satisfies the cash invariance property.}

On the other hand, using the dual representation in (\ref{expDP}) yields
\begin{align}
\tilde{F}\left(  t,P;\eta,T\right)   &  =\operatorname*{essinf}_{q\in
\mathcal{Q}_{[t,T]}}\mathbb{E}\left[  \left.  \eta M_{T}^{t,q}\left(
\frac{\ln M_{T}^{t,q}}{\gamma}-\frac{Y_{T}^{e}-\lambda T}{\gamma}+P\right)
\right \vert \mathcal{F}_{t}\right]  +\frac{\eta}{\gamma}\left(  Y_{t}%
^{e}-\lambda t\right) \nonumber \\
&  =\eta \operatorname*{essinf}_{q\in \mathcal{Q}_{[t,T]}}\mathbb{E}%
^{\mathbb{Q}^{q}}\left[  \left.  \frac{\ln M_{T}^{t,q}}{\gamma}-\frac
{Y_{T}^{e}-\lambda T}{\gamma}+P\right \vert \mathcal{F}_{t}\right]  +\frac
{\eta}{\gamma}\left(  Y_{t}^{e}-\lambda t\right)  , \label{expFOCE}%
\end{align}
where $\tfrac{d\mathbb{Q}^{q}}{d\mathbb{P}}|_{\mathcal{F}_{T}}=M_{T}^{t,q}$.
Note that%
\[
\ln M_{T}^{t,q}=\frac{1}{2}\int_{t}^{T}\left(  \left \vert \theta \left(
V_{r}\right)  \right \vert ^{2}+\left \vert q_{r}\right \vert ^{2}\right)
dr-\int_{t}^{T}\theta \left(  V_{r}\right)  dW_{r}^{1,\theta}-\int_{t}^{T}%
q_{r}dW_{r}^{2,q},
\]
with $(dW_{r}^{1,\theta},dW_{r}^{2,q}):=(dW_{r}^{1}+\theta \left(
V_{r}\right)  dr,dW_{r}^{2}+q_{r}dr)$, and, in accordance with the ergodic
BSDE (\ref{expEBSDE}), we have
\[
Y_{T}^{e}-\lambda T=\ Y_{t}^{e}-\lambda t+\int_{t}^{T}\left(  \frac{1}%
{2}\left \vert \theta \left(  V_{r}\right)  \right \vert ^{2}-\frac{1}%
{2}\left \vert Z_{r}^{e,2}\right \vert ^{2}-Z_{r}^{e,2}q_{r}\right)  dr+\int
_{t}^{T}Z_{r}^{e,1}dW_{r}^{1,\theta}+\int_{t}^{T}Z_{r}^{e,2}dW_{r}^{2,q}.
\]
As a result, the dual representation (\ref{expFOCE}) can be expressed more
explicitly as
\begin{equation}
\tilde{F}\left(  t,P;\eta,T\right)  =\eta \operatorname*{essinf}_{q\in
\mathcal{Q}_{[t,T]}}\mathbb{E}^{\mathbb{Q}^{q}}\left[  \left.  P+\frac
{1}{2\gamma}\int_{t}^{T}\left \vert Z_{r}^{e,2}+q_{r}\right \vert ^{2}%
dr\right \vert \mathcal{F}_{t}\right]  , \label{EXPFOCE2}%
\end{equation}
which is the negative of the dual representation for the forward entropic risk
measure established in \cite[Theorem 3.5]{CHLZ2019} for $\eta=1$.

\begin{remark}
\emph{In \cite{A2014}, the author also studied the forward indifference
pricing under the exponential forward performance process investigated in
\cite{Z2009}. Based on the same representation of the exponential forward
performance process, the corresponding forward OCE can be derived directly.
Notably, this forward OCE coincides with the explicit expression provided in
\cite[(3.4)]{A2014}, which represents the exponential forward indifference
price.}
\end{remark}

\section{Backward SPDEs for forward performance criteria with random
endowment}

\label{Relation between FBSDE and SPDE}

To conclude, we briefly discuss how to tackle the primal problem (\ref{P}) and
the dual problem (\ref{CP}) by directly characterizing their value functions
in terms of the solutions of backward SPDEs. We also explore how to use these
backward SPDE solutions to construct the solutions of the primal FBSDE
(\ref{FBSDEsystem}) and dual FBSDE (\ref{DFBSDEsystem}). Since the solvability
of the corresponding backward SPDEs is an open problem, the discussion that
follows is formal. Nevertheless, we find that building this connection is useful.

\subsection{Formal derivation of backward SPDEs}

For a given $T>0$ and random endowment $P\in L^{\infty}(\mathcal{F}_{T})$, we
recall the value function of the primal problem (\ref{P}),
\begin{equation}
u^{P}\left(  t,x;T\right)  =\sup_{\pi \in \mathcal{A}_{[t,T]}}\mathbb{E}\left[
\left.  U\left(  T,X_{T}^{\pi}+P\right)  \right \vert \mathcal{F}_{t}%
,X_{t}^{\pi}=x\right]  ,\quad0\leq t\leq T, \quad x\in \mathbb{R}.
\label{value}%
\end{equation}
Assume that%
\[
du^{P}\left(  t,x;T\right)  =b^{P}\left(  t,x;T\right)  dt+\left(
a^{P}\left(  t,x;T\right)  \right)  ^{\top}dW_{t},
\]
and the It\^{o}-Ventzel formula yields%
\begin{align*}
du^{P}\left(  t,X_{t}^{\pi};T\right)  =  &  \left(  b^{P}\left(  t,X_{t}^{\pi
};T\right)  +\left(  u_{x}^{P}\left(  t,X_{t}^{\pi};T\right)  \theta_{t}%
+a_{x}^{P,1}\left(  t,X_{t}^{\pi};T\right)  \right)  \pi_{t}+\frac{1}{2}%
u_{xx}^{P}\left(  t,X_{t}^{\pi};T\right)  \left \vert \pi_{t}\right \vert
^{2}\right)  dt\\
&  +\left(  a^{P,1}\left(  t,X_{t}^{\pi};T\right)  +u_{x}^{P}\left(
t,X_{t}^{\pi};T\right)  \pi_{t}\right)  dW_{t}^{1}+a^{P,2}\left(  t,X_{t}%
^{\pi};T\right)  dW_{t}^{2}.
\end{align*}

Following the standard martingale optimality condition (as described, for
example, in \cite{LZ2017}), we select the drift $b^{P}(t,x;T)$ to ensure that
$u^{P}(t,X_{t}^{\pi};T)$ is a supermartingale for any $\pi \in \mathcal{A}%
_{[t,T]}$ and a martingale for an optimal control process. This yields that
the drift must of the form
\begin{align*}
b^{P}\left(  t,x;T\right)   &  =-\sup_{\pi \in \mathbb{R}^{n}}\left(  \frac
{1}{2}u_{xx}^{P}\left(  t,x;T\right)  \left \vert \pi \right \vert ^{2}+\left(
u_{x}^{P}\left(  t,x;T\right)  \theta_{t}+a_{x}^{P,1}\left(  t,x;T\right)
\right)  \pi \right) \\
&  =\frac{1}{2}\frac{\left \vert u_{x}^{P}\left(  t,x;T\right)  \theta
_{t}+a_{x}^{P,1}\left(  t,x;T\right)  \right \vert ^{2}}{u_{xx}^{P}\left(
t,x;T\right)  },
\end{align*}
with the above optimum given by
\[
\pi^{\ast,P}\left(  t,x;T\right)  =-\frac{u_{x}^{P}\left(  t,x;T\right)
\theta_{t}+a_{x}^{P,1}\left(  t,x;T\right)  }{u_{xx}^{P}\left(  t,x;T\right)
}.
\]
Thus, $u^{P}$ is expected to satisfy the backward SPDE%
\begin{equation}
du^{P}\left(  t,x;T\right)  =\frac{1}{2}\frac{\left \vert u_{x}^{P}\left(
t,x;T\right)  \theta_{t}+a_{x}^{P,1}\left(  t,x;T\right)  \right \vert ^{2}%
}{u_{xx}^{P}\left(  t,x;T\right)  }dt+\left(  a^{P}\left(  t,x;T\right)
\right)  ^{\top}dW_{t} \label{SPDEup}%
\end{equation}
with terminal condition $u^{P}(T,x;T)=U(T,x+P)$. Note that, different from the
SPDE characterization (\ref{FSPDE}) of the forward performance $U$ where the
volatility is a model input, (\ref{SPDEup}) is a \textit{backward} SPDE with
the volatility $a^{P}$ being part of the solution.

Differentiating (\ref{SPDEup}) in $x$ yields
\begin{equation}
du_{x}^{P}\left(  t,x;T\right)  =b_{x}^{P}\left(  t,x;T\right)  dt+\left(
a_{x}^{P}\left(  t,x;T\right)  \right)  ^{\top}dW_{t},\quad0\leq t\leq T,
\quad x\in \mathbb{R}\text{,} \label{BSPDEP}%
\end{equation}
with terminal condition $u_{x}^{P}(T,x;T)=U_{x}(T,x+P)$, where
\begin{align*}
b_{x}^{P}\left(  t,x;T\right)  =  &  \left(  u_{x}^{P}\left(  t,x;T\right)
\theta_{t}+a_{x}^{P,1}\left(  t,x;T\right)  \right)  \theta_{t}+\frac{\left(
u_{x}^{P}\left(  t,x;T\right)  \theta_{t}+a_{x}^{P,1}\left(  t,x;T\right)
\right)  a_{xx}^{P,1}\left(  t,x;T\right)  }{u_{xx}^{P}\left(  t,x;T\right)
}\\
&  -\frac{1}{2}\frac{\left \vert u_{x}^{P}\left(  t,x;T\right)  \theta
_{t}+a_{x}^{P,1}\left(  t,x;T\right)  \right \vert ^{2}u_{xxx}^{P}\left(
t,x;T\right)  }{\left \vert u_{xx}^{P}\left(  t,x;T\right)  \right \vert ^{2}}.
\end{align*}

Next, we define the convex conjugate of $u^{P}$,
\[
\tilde{u}^{P}\left(  t,z;T\right)  :=\sup_{x\in \mathbb{R}}\left(  u^{P}\left(
t,x;T\right)  -xz\right)  ,\quad0\leq t\leq T, \quad z>0\text{.}%
\]
We know that $u^{P}$ and $\tilde{u}^{P}$ satisfy the analogous relations as
$U$ and $\tilde{U}$ do. As in Section \ref{FPP and its convex conjugate}, we
have the dual relation
\begin{equation}
u_{x}^{P}\left(  t,-\tilde{u}_{z}^{P}\left(  t,z;T\right)  ;T\right)  =z,
\label{RSPDE}%
\end{equation}
from which we can also deduce the dynamics of $\tilde{u}_{z}^{P}(t,z;T)$,
similarly to $\tilde{U}_{z}(t,z)$ in (\ref{DUSPDE}).

\subsection{Backward SPDEs and FBSDEs}

We will use the backward SPDE (\ref{BSPDEP}) satisfied by $u_{x}^{P}$ to
construct a solution of the primal FBSDE (\ref{FBSDEsystem}) and an optimal
control for the primal problem (\ref{P}) as well as the corresponding
quantities for the dual problem (\ref{CP}).

\begin{proposition}
\label{PB}Let $w$ be a solution to the backward SPDE
\begin{equation}
\left \{
\begin{array}
[c]{l}%
dw\left(  t,x\right)  =h\left(  t,x\right)  dt+d^{\top}\left(  t,x\right)
dW_{t},\quad0\leq t\leq T,\quad x\in \mathbb{R},\\
w\left(  T,x\right)  =U_{x}\left(  T,x+P\right)  ,
\end{array}
\right.  \label{SPDEr}%
\end{equation}
where%
\[
h\left(  t,x\right)  =\left(  w\left(  t,x\right)  \theta_{t}+d^{1}\left(
t,x\right)  \right)  \left(  \theta_{t}+\frac{d_{x}^{1}\left(  t,x\right)
}{w_{x}\left(  t,x\right)  }\right)  -\frac{1}{2}\frac{\left \vert w\left(
t,x\right)  \theta_{t}+d^{1}\left(  t,x\right)  \right \vert ^{2}w_{xx}\left(
t,x\right)  }{\left \vert w_{x}\left(  t,x\right)  \right \vert ^{2}},
\]
and assume that $w(t,x)$ is a progressive $\mathcal{K}_{loc}^{1,\delta}%
$-semimartingale random field and $w\in \mathcal{C}^{2}$. Then, for $0\leq
t\leq T$, the following assertions hold.

\begin{enumerate}
\item[(i)] For $\xi \in \cap_{p\geq1}L^{p}(\mathcal{F}_{t})$, the triplet
$(X^{P},Y^{P},Z^{P})$ given by, for $t\leq s\leq T$,
\begin{equation}
X_{s}^{P}:=\xi-\int_{t}^{s}\frac{w\left(  r,X_{r}^{P}\right)  \theta_{r}%
+d^{1}\left(  r,X_{r}^{P}\right)  }{w_{x}\left(  r,X_{r}^{P}\right)  }\left(
\theta_{r}dr+dW_{r}^{1}\right)  , \label{2802}%
\end{equation}%
\[
Y_{s}^{P}:=-\tilde{U}_{z}\left(  s,w\left(  s,X_{s}^{P}\right)  \right)
-X_{s}^{P},
\]
and
\begin{equation}%
\begin{array}
[c]{l}%
Z_{s}^{P,1}:=\dfrac{w\left(  s,X_{s}^{P}\right)  \theta_{s}+d^{1}\left(
s,X_{s}^{P}\right)  }{w_{x}\left(  s,X_{s}^{P}\right)  }-\tilde{\alpha}%
_{z}^{1}\left(  s,w\left(  s,X_{s}^{P}\right)  \right)  +\tilde{U}_{zz}\left(
s,w\left(  s,X_{s}^{P}\right)  \right)  w\left(  s,X_{s}^{P}\right)
\theta_{s},\\
Z_{s}^{P,2}:=-\tilde{\alpha}_{z}^{2}\left(  s,w\left(  s,X_{s}^{P}\right)
\right)  -d^{2}\left(  s,X_{s}^{P}\right)  \tilde{U}_{zz}\left(  s,w\left(
s,X_{s}^{P}\right)  \right)  ,
\end{array}
\label{2801}%
\end{equation}
satisfies FBSDE (\ref{FBSDEsystem}) on $[t,T]$ with initial-terminal condition
$(\xi,P)$.

\item[(ii)] If, furthermore, the processes%
\[
\dfrac{w\left(  s,X_{s}^{P}\right)  \theta_{s}+d^{1}\left(  s,X_{s}%
^{P}\right)  }{w_{x}\left(  s,X_{s}^{P}\right)  }\in \mathbb{L}_{BMO}%
^{2}[t,T]\quad \text{and}\quad d^{2}\left(  s,X_{s}^{P}\right)  \tilde{U}%
_{zz}\left(  s,w\left(  s,X_{s}^{P}\right)  \right)  \in \mathbb{L}_{BMO}%
^{2}[t,T],
\]
then, $Z^{P,i}\in \mathbb{L}_{BMO}^{2}[t,T]$, $i=1,2$, and $\pi^{\ast,P}$ given
by%
\[
\pi_{s}^{\ast,P}:=-\dfrac{w\left(  s,X_{s}^{P}\right)  \theta_{s}+d^{1}\left(
s,X_{s}^{P}\right)  }{w_{x}\left(  s,X_{s}^{P}\right)  },\quad t\leq s\leq T,
\]
is optimal for $u^{P}(t,\xi;T)$. Respectively, the process $q^{\ast,P}$ given
by%
\begin{equation}
q_{s}^{\ast,P}:=-\frac{d^{2}\left(  s,X_{s}^{P}\right)  }{w\left(  s,X_{s}%
^{P}\right)  },\quad t\leq s\leq T, \label{2803}%
\end{equation}
is optimal for $\tilde{u}^{P}\left(  t,\hat{\eta};T\right)  $ with $\hat{\eta
}:=w(t,\xi)$. Furthermore, for any $\eta \in L^{0,+}(\mathcal{F}_{t})$ such
that $\hat{\xi}:=w^{-1}(t,\eta)\in \cap_{p\geq1}L^{p}(\mathcal{F}_{t})$, the
control process $q^{\ast,P}$ defined by (\ref{2803}) with $X^{P}$ satisfying
$X_{t}^{P}=\hat{\xi}$ is optimal for $\tilde{u}^{P}\left(  t,\eta;T\right)  $.
\end{enumerate}
\end{proposition}

It is easy to verify (i) by the It\^{o}-Ventzel formula (\ref{IV}) and the
duality relations between $U$ and $\tilde{U}$ given in Section
\ref{FPP and its convex conjugate}. Part (ii) can be obtained by Theorems
\ref{FBSDE1} (ii), \ref{DO1} and Remark \ref{DO2}.

Similarly to Proposition \ref{PB}, we also have the following results from the
dual backward SPDE.

\begin{proposition}
\label{DPB}Let $\tilde{w}$ be a solution to the backward SPDE
\begin{equation}
\left \{
\begin{array}
[c]{l}%
d\tilde{w}\left(  t,z\right)  =\tilde{h}\left(  t,z\right)  dt+\tilde{d}%
^{\top}\left(  t,z\right)  dW_{t},\quad0\leq t\leq T,\quad z>0,\\
\tilde{w}\left(  T,z\right)  =\tilde{U}_{z}\left(  T,z\right)  +P,
\end{array}
\right.  \label{DSPDEr}%
\end{equation}
with%
\begin{align*}
\tilde{h}\left(  t,z\right)  = &  -\tilde{w}_{z}\left(  t,z\right)
z\left \vert \theta_{t}\right \vert ^{2}+\tilde{d}_{z}^{1}\left(  t,z\right)
z\theta_{t}+\tilde{d}^{1}\left(  t,z\right)  \theta_{t}-\frac{1}{2}\tilde
{w}_{zz}\left(  t,z\right)  \left \vert z\theta_{t}\right \vert ^{2}\\
&  -\frac{1}{2}\frac{\tilde{w}_{zz}\left(  t,z\right)  }{\left \vert \tilde
{w}_{z}\left(  t,z\right)  \right \vert ^{2}}\left \vert \tilde{d}^{2}\left(
t,z\right)  \right \vert ^{2}+\frac{\tilde{d}^{2}\left(  t,z\right)  \tilde
{d}_{z}^{2}\left(  t,z\right)  }{\tilde{w}_{z}\left(  t,z\right)  },
\end{align*}
and assume that $\tilde{w}(t,z)$ is a progressive $\mathcal{K}_{loc}%
^{1,\delta}$-semimartingale random field and $\tilde{w}\in \mathcal{C}^{2}$.
Then, for $0\leq t\leq T$, the following assertions hold.

\begin{enumerate}
\item[(i)] For $\eta \in L^{0,+}(\mathcal{F}_{t})$, the triplet $(D^{P}%
,\tilde{Y}^{P},\tilde{Z}^{P})$ given by, for $t\leq s\leq T$,
\[
D_{s}^{P}:=\eta-%
{\displaystyle \int_{t}^{s}}
\left(  D_{r}^{P}\theta_{r}dW_{r}^{1}+\frac{\tilde{d}^{2}\left(  r,D_{r}%
^{P}\right)  }{\tilde{w}_{z}\left(  r,D_{r}^{P}\right)  }dW_{r}^{2}\right)  ,
\]%
\[
\tilde{Y}_{s}^{P}:=\tilde{w}\left(  s,D_{s}^{P}\right)  -\tilde{U}_{z}\left(
s,D_{s}^{P}\right)  ,
\]
and
\[%
\begin{array}
[c]{l}%
\tilde{Z}_{s}^{P,1}:=\left(  \tilde{w}_{z}\left(  s,D_{s}^{P}\right)
-\tilde{U}_{zz}\left(  s,D_{s}^{P}\right)  \right)  D_{s}^{P}\theta_{s}%
-\tilde{\alpha}_{z}^{1}\left(  s,D_{s}^{P}\right)  ,\\
\tilde{Z}_{s}^{P,2}:=\left(  \tilde{w}_{z}\left(  s,D_{s}^{P}\right)
-\tilde{U}_{zz}\left(  s,D_{s}^{P}\right)  \right)  \dfrac{\tilde{d}%
^{2}\left(  s,D_{s}^{P}\right)  }{\tilde{w}_{z}\left(  s,D_{s}^{P}\right)
}-\tilde{a}_{z}^{2}\left(  s,D_{s}^{P}\right)  ,
\end{array}
\]
satisfies FBSDE (\ref{DFBSDEsystem}) on $[t,T]$ with initial-terminal
condition $(\eta,P)$.

\item[(ii)] If, furthermore, the processes%
\[
\tilde{w}_{z}(s,D_{s}^{P})D_{s}^{P}\in \mathbb{L}_{BMO}^{2}[t,T]\quad
\text{and}\quad \dfrac{\left(  \tilde{w}_{z}\left(  s,D_{s}^{P}\right)
-\tilde{U}_{zz}\left(  s,D_{s}^{P}\right)  \right)  \tilde{d}^{2}\left(
s,D_{s}^{P}\right)  }{\tilde{w}_{z}\left(  s,D_{s}^{P}\right)  }\in
\mathbb{L}_{BMO}^{2}[t,T],
\]
then, $\tilde{Z}^{P,i}\in \mathbb{L}_{BMO}^{2}[t,T]$, $i=1,2$, and $q^{\ast,P}$
given by%
\[
q_{s}^{\ast,P}:=\frac{\tilde{d}^{2}\left(  s,D_{s}^{P}\right)  }{\tilde{w}%
_{z}\left(  s,D_{s}^{P}\right)  D_{s}^{P}},\quad t\leq s\leq T,
\]
is optimal for $\tilde{u}^{P}\left(  t,\eta;T\right)  $. Respectively, the
process $\pi^{\ast,P}$ defined by%
\begin{equation}
\pi_{s}^{\ast,P}:=-\tilde{w}_{z}\left(  s,D_{s}^{P}\right)  D_{s}^{P}%
\theta_{s},\quad t\leq s\leq T, \label{2804}%
\end{equation}
is optimal for $u^{P}(t,\hat{\xi};T)$ with $\hat{\xi}:=-\tilde{w}(t,\eta)$.
Furthermore, for any $\xi \in \cap_{p\geq1}L^{p}(\mathcal{F}_{t})$ and
$\hat{\eta}:=\tilde{w}^{-1}(t,-\xi)$, the control process $\pi^{\ast,P}$
defined by (\ref{2804}) with $D^{P}$ satisfying $D_{t}^{P}=\hat{\eta}$ is
optimal for $u^{P}(t,\xi;T)$.
\end{enumerate}
\end{proposition}

\begin{remark}
Let $w$ and $\tilde{w}$ be the solutions to backward SPDEs (\ref{SPDEr}) and
(\ref{DSPDEr}), respectively, with enough regularity. Applying the
It\^{o}-Ventzel formula to these two equations, one can directly verify that,
for $0\leq t\leq T$, $x\in \mathbb{R}$ and $z>0$,
\begin{equation}
-\tilde{w}\left(  t,w\left(  t,x\right)  \right)  =x\quad \text{and}\quad
w\left(  t,-\tilde{w}\left(  t,z\right)  \right)  =z, \label{rSPDE}%
\end{equation}
which gives the relation between the primal and dual backward SPDEs.
\end{remark}

\section{Conclusions}

\label{sec: conclusion} In this paper, we extended the notion of forward
performance criteria to settings with random endowment in incomplete markets
and studied the related stochastic optimization problems. For this, we
developed a new methodology by directly studying the candidate optimal control
processes for both the primal and dual forward problems. We constructed two
new system of FBSDEs and established necessary and sufficient conditions for
optimality, and various equivalences between the primal and dual problems.
This new approach is general and complements the existing one based on
backward SPDEs for the related value functions.

Building on these results, we introduced and developed the novel concept of
forward optimized certainty equivalent, which offers a genuinely dynamic
valuation mechanism that accommodates progressively adaptive market model
updates, stochastic risk preferences, and incoming claims with arbitrary
maturities. We, also, considered representative examples for both forward
performance criteria with random endowment and forward OCE. In particular, for
the case of exponential forward performance criteria, we investigated the
connection of forward OCE with the forward entropic risk measure.

\appendix

\section*{Appendices}

\addcontentsline{toc}{section}{Appendices}
\renewcommand{\thesubsection}{\Alph{subsection}}
\renewcommand{\theequation}{\Alph{subsection}.\arabic{equation}}\setcounter{equation}{0}
\renewcommand{\thetheorem}{\Alph{subsection}.\arabic{theorem}}\setcounter{theorem}{0}

\subsection{\emph{Differential rule and It\^{o}-Ventzel formula}\label{I-V}}

We recall various results we use herein for the derivatives of processes and
the related It\^{o}-Ventzel formula, stating in detail the underlying regality
assumptions. We refer the reader to Kunita \cite{K1997} for further details of
regularity and properties of random fields, and also to
\cite{El_Karoui_2018,EM2014} for applications of some of these results to the
SPDE associated with forward performance processes.

We begin with background definitions for \textit{regular random fields}.
Calligraphic notations are used for the set of random fields.

For $\mathbb{H}=\mathbb{D},\mathbb{D}^{+}$, let $G(t,x)$, $(t,x)\in \mathbb{H}%
$, be a progressive random field. Let $m\geq0$ be an integer and $\delta
\in(0,1]$.

\begin{itemize}
\item $G$ is called $\mathcal{C}^{m}$-regular if $G$ is $m$-times continuously
differentiable in $x$ for any $t$ a.s. We denote this by $G\in \mathcal{C}^{m}$
and the $j^{\text{th}}$-derivative of $G$ in $x$ by $G_{x}^{(j)}$ for $1\leq
j\leq m$.

\item $G$ is called $\mathcal{C}^{m,\delta}$-regular if $G$ is $m$-times
continuously differentiable in $x$ with $G_{x}^{(m)}$ being $\delta
$-H\"{o}lder in $x$ for any $t$ a.s. We denote this by $G\in \mathcal{C}%
^{m,\delta}$.

\item $G$ is called $\mathcal{K}_{loc}^{m,\delta}$-regular if $G\in
\mathcal{C}^{m,\delta}$ and, for any compact set $K\subset \mathbb{R}$ (resp.
$\mathbb{R}^{+}$) when $\mathbb{H}=\mathbb{D}$ (resp. $\mathbb{D}^{+}$) and
any $T>0$,
\[
\int_{0}^{T}\Vert G\Vert_{m,\delta:K}(t)\,dt<\infty \text{, a.s.,}%
\]
where the random $K$-seminorm is defined by
\[
\Vert G\Vert_{m,\delta:K}(t):=\sup_{x\in K}\frac{|G(t,x)|}{1+|x|}+\sum
_{j=1}^{m}\sup_{x\in K}|G_{x}^{(j)}(t,x)|+\sup_{\substack{x,y\in K\\x\neq
y}}\frac{|G_{x}^{(m)}(t,x)-G_{x}^{(m)}(t,y)|}{|x-y|^{\delta}}.
\]
We denote this by $G\in \mathcal{K}_{loc}^{m,\delta}$.

\item $G$ is called $\mathcal{\bar{K}}_{loc}^{m,\delta}$-regular if
$G\in \mathcal{C}^{m,\delta}$ and, for any compact set $K\subset \mathbb{R}$
(resp. $\mathbb{R}^{+}$) when $\mathbb{H}=\mathbb{D}$ (resp. $\mathbb{D}^{+}$)
and any $T>0$,
\[
\int_{0}^{T}\Vert G\Vert_{m,\delta:K}^{2}(t)\,dt<\infty \text{, a.s.}%
\]
We denote this by $G\in \mathcal{\bar{K}}_{loc}^{m,\delta}$.
\end{itemize}

Next, for $\mathbb{H}=\mathbb{D},\mathbb{D}^{+}$, let $F(t,x)$, $(t,x)\in
\mathbb{H}$, be a progressive It\^{o} semimartingale random field with
dynamics given by
\begin{equation}
dF(t,x)=\phi(t,x)\,dt+\psi(t,x)^{\top}dW_{t},\quad(t,x)\in \mathbb{H},
\label{10.1}%
\end{equation}
where $W=(W^{1},W^{2})$ is a two-dimensional Brownian motion, $\phi(t,x)$ and
$\psi(t,x)=(\psi^{1}(t,x),\psi^{2}(t,x))^{\top}$ are progressive random fields
valued in $\mathbb{R}$ and $\mathbb{R}^{2}$, respectively. Let $m\geq0$ be an
integer and $\delta \in(0,1]$.

\begin{itemize}
\item $F$ is called a $\mathcal{K}_{loc}^{m,\delta}$-semimartingale if
$F(0,x)$ is $m$-times continuously differentiable in $x$ with $F_{x}%
^{(m)}(0,x)$ being $\delta$-H\"{o}lder in $x$ a.s. and
\[
\int_{0}^{t}\phi(s,x)\,ds\in \mathcal{K}_{loc}^{m,\delta},{\quad}\int_{0}%
^{t}\psi(s,x)^{\top}dW_{s}\in \mathcal{\bar{K}}_{loc}^{m,\delta}.
\]

\end{itemize}

\begin{theorem}
[Differential Rule and It\^{o}-Ventzel Formula]\label{DI}Let $F(t,x)$,
$(t,x)\in \mathbb{H}$, be a progressive It\^{o} semimartingale random field
with dynamics as in (\ref{10.1}).

\begin{enumerate}
\item[(i)] If $F$ is a $\mathcal{K}_{loc}^{m,\delta}$-semimartingale for some
integer $m\geq0$ and $\delta \in(0,1]$, then $(\phi,\psi)\in \mathcal{K}%
_{loc}^{m,\varepsilon}\times \mathcal{\bar{K}}_{loc}^{m,\varepsilon}$ for any
$\varepsilon<\delta$. Conversely, if $(\phi,\psi)\in \mathcal{K}_{loc}%
^{m,\delta}\times \mathcal{\bar{K}}_{loc}^{m,\delta}$ for some $m\geq0$ and
$\delta \in(0,1]$, then $F$ is a $\mathcal{K}_{loc}^{m,\varepsilon}%
$-semimartingale for any $\varepsilon<\delta$.

\item[(ii)] If $F$ is a $\mathcal{K}_{loc}^{m,\delta}$-semimartingale for some
integer $m\geq1$ and $\delta \in(0,1]$, then the differential rule holds:
\[
dF_{x}(t,x)=\phi_{x}(t,x)\,dt+\psi_{x}(t,x)^{\top}dW_{t}.
\]

\item[(iii)] If $F\in \mathcal{C}^{2}$ and it is a $\mathcal{K}_{loc}%
^{1,\delta}$-semimartingale for some $\delta \in(0,1]$, then the
It\^{o}-Ventzel formula holds. Specifically, for any It\^{o} process $X$
valued in $\mathbb{H}$, the process $F(t,X_{t})$, $t\geq0$, is a continuous
It\^{o} semimartingale satisfying
\begin{align}
F(t,X_{t})=  &  \,F(0,X_{0})+\int_{0}^{t}\phi(s,X_{s})\,ds+\int_{0}^{t}%
\psi(s,X_{s})^{\top}dW_{s}\label{IV}\\
&  +\int_{0}^{t}F_{x}(s,X_{s})\,dX_{s}+\frac{1}{2}\int_{0}^{t}F_{xx}%
(s,X_{s})\,d\langle X\rangle_{s}+\int_{0}^{t}\psi_{x}(s,X_{s})^{\top}d\langle
W,X\rangle_{s}.\nonumber
\end{align}

\end{enumerate}
\end{theorem}

\renewcommand{\theequation}{\Alph{subsection}.\arabic{equation}}\setcounter{equation}{0}
\renewcommand{\thetheorem}{\Alph{subsection}.\arabic{theorem}}\setcounter{theorem}{0}

\subsection{\emph{Decoupling field method\label{DF}}}

We present background results that we use to establish the solvability of the
FBSDE derived in Section \ref{Decoupling}. Specifically, we recall the notion
of a decoupling field and some of the general results established in
\cite{FI2020}. To this end, consider a general FBSDE, for $0\leq s\leq T$ and
$x\in \mathbb{R}^{2}$,
\begin{equation}
\left \{
\begin{array}
[c]{l}%
X_{s}=x+%
{\displaystyle \int_{0}^{s}}
\mu \left(  r,X_{r},Y_{r},Z_{r}\right)  dr+%
{\displaystyle \int_{0}^{s}}
\sigma \left(  r,X_{r},Y_{r},Z_{r}\right)  ^{\top}dW_{r},\\
Y_{s}=H\left(  X_{T}\right)  -%
{\displaystyle \int_{s}^{T}}
f\left(  r,X_{r},Y_{r},Z_{r}\right)  dr-%
{\displaystyle \int_{s}^{T}}
Z_{r}^{\top}dW_{r},
\end{array}
\right.  \label{generalFBSDE}%
\end{equation}
where $\mu,\sigma,f:[0,T]\times \mathbb{R}^{2}\times \mathbb{R}\times
\mathbb{R}^{2}\rightarrow \mathbb{R}^{2},\mathbb{R}^{2\times2},\mathbb{R}$ and
$H:\mathbb{R}^{2}\mathbb{\rightarrow R}$ are deterministic measurable functions.

\begin{definition}
Let $T$ be arbitrary and fixed, and $0\leq t\leq T$. A function $w:[t,T]\times
\mathbb{R}^{2}\mathbb{\rightarrow R}$ with $w(T,x)=H(x)$ is called a Markovian
decoupling field for FBSDE (\ref{generalFBSDE}) on $[t,T]$ if, for all
$\,t\leq t_{1}\leq t_{2}\leq T$ and any $\mathcal{F}_{t_{1}}$-measurable
$X_{t_{1}}:\Omega \rightarrow \mathbb{R}^{2}$, there exist progressively
measurable processes $(X,Y,Z)$ on $[t_{1},t_{2}]$ such that $\left \Vert Z
\right \Vert _{\infty}<\infty$ and
\begin{equation}
\left \{
\begin{array}
[c]{l}%
X_{s}=X_{t_{1}}+%
{\displaystyle \int_{t_{1}}^{s}}
\mu \left(  r,X_{r},Y_{r},Z_{r}\right)  dr+%
{\displaystyle \int_{t_{1}}^{s}}
\sigma \left(  r,X_{r},Y_{r},Z_{r}\right)  ^{\top}dW_{r},\\
Y_{s}=Y_{t_{2}}-%
{\displaystyle \int_{s}^{t_{2}}}
f\left(  r,X_{r},Y_{r},Z_{r}\right)  dr-%
{\displaystyle \int_{s}^{t_{2}}}
Z_{r}^{\top}dW_{r},\\
Y_{s}=w\left(  s,X_{s}\right)  ,
\end{array}
\right.  \label{2102}%
\end{equation}
hold a.s. for all $t_{1}\leq s\leq t_{2}$.
\end{definition}

\begin{notation}
For a measurable function $F:\mathbb{R}^{n} \rightarrow \mathbb{R}^{m} $, we
define
\begin{equation}
\label{LC}L_{F,x} := \inf \left \{  L \geq0 : |F(x) - F(x^{\prime})| \leq L |x -
x^{\prime}| \text{ for all } x, x^{\prime}\in \mathbb{R}^{n} \right \} ,
\end{equation}
with the convention $\inf \varnothing:= \infty$. Here, we use $x $ in $L_{F,x}$
to denote the argument of the function $F(x) $.
\end{notation}

\begin{definition}
Let $T$ be arbitrary and fixed, and $0\leq t\leq T$. Let $w$ be a Markovian
decoupling field for FBSDE (\ref{generalFBSDE}) on $[t,T]$.

\begin{enumerate}
\item[(i)] $w$ is weakly regular if $L_{w,x}<\frac{1}{L_{\sigma,z}}$ and
$\sup_{t\leq s\leq T}|w(s,0)|<\infty$.

\item[(ii)] $w$ is strongly regular if it is weakly regular, and for all
$t\leq t_{1}\leq t_{2}\leq T$, the processes $(X,Y,Z)$ in (\ref{2102}) satisfy
the conditions:

\begin{enumerate}
\item[(a)] they are unique for each constant $X_{t_{1}}=x\in \mathbb{R}^{2}$,

\item[(b)] for all $x\in \mathbb{R}^{2}$,
\[
\sup_{t_{1}\leq s\leq t_{2}}\mathbb{E}_{t_{1},\infty}\left[  \left \vert
X_{s}\right \vert ^{2}\right]  +\sup_{t_{1}\leq s\leq t_{2}}\mathbb{E}%
_{t_{1},\infty}\left[  \left \vert Y_{s}\right \vert ^{2}\right]  <\infty
\text{,}%
\]
where $\mathbb{E}_{t_{1},\infty}[\cdot]:=\operatorname*{esssup}\mathbb{E[}%
\cdot|\mathcal{F}_{t_{1}}\mathbb{]}$,

\item[(c)] they are measurable and weakly differentiable with respect to
$x\in \mathbb{R}^{2}$ such that
\[
\operatorname*{esssup}_{x \in \mathbb{R}^{2}} \sup_{\zeta \in S^{1}} \sup_{t_{1}
\leq s \leq t_{2}} \mathbb{E}_{t_{1},\infty} \left[  |\partial_{x}
X_{s}|_{\zeta}^{2} \right]  < \infty,
\]
\[
\operatorname*{esssup}_{x \in \mathbb{R}^{2}} \sup_{\zeta \in S^{1}} \sup_{t_{1}
\leq s \leq t_{2}} \mathbb{E}_{t_{1},\infty} \left[  |\partial_{x}
Y_{s}|_{\zeta}^{2} \right]  < \infty,
\]
\[
\operatorname*{esssup}_{x \in \mathbb{R}^{2}} \sup_{\zeta \in S^{1}}
\mathbb{E}_{t_{1},\infty} \left[  \int_{t_{1}}^{t_{2}} |\partial_{x}
Z_{s}|_{\zeta}^{2} \, ds \right]  < \infty,
\]

where $S^{1} := \{ z \in \mathbb{R}^{2} : |z| = 1 \} $, and $|\partial_{x}
X_{s}|_{\zeta}= |\langle \partial_{x} X_{s}, \zeta \rangle| $ is the magnitude
of the directional derivative in direction $\zeta$.
\end{enumerate}
\end{enumerate}
\end{definition}

For the definition and properties of weak derivatives of processes with
respect to their initial values, we refer to \cite{Fthesis}.

\begin{definition}
Let $w$ be a Markovian decoupling field for FBSDE (\ref{generalFBSDE}). We
call $w$ controlled in $z$ if there exists a constant $C>0$ such that, for all
$t\leq t_{1}\leq t_{2}\leq T$ and all initial values $X_{t_{1}}\in
\mathcal{F}_{t_{1}}$, the corresponding processes $(X,Y,Z)$ in (\ref{2102})
satisfy $||Z||_{\infty}\leq C$.
\end{definition}

\begin{definition}
\label{Def_MLLC} The FBSDE (\ref{generalFBSDE}) satisfies the modified local
Lipschitz condition (MLLC) if

\begin{enumerate}
\item[(i)] $\mu,\sigma,f$ are Lipschitz continuous in $x,y,z$ on
$[0,T]\times \mathbb{R}^{2}\times \mathbb{R}\times B$ with $B\subset
\mathbb{R}^{2}$ being an arbitrary bounded set,

\item[(ii)] $\left \Vert \mu(\cdot,0,0,0)\right \Vert _{\infty},\left \Vert
\sigma(\cdot,0,0,0)\right \Vert _{\infty},\left \Vert f(\cdot,0,0,0)\right \Vert
_{\infty}<\infty$,

\item[(iii)] $L_{\sigma,z}<\infty$ and $L_{H,x}<\frac{1}{L_{\sigma,z}}$.
\end{enumerate}
\end{definition}

\begin{notation}
For the FBSDE (\ref{generalFBSDE}), let $I_{\max}^{M}\subset \lbrack0,T]$ be
the union of all intervals $[t,T]\subset \lbrack0,T]$, $0\leq t\leq T$, such
that there exists a weakly regular Markovian decoupling field $w$ on $[t,T]$.
\end{notation}

\begin{theorem}
\label{MLLC}Let FBSDE (\ref{generalFBSDE}) satisfy MLLC. Then, there exists a
unique weakly regular Markovian decoupling field $u$ on $I_{\max}^{M}$, which
is strongly regular, controlled in $z$ and continuous. Furthermore, it is
either the case $I_{\max}^{M}=[0,T]$ or $I_{\max}^{M}=(t_{\min}^{M},T]$, with
$0\leq t_{\min}^{M}<T$.
\end{theorem}

To prove the well-posedness of the FBSDE \eqref{FBSDEsystem}, we recall the
abstract result in \cite[Section 3]{FI2020}. For $\varepsilon> 0 $, consider
the FBSDE
\begin{equation}
\left \{
\begin{array}
[c]{l}%
\tilde{X}_{s}=\tilde{x}+%
{\displaystyle \int_{t}^{s}}
\dfrac{1}{\varepsilon}\tilde{\mu}\left(  r,\varepsilon \tilde{X}_{r}\right)
dr+\dfrac{1}{\varepsilon}%
{\displaystyle \int_{t}^{s}}
\left(  \rho dW_{r}^{1}+\sqrt{1-\rho^{2}}dW_{r}^{2}\right)  ,\, \rho
\in(0,1),\\
\bar{X}_{s}=\bar{x}+%
{\displaystyle \int_{t}^{s}}
\left(  \bar{\sigma}\left(  r,\varepsilon \tilde{X}_{r},\bar{X}_{r}+Y_{r}%
,Z_{r}^{1}\right)  \right)  dW_{r}^{1},\\
Y_{s}=H\left(  \varepsilon \tilde{X}_{T},\bar{X}_{T}\right)  -%
{\displaystyle \int_{s}^{T}}
f\left(  r,\varepsilon \tilde{X}_{r},\bar{X}_{r}+Y_{r},Z_{r}^{2}\right)  dr-%
{\displaystyle \int_{s}^{T}}
Z_{r}^{1}dW_{r}^{1}-%
{\displaystyle \int_{s}^{T}}
Z_{r}^{2}dW_{r}^{2},
\end{array}
\right.  \label{abMLLC}%
\end{equation}
with $0\leq t\leq s\leq T$, $\tilde{x},\bar{x}\in \mathbb{R}$, and the
coefficients satisfying

\begin{enumerate}
\item[(i)] $\tilde{\mu}:[0,T]\times \mathbb{R\rightarrow R}$ is measurable and
Lipschitz continuous in the second component, and $\left \Vert \tilde{\mu
}(\cdot,0)\right \Vert _{\infty}<\infty$,

\item[(ii)] $\bar{\sigma}:[0,T]\times \mathbb{R\times \mathbb{R}\times
\mathbb{R}\rightarrow R}$ is measurable and Lipschitz continuous in the last
three components, and $\left \Vert \bar{\sigma}(\cdot,\cdot,\cdot,0)\right \Vert
_{\infty}<\infty$,

\item[(iii)] $f:[0,T]\times \mathbb{R\times \mathbb{R}\times \mathbb{R}%
\rightarrow R}$ is measurable and Lipschitz continuous in the last three
components on all sets of the form $[0,T]\times \mathbb{R\times \mathbb{R}%
\times}B$ with $B\subset \mathbb{R}$ being an arbitrary bounded set, and
$\left \Vert f(\cdot,0,0,0)\right \Vert _{\infty}<\infty$,

\item[(iv)] $H:\mathbb{R\times \mathbb{R\rightarrow R}}$ is Lipschitz
continuous in both components, and $L_{H,\bar{x}}<\frac{1}{L_{\bar{\sigma
},z^{1}}}.$
\end{enumerate}

As shown in \cite{FI2020}, under the above conditions, the FBSDE
(\ref{abMLLC}) satisfies the MLLC for some $\varepsilon>0$ small enough.

Finally, to establish the global well-posedness of (\ref{abMLLC}), we also
introduce additional assumptions:

\begin{enumerate}
\item[(v)] $f$ is differentiable in $(\tilde{x},\bar{x}+y,z^{2})$, $\left \Vert
f(\cdot,\cdot,\cdot,0)\right \Vert _{\infty}<\infty$ and there exists some
constant $C>0$ such that
\[
\left \vert f_{\tilde{x}}\right \vert \leq C\left(  1+\left \vert z^{2}%
\right \vert ^{2}\right)  ,\quad \left \vert f_{\bar{x}+y}\right \vert \leq
C\left(  1+\left \vert z^{2}\right \vert ^{2}\right) ,\quad \text{and}%
\quad \left \vert f_{z}\right \vert \leq C\left(  1+\left \vert z^{2}\right \vert
\right)  ,
\]

\item[(vi)] $\bar{\sigma}$ is differentiable in $(\tilde{x},\bar{x}+y,z^{1})$
with bounded derivatives and $L_{\bar{\sigma},z^{1}}\leq1$,

\item[(vii)] $\left \Vert H\right \Vert _{\infty}<\infty$, and $L_{H,\bar{x}%
}<1.$
\end{enumerate}

\begin{theorem}
[{\cite[Theorem 3.3]{FI2020}}]\label{3.3}Let assumptions (i)-(vii) hold and,
furthermore, assume that there exist constants $K\in \lbrack0,\frac{1}%
{L_{\bar{\sigma},z^{1}}})$ and $\varepsilon_{0}\in(0,\infty)$ such that, for
any $\varepsilon \in(0,\varepsilon_{0}]$ and any $0\leq t\leq T$, if a weakly
regular Markovian decoupling field $w^{\varepsilon}:[t,T]\times \mathbb{R\times
R\rightarrow R}$ associated with FBSDE (\ref{abMLLC}) exists, it also
satisfies $\sup_{t\leq s\leq T}\left \Vert w_{\bar{x}}^{\varepsilon}\left(
s,\cdot,\cdot \right)  \right \Vert _{\infty}<K$. Then, there exists
$\varepsilon>0$ such that $I_{\max}^{M}=[0,T]$ for FBSDE (\ref{abMLLC}).
\end{theorem}

\begin{remark}
\label{ReZ}We note that herein we assume that $|f_{\tilde{x}}|\leq
C(1+|z^{2}|^{2})$ which is weaker than the assumption in Theorem 3.3 in
\cite{FI2020}. However, this condition does not affect the results of Theorem
\ref{3.3} due to the uniform boundedness of the $Z$-component. More precisely,
in the proof of Theorem 3.3 in \cite{FI2020}, $f_{\tilde{x}}$ only appears in
the representation of the derivative process $\tilde{R}$ (see p.2676) and is
used to establish the uniform bound for $\tilde{R}$ (see the middle of
p.2677). On the other hand, $Z$ is uniformly bounded by the definition of
decoupling field and this bound is independent of both $t_{1}$ and
$\varepsilon$, by the argument in the middle of p.2685. Therefore, the
estimate of $\tilde{R}$ will not change and, hence, Theorem \ref{3.3} still
holds under this weaker assumption.
\end{remark}

\renewcommand{\theequation}{\Alph{subsection}.\arabic{equation}}\setcounter{equation}{0}
\renewcommand{\thetheorem}{\Alph{subsection}.\arabic{theorem}}\setcounter{theorem}{0}

\subsection{Proofs}

\subsubsection{Proof of Lemma \ref{XZmart}}

Let $X_{s}^{t,\pi}:=\int_{t}^{s}\pi_{r}\left(  \theta_{r}dr+dW_{r}^{1}\right)
,\, t\leq s\leq T$. It\^{o}'s formula yields that $X^{t,\pi}M^{t,q}$ is a
local martingale, given by
\[
d\left(  X_{s}^{t,\pi}M_{s}^{t,q}\right)  =\left(  -X_{s}^{t,\pi}M_{s}%
^{t,q}\theta_{s}+M_{s}^{t,q}\pi_{s}\right)  dW_{s}^{1}-X_{s}^{t,\pi}%
M_{s}^{t,q}q_{s}dW_{s}^{2}.
\]
We observe that for $n>2$, and using that $\pi \in \mathbb{L}_{BMO}^{2}[0,T]$,
the energy inequality implies that
\begin{equation}
\mathbb{E}\left[  \left(  \int_{0}^{T}\left \vert \pi_{r}\right \vert
^{2}dr\right)  ^{\frac{n}{2}}\right]  \leq(\frac{n}{2})!||\int_{0}^{\cdot}%
\pi_{r}dW_{r}^{1}||_{BMO}^{n}<\infty. \label{2105}%
\end{equation}
In turn, using the Burkholder-Davis-Gundy inequality and the uniform
boundedness of $\theta$, we deduce
\begin{align}
\mathbb{E}\left[  \sup_{t\leq s\leq T}\left \vert X_{s}^{t,\pi}\right \vert
^{n}\right]  \leq &  \ C\mathbb{E}\left[  \sup_{t\leq s\leq T}\left(
\left \vert \int_{t}^{s}\pi_{r}\theta_{r}dr\right \vert ^{n}+\left \vert \int
_{t}^{s}\pi_{r}dW_{r}^{1}\right \vert ^{n}\right)  \right] \nonumber \\
\leq &  \ C\mathbb{E}\left[  \left(  \int_{t}^{T}\left \vert \pi_{r}\theta
_{r}\right \vert dr\right)  ^{n}\right]  +C\mathbb{E}\left[  \left(  \int
_{t}^{T}\left \vert \pi_{r}\right \vert ^{2}dr\right)  ^{\frac{n}{2}}\right]
<\infty. \label{X}%
\end{align}

Next, we show that $\int_{t}^{s}X_{r}^{t,\pi}M_{r}^{t,q}q_{r}dW_{r}^{2}$,
$t\leq s\leq T,$ is a true $\mathbb{F}$-martingale. For this, we first observe
that the stochastic exponential $M^{t,q}$ is a uniformly integrable martingale
due to the boundedness of $\theta$ and the fact that $q\in \mathbb{L}_{BMO}%
^{2}[t,T]$. Using Doob's inequality and the reverse H\"{o}lder's inequality,
we deduce%
\[
\mathbb{E}\left[  \sup_{t\leq s\leq T}\left \vert M_{s}^{t,q}\right \vert
^{p_{0}}\right]  \leq \left(  \frac{p_{0}}{p_{0}-1}\right)  ^{p_{0}}%
\mathbb{E}\left[  |M_{T}^{t,q}|^{p_{0}}\right]  <\infty,
\]
for some $p_{0}>1$. Using once more the Burkholder-Davis-Gundy inequality
gives
\begin{align*}
\mathbb{E}\left[  \sup_{t\leq s\leq T}\left \vert \int_{t}^{s}X_{r}^{t,\pi
}M_{r}^{t,q}q_{r}dW_{r}^{2}\right \vert \right]  \leq &  \ C\mathbb{E}\left[
\sup_{t\leq s\leq T}\left \vert M_{s}^{t,q}\right \vert \left(  \int_{t}%
^{T}\left \vert X_{r}^{t,\pi}q_{r}\right \vert ^{2}dr\right)  ^{\frac{1}{2}%
}\right] \\
\leq &  \ C\mathbb{E}\left[  \sup_{t\leq s\leq T}\left \vert M_{s}%
^{t,q}\right \vert ^{p_{0}}\right]  +C\mathbb{E}\left[  \left(  \int_{t}%
^{T}\left \vert X_{r}^{t,\pi}q_{r}\right \vert ^{2}dr\right)  ^{\frac{1}{2}%
\frac{p_{0}}{p_{0}-1}}\right]  .
\end{align*}
Moreover, inequality (\ref{X}) and that $q\in \mathbb{L}_{BMO}^{2}[t,T]$ yield
\begin{align*}
\mathbb{E}\left[  \left(  \int_{t}^{T}\left \vert X_{r}^{t,\pi}q_{r}\right \vert
^{2}dr\right)  ^{\frac{1}{2}\frac{p_{0}}{p_{0}-1}}\right]  \leq &
\  \mathbb{E}\left[  \sup_{t\leq s\leq T}\left \vert X_{s}^{t,\pi}\right \vert
^{\frac{p_{0}}{p_{0}-1}}\left(  \int_{t}^{T}\left \vert q_{r}\right \vert
^{2}dr\right)  ^{\frac{1}{2}\frac{p_{0}}{p_{0}-1}}\right] \\
\leq &  \  \mathbb{E}\left[  \sup_{t\leq s\leq T}\left \vert X_{s}^{t,\pi
}\right \vert ^{\frac{2p_{0}}{p_{0}-1}}\right]  +\mathbb{E}\left[  \left(
\int_{t}^{T}\left \vert q_{r}\right \vert ^{2}dr\right)  ^{\frac{p_{0}}{p_{0}%
-1}}\right]  <\infty,
\end{align*}
and we conclude. Using similar arguments, we deduce that $\int_{t}^{s}%
(-X_{r}^{t,\pi}M_{r}^{t,q}\theta_{r}+M_{r}^{t,q}\pi_{r})dW_{r}^{1}$, $t\leq
s\leq T$, is also a true $\mathbb{F}$-martingale.

\subsubsection{Proof of Theorem \ref{FBSDE1}}

In this proof, we omit the superscript $P$ for the primal FBSDE for notational simplicity.

\textbf{Part (i).} \underline{Step 1}. Since $\pi^{\ast,P}$ is optimal, we
must have, for any bounded $\pi \in \mathcal{A}_{[t,T]}$ and $\varepsilon
\in(0,1)$, that a.s.,
\begin{equation}
\lim_{\varepsilon \rightarrow0}\frac{1}{\varepsilon}\mathbb{E}\left[  \left.
U\left(  T,X_{T}^{\varepsilon,\pi}+P\right)  -U\left(  T,X_{T}^{\ast
}+P\right)  \right \vert \mathcal{F}_{t}\right]  \leq0, \label{0101}%
\end{equation}
where $X_{T}^{\varepsilon,\pi}:=\xi+\int_{t}^{T}(\pi_{r}^{\ast,P}%
+\varepsilon \pi_{r})(\theta_{r}dr+dW_{r}^{1})$ and $X_{T}^{\ast}$ as in
(\ref{SDE}). By the fundamental theorem of calculus and the chain rule, we can
write
\[
\frac{1}{\varepsilon}\left[  U\left(  T,X_{T}^{\varepsilon,\pi}+P\right)
-U\left(  T,X_{T}^{\ast}+P\right)  \right]  =\int_{0}^{1}U_{x}\left(
T,X_{T}^{\delta \varepsilon,\pi}+P\right)  d\delta \int_{t}^{T}\pi_{r}\left(
\theta_{r}dr+dW_{r}^{1}\right) .
\]
Substituting this into the left-hand side of (\ref{0101}) and applying the
uniform integrability assumption along with the dominated convergence theorem,
we obtain
\begin{align*}
0  &  \geq \lim_{\varepsilon \rightarrow0}\mathbb{E}\left[  \left.  \int_{0}%
^{1}U_{x}\left(  T,X_{T}^{\delta \varepsilon,\pi}+P\right)  d\delta \int_{t}%
^{T}\pi_{r}\left(  \theta_{r}dr+dW_{r}^{1}\right)  \right \vert \mathcal{F}%
_{t}\right] \\
&  =\mathbb{E}\left[  \left.  U_{x}\left(  T,X_{T}^{\ast}+P\right)  \int
_{t}^{T}\pi_{r}\left(  \theta_{r}dr+dW_{r}^{1}\right)  \right \vert
\mathcal{F}_{t}\right]  .
\end{align*}
Replacing $\pi$ by $-\pi$ implies that $0\leq \mathbb{E}[U_{x}\left(
T,X_{T}^{\ast}+P\right)  \int_{t}^{T}\pi_{r}\left(  \theta_{r}dr+dW_{r}%
^{1}\right)  |\mathcal{F}_{t}]$, and, thus,%
\begin{equation}
0=\mathbb{E}\left[  \left.  U_{x}\left(  T,X_{T}^{\ast}+P\right)  \int_{t}%
^{T}\pi_{r}\left(  \theta_{r}dr+dW_{r}^{1}\right)  \right \vert \mathcal{F}%
_{t}\right]  =\mathbb{E}\left[  \left.  N_{T}A_{T}^{\pi}\right \vert
\mathcal{F}_{t}\right]  , \label{NA}%
\end{equation}
where, for $0\leq t\leq s\leq T$,%
\begin{equation}
N_{s}:=\mathbb{E}\left[  \left.  U_{x}\left(  T,X_{T}^{\ast}+P\right)
\right \vert \mathcal{F}_{s}\right]  \quad \text{ and }\quad A_{s}^{\pi}%
:=\int_{t}^{s}\pi_{r}\left(  \theta_{r}dr+dW_{r}^{1}\right)  . \label{N}%
\end{equation}
By the martingale representation theorem, there exists a square integrable
density process, say $\varphi=(\varphi^{1},\varphi^{2})$ on $[t,T],$ such that%
\begin{equation}
dN_{s}=\varphi_{s}^{\top}dW_{s},\quad0\leq t\leq s\leq T. \label{Nr}%
\end{equation}

\underline{Step 2}. We define an $\mathbb{F}$-progressively measurable process
$Y$ on $[t,T]$ such that%
\begin{equation}
U_{x}\left(  s,X_{s}^{\ast}+Y_{s}\right)  =N_{s}=\mathbb{E}\left[  \left.
U_{x}\left(  T,X_{T}^{\ast}+P\right)  \right \vert \mathcal{F}_{s}\right]
,\quad0\leq t\leq s\leq T. \label{NY}%
\end{equation}
In other words, $Y_{s}:=U_{x}^{-1}\left(  s,N_{s}\right)  -X_{s}^{\ast}\text{,
}t\leq s\leq T\text{, and }Y_{T}=P.$ From the dual relation (\ref{dr1}), we
have that%
\begin{equation}
Y_{s}=-\tilde{U}_{z}\left(  s,N_{s}\right)  -X_{s}^{\ast},\quad t\leq s\leq T.
\label{PY}%
\end{equation}
In turn, applying the It\^{o}-Ventzel formula to $\tilde{U}_{z}$ in
(\ref{DUSPDE}) with $N$ as in (\ref{Nr}), we obtain%
\begin{align}
dY_{s}=  &  \left(  -\tilde{\beta}_{z}\left(  s,N_{s}\right)  -\frac{1}%
{2}\tilde{U}_{zzz}\left(  s,N_{s}\right)  \left \vert \varphi_{s}\right \vert
^{2}-\tilde{\alpha}_{zz}^{\top}\left(  s,N_{s}\right)  \varphi_{s}-\pi
_{s}^{\ast,P}\theta_{s}\right)  ds\nonumber \\
&  +\left(  -\tilde{\alpha}_{z}^{1}\left(  s,N_{s}\right)  -\tilde{U}%
_{zz}\left(  s,N_{s}\right)  \varphi_{s}^{1}-\pi_{s}^{\ast,P}\right)
dW_{s}^{1}+\left(  -\tilde{\alpha}_{z}^{2}\left(  s,N_{s}\right)  -\tilde
{U}_{zz}\left(  s,N_{s}\right)  \varphi_{s}^{2}\right)  dW_{s}^{2}.
\label{2601}%
\end{align}
Recalling (\ref{dr5}) and (\ref{dr5'}), we have%
\begin{equation}
\tilde{U}_{zz}\left(  s,N_{s}\right)  =-\frac{1}{U_{xx}\left(  s,X_{s}^{\ast
}+Y_{s}\right)  } \quad \text{ and } \quad \tilde{U}_{zzz}\left(  s,N_{s}%
\right)  =\frac{U_{xxx}\left(  s,X_{s}^{\ast}+Y_{s}\right)  }{\left(
U_{xx}\left(  s,X_{s}^{\ast}+Y_{s}\right)  \right)  ^{3}}. \label{2604}%
\end{equation}
We also have by (\ref{DDA}) that%
\begin{align}
\tilde{\alpha}_{z}\left(  s,N_{s}\right)   &  =\frac{\alpha_{x}\left(
s,X_{s}^{\ast}+Y_{s}\right)  }{U_{xx}\left(  s,X_{s}^{\ast}+Y_{s}\right)
},\label{2605}\\
\tilde{\alpha}_{zz}\left(  s,N_{s}\right)   &  =-\frac{U_{xxx}\left(
s,X_{s}^{\ast}+Y_{s}\right)  \alpha_{x}\left(  s,X_{s}^{\ast}+Y_{s}\right)
}{\left(  U_{xx}\left(  s,X_{s}^{\ast}+Y_{s}\right)  \right)  ^{3}}%
+\frac{\alpha_{xx}\left(  s,X_{s}^{\ast}+Y_{s}\right)  }{\left \vert
U_{xx}\left(  s,X_{s}^{\ast}+Y_{s}\right)  \right \vert ^{2}}, \label{2606}%
\end{align}
and by (\ref{B2}) that%
\begin{align}
\tilde{\beta}_{z}\left(  s,N_{s}\right)  =  &  \frac{\beta_{x}\left(
s,X_{s}^{\ast}+Y_{s}\right)  }{U_{xx}\left(  s,X_{s}^{\ast}+Y_{s}\right)
}+\frac{1}{2}\frac{U_{xxx}\left(  s,X_{s}^{\ast}+Y_{s}\right)  \left \vert
\alpha_{x}\left(  s,X_{s}^{\ast}+Y_{s}\right)  \right \vert ^{2}}{\left(
U_{xx}\left(  s,X_{s}^{\ast}+Y_{s}\right)  \right)  ^{3}}\nonumber \\
&  -\frac{\alpha_{x}^{\top}\left(  s,X_{s}^{\ast}+Y_{s}\right)  \alpha
_{xx}\left(  s,X_{s}^{\ast}+Y_{s}\right)  }{\left \vert U_{xx}\left(
s,X_{s}^{\ast}+Y_{s}\right)  \right \vert ^{2}}, \label{2607}%
\end{align}
with $\beta_{x}(t,x)$ given in (\ref{DB}). Combining (\ref{2601})-(\ref{2607})
yields that%
\begin{equation}
dY_{s}=-f\left(  s,X_{s}^{\ast},Y_{s},Z_{s}\right)  ds+Z_{s}^{\top}dW_{s},
\quad t\leq s\leq T, \label{2104}%
\end{equation}
where%
\begin{align}
f\left(  s,x,y,z\right)  :=  &  \frac{\beta_{x}\left(  s,x+y\right)  }%
{U_{xx}\left(  s,x+y\right)  }+\frac{1}{2}\frac{U_{xxx}\left(  s,x+y\right)
\left \vert \varphi_{s}-\alpha_{x}\left(  s,x+y\right)  \right \vert ^{2}%
}{\left(  U_{xx}\left(  s,x+y\right)  \right)  ^{3}}\label{2103}\\
&  +\frac{\alpha_{xx}^{\top}\left(  s,x+y\right)  \left(  \varphi_{s}%
-\alpha_{x}\left(  s,x+y\right)  \right)  }{\left \vert U_{xx}\left(
s,x+y\right)  \right \vert ^{2}}+\pi_{s}^{\ast,P}\theta_{s}\nonumber
\end{align}
and
\begin{equation}
Z_{s}^{1}:=\dfrac{\varphi_{s}^{1}-\alpha_{x}^{1}\left(  s,X_{s}^{\ast}%
+Y_{s}\right)  }{U_{xx}\left(  s,X_{s}^{\ast}+Y_{s}\right)  }-\pi_{s}^{\ast,P}
\quad \text{ and}\quad Z_{s}^{2}:=\dfrac{\varphi_{s}^{2}-\alpha_{x}^{2}\left(
s,X_{s}^{\ast}+Y_{s}\right)  }{U_{xx}\left(  s,X_{s}^{\ast}+Y_{s}\right)  }.
\label{PZ}%
\end{equation}
Thus, the processes $\varphi=(\varphi^{1},\varphi^{2})$ in (\ref{Nr}) must
satisfy
\begin{equation}%
\begin{array}
[c]{l}%
\varphi_{s}^{1}=\alpha_{x}^{1}\left(  s,X_{s}^{\ast}+Y_{s}\right)
+U_{xx}\left(  s,X_{s}^{\ast}+Y_{s}\right)  \pi_{s}^{\ast,P}+U_{xx}\left(
s,X_{s}^{\ast}+Y_{s}\right)  Z_{s}^{1},\\
\varphi_{s}^{2}=\alpha_{x}^{2}\left(  s,X_{s}^{\ast}+Y_{s}\right)
+U_{xx}\left(  s,X_{s}^{\ast}+Y_{s}\right)  Z_{s}^{2}.
\end{array}
\label{j2}%
\end{equation}
Therefore, substituting the representation of $\varphi$ into (\ref{2103}), we
obtain that the drift $f$ in BSDE (\ref{2104}) is given by
\begin{align}
f\left(  s,x,y,z\right)  =  &  \  \frac{\beta_{x}\left(  s,x+y\right)  }%
{U_{xx}\left(  s,x+y\right)  }+\frac{1}{2}\frac{U_{xxx}\left(  s,x+y\right)
\left(  \left \vert z^{1}+\pi_{s}^{\ast,P}\right \vert ^{2}+\left \vert
z^{2}\right \vert ^{2}\right)  }{U_{xx}\left(  s,x+y\right)  }\nonumber \\
&  +\frac{\alpha_{xx}^{1}\left(  s,x+y\right)  \left(  z^{1}+\pi_{s}^{\ast
,P}\right)  +\alpha_{xx}^{2}\left(  s,x+y\right)  z^{2}}{U_{xx}\left(
s,x+y\right)  }+\pi_{s}^{\ast,P}\theta_{s}. \label{F1}%
\end{align}

\underline{Step 3}. Recalling $N$ in (\ref{Nr}) with $\varphi$ as in
(\ref{j2}) and $A^{\pi}$ in (\ref{N}), we obtain
\begin{align*}
d(N_{s}A_{s}^{\pi})=  &  \  \pi_{s}\left(  N_{s}\theta_{s}+\alpha_{x}%
^{1}\left(  s,X_{s}^{\ast}+Y_{s}\right)  +U_{xx}\left(  s,X_{s}^{\ast}%
+Y_{s}\right)  \pi_{s}^{\ast,P}+U_{xx}\left(  s,X_{s}^{\ast}+Y_{s}\right)
Z_{s}^{1}\right)  ds\\
&  +\left(  A_{s}^{\pi}\varphi_{s}^{1}+N_{s}\pi_{s}\right)  dW_{s}^{1}%
+A_{s}^{\pi}\varphi_{s}^{2}dW_{s}^{2}.
\end{align*}
By Lemma \ref{IM}, which, to ease the presentation, is provided right after
this proof, we have that $0\leq \int_{t}^{s}\left(  A_{r}^{\pi}\varphi_{r}%
^{1}+N_{r}\pi_{r}\right)  dW_{r}^{1}+\int_{t}^{s}A_{r}^{\pi}\varphi_{r}%
^{2}dW_{r}^{2}$, $0\leq t\leq s\leq T$, is a true $\mathbb{F}$-martingale.
Then, integrating both sides from $t$ to $T$ and taking conditional
expectation, we obtain using (\ref{NA}) that%
\[
\mathbb{E}\left[  \left.  \int_{t}^{T}\pi_{s}\left(  N_{s}\theta_{s}%
+\alpha_{x}^{1}\left(  s,X_{s}^{\ast}+Y_{s}\right)  +U_{xx}\left(
s,X_{s}^{\ast}+Y_{s}\right)  \pi_{s}^{\ast,P}+U_{xx}\left(  s,X_{s}^{\ast
}+Y_{s}\right)  Z_{s}^{1}\right)  ds\right \vert \mathcal{F}_{t}\right]  =0.
\]
Since $\pi$ is an arbitrary admissible control process, we deduce that, for
$0\leq t\leq s\leq T$, it must be that at the optimum
\begin{equation}
\pi_{s}^{\ast,P}=-\frac{N_{s}\theta_{s}+\alpha_{x}^{1}\left(  s,X_{s}^{\ast
}+Y_{s}\right)  }{U_{xx}\left(  s,X_{s}^{\ast}+Y_{s}\right)  }-Z_{s}%
^{1}=-\frac{U_{x}\left(  s,X_{s}^{\ast}+Y_{s}\right)  \theta_{s}+\alpha
_{x}^{1}\left(  s,X_{s}^{\ast}+Y_{s}\right)  }{U_{xx}\left(  s,X_{s}^{\ast
}+Y_{s}\right)  }-Z_{s}^{1}. \label{pistar}%
\end{equation}

\underline{Step 4}. Combining (\ref{F1}), (\ref{DB}) and (\ref{pistar})
yields
\[
f\left(  s,x,y,z\right)  =-z^{1}\theta_{s}+\frac{\frac{1}{2}U_{xxx}\left(
s,x+y\right)  \left \vert z^{2}\right \vert ^{2}+\alpha_{xx}^{2}\left(
s,x+y\right)  z^{2}}{U_{xx}\left(  s,x+y\right)  }.
\]
In turn, we easily deduce the BSDE
\[
dY_{s}=-f\left(  s,X_{s}^{\ast},Y_{s},Z_{s}\right)  ds+Z_{s}^{\top}%
dW_{s},\quad Y_{T}=P.
\]
Meanwhile, the process $X^{\ast} $, defined in (\ref{WS}) and associated with
the optimal control $\pi^{\ast,P} $ from (\ref{pistar}), provides the forward
component of the FBSDE (\ref{FBSDEsystem}). \medskip

\noindent \textbf{Part (ii).} If $(X,Y,Z)$ satisfies FBSDE (\ref{FBSDEsystem}),
it follows easily that $\pi^{\ast,P}\in \mathcal{A}_{[t,T]}$ using Assumption
\ref{Assump} (ii) and the fact that $Z^{1}\in \mathbb{L}_{BMO}^{2}[t,T]$.
Moreover, we have that
\[
X_{s}=\xi+\int_{t}^{s}\pi_{r}^{\ast,P}\left(  \theta_{r}dr+dW_{r}^{1}\right)
=X_{s}^{\ast},\quad0\leq t\leq s\leq T.
\]

For any $\pi \in \mathcal{A}_{[t,T]}$, denote $X_{s}^{\pi}:=\xi+\int_{t}^{s}%
\pi_{r}\left(  \theta_{r}dr+dW_{r}^{1}\right)  $. Next, we show that, for any
$\pi \in \mathcal{A}_{[t,T]}$,%
\begin{equation}
\mathbb{E}\left[  \left.  U\left(  T,X_{T}^{\pi}+P\right)  \right \vert
\mathcal{F}_{t}\right]  \leq \mathbb{E}\left[  \left.  U\left(  T,X_{T}%
+P\right)  \right \vert \mathcal{F}_{t}\right] , \label{0201}%
\end{equation}
where $X$ is given by the forward component of the FBSDE (\ref{FBSDEsystem}).
To this end, working as in Lemma \ref{NM}, we introduce a probability measure
\[
\left.  \frac{d\mathbb{Q}^{\ast}}{d\mathbb{P}}\right \vert _{\mathcal{F}_{T}%
}=M_{T}^{t,q^{\ast,P}}=\frac{U_{x}\left(  T,X_{T}+P\right)  }{U_{x}\left(
t,\xi+Y_{t}\right)  },
\]
where $M^{t,q^{\ast,P}}$ is the optimal state price density process. Then,
Girsanov's theorem yields that $\left(  dW_{s}^{\ast,1},dW_{s}^{\ast
,2}\right)  :=\left(  dW_{s}^{1}+\theta_{s}ds,dW_{s}^{2}+q_{s}^{\ast
,P}ds\right)  $ is a Brownian motion under $\mathbb{Q}^{\ast}$. Using the
spatial concavity of $U$, we deduce
\begin{align}
&  \mathbb{E}\left[  \left.  U\left(  T,X_{T}^{\pi}+P\right)  \right \vert
\mathcal{F}_{t}\right]  -\mathbb{E}\left[  \left.  U\left(  T,X_{T}+P\right)
\right \vert \mathcal{F}_{t}\right] \nonumber \\
\leq &  \mathbb{E}\left[  \left.  U_{x}\left(  T,X_{T}+P\right)  \left(
X_{T}^{\pi}-X_{T}\right)  \right \vert \mathcal{F}_{t}\right] \nonumber \\
=  &  U_{x}\left(  t,\xi+Y_{t}\right)  \mathbb{E}^{\mathbb{Q}^{\ast}}\left[
\left.  \int_{t}^{T}\left(  \pi_{r}-\pi_{r}^{\ast,P}\right)  dW_{r}^{\ast
,1}\right \vert \mathcal{F}_{t}\right]  . \label{TT1}%
\end{align}
Applying the reverse H\"older inequality to $M^{t, q^{\ast,P}} $, we obtain
that there exists a constant $p_{0} > 1 $ such that $\mathbb{E}\left[  \left|
M_{T}^{t, q^{\ast,P}} \right| ^{p_{0}} \right]  < \infty. $ Moreover, noting
that $\pi- \pi^{\ast,P} \in \mathbb{L}_{\mathrm{BMO}}^{2}[t,T] $ and invoking
the energy inequality (\ref{2105}), we deduce that
\[
\mathbb{E}\left[  \left(  \int_{t}^{T} \left|  \pi_{r} - \pi_{r}^{\ast,P}
\right| ^{2} dr \right) ^{\frac{n}{2}} \right]  < \infty \quad \text{for all } n
> 2.
\]

In turn, the Burkholder-Davis-Gundy inequality yields
\begin{align*}
\mathbb{E}^{\mathbb{Q}^{\ast}}\left[  \sup_{t\leq s\leq T}\left \vert \int
_{t}^{s}\left(  \pi_{r}-\pi_{r}^{\ast,P}\right)  dW_{r}^{\ast,1}\right \vert
\right]  \leq &  \ C\mathbb{E}^{\mathbb{Q}^{\ast}}\left[  \left(  \int_{t}%
^{T}\left \vert \pi_{r}-\pi_{r}^{\ast,P}\right \vert ^{2}dr\right)  ^{\frac
{1}{2}}\right] \\
\leq &  \ C\mathbb{E}\left[  \left \vert M_{T}^{t,q^{\ast,P}}\right \vert
^{p_{0}}\right]  +C\mathbb{E}\left[  \left(  \int_{t}^{T}\left \vert \pi
_{r}-\pi_{r}^{\ast,P}\right \vert ^{2}dr\right)  ^{\frac{1}{2}\frac{p_{0}%
}{p_{0}-1}}\right]  <\infty,
\end{align*}
which further implies that the right hand side of (\ref{TT1}) is zero and
(\ref{0201}) follows.

\begin{lemma}
\label{IM}For any $\pi \in \mathcal{A}_{[t,T]}$, the process
\[
\int_{t}^{s}\left(  A_{r}^{\pi}\varphi_{r}^{1}+N_{r}\pi_{r}\right)  dW_{r}%
^{1}+\int_{t}^{s}A_{r}^{\pi}\varphi_{r}^{2}dW_{r}^{2},\text{\quad}t\leq s\leq
T\text{,}%
\]
is a true $\mathbb{F}$-martingale on $[t,T]$.
\end{lemma}

\begin{proof}
By (\ref{X}), we know that $\mathbb{E}[\sup_{t\leq s\leq T}|A_{s}^{\pi}%
|^{2}]<\infty$. The Burkholder-Davis-Gundy inequality yields%
\begin{align*}
\mathbb{E}\left[  \sup_{t\leq s\leq T}\left \vert \int_{t}^{s}A_{r}^{\pi
}\varphi_{r}^{i}dW_{r}^{i}\right \vert \right]   &  \leq C\mathbb{E}\left[
\left(  \int_{t}^{T}\left \vert A_{r}^{\pi}\varphi_{r}^{i}\right \vert
^{2}dr\right)  ^{\frac{1}{2}}\right] \\
&  \leq C\mathbb{E}\left[  \sup_{t\leq s\leq T}\left \vert A_{s}^{\pi
}\right \vert ^{2}\right]  +C\mathbb{E}\left[  \int_{t}^{T}\left \vert
\varphi_{r}^{i}\right \vert ^{2}dr\right]  <\infty.
\end{align*}
Moreover, by the assumption of square integrability of $U_{x}(T,X_{T}^{\ast
}+P)$ and Doob's $L^{p}$-inequality, we deduce that%
\[
\mathbb{E}\left[  \sup_{t\leq s\leq T}\left \vert N_{s}\right \vert ^{2}\right]
\leq C\mathbb{E}\left[  \left \vert U_{x}\left(  T,X_{T}^{\ast}+P\right)
\right \vert ^{2}\right]  <\infty.
\]
Using similar arguments, we obtain that $\mathbb{E}[\sup_{t\leq s\leq T}%
|\int_{t}^{s}N_{r}\pi_{r}dW_{r}^{1}|]<\infty$.
\end{proof}

\subsubsection{Proof of Theorem \ref{DO1}}

\textbf{Part (i).} By Lemma \ref{NM}, we have that $q^{\ast,P}\in
\mathcal{Q}_{[t,T]}$. Next, we show that, for any $q\in \mathcal{Q}_{[t,T]}$,%
\begin{equation}
\mathbb{E}\left[  \left.  \tilde{U}\left(  T,\hat{\eta}M_{T}^{t,q^{\ast,P}%
}\right)  +\hat{\eta}M_{T}^{t,q^{\ast,P}}P\right \vert \mathcal{F}_{t}\right]
\leq \mathbb{E}\left[  \left.  \tilde{U}\left(  T,\hat{\eta}M_{T}^{t,q}\right)
+\hat{\eta}M_{T}^{t,q}P\right \vert \mathcal{F}_{t}\right]  . \label{Ob1}%
\end{equation}
Note that by Lemma \ref{NM} and the definition of $\hat{\eta}$ in (\ref{eta}),%
\begin{equation}
\hat{\eta}M_{T}^{t,q^{\ast,P}}=U_{x}\left(  T,X_{T}^{P}+P\right)  . \label{Z*}%
\end{equation}
From part (ii) in Theorem \ref{FBSDE1}, we have that $X_{T}^{P}=\xi+\int
_{t}^{T}\pi_{r}^{\ast,P}(\theta_{r}dr+dW_{r}^{1})$ with $\pi^{\ast,P}%
\in \mathcal{A}_{[t,T]}$ as in (\ref{Opi}). It follows from (\ref{dr0}),
(\ref{Z*}) and Lemma \ref{XZmart} that%
\begin{align}
\mathbb{E}\left[  \left.  \tilde{U}\left(  T,\hat{\eta}M_{T}^{t,q^{\ast,P}%
}\right)  +\hat{\eta}M_{T}^{t,q^{\ast,P}}P\right \vert \mathcal{F}_{t}\right]
=  &  \  \mathbb{E}\left[  \left.  \tilde{U}\left(  T,U_{x}\left(  T,X_{T}%
^{P}+P\right)  \right)  +U_{x}\left(  T,X_{T}^{P}+P\right)  P\right \vert
\mathcal{F}_{t}\right] \nonumber \\
=  &  \  \mathbb{E}\left[  \left.  U\left(  T,X_{T}^{P}+P\right)  -U_{x}\left(
T,X_{T}^{P}+P\right)  X_{T}^{P}\right \vert \mathcal{F}_{t}\right] \nonumber \\
=  &  \  \mathbb{E}\left[  \left.  U\left(  T,X_{T}^{P}+P\right)  \right \vert
\mathcal{F}_{t}\right]  -\xi \hat{\eta}. \label{0301}%
\end{align}
On the other hand, for any $q\in \mathcal{Q}_{[t,T]}$, by (\ref{dr}) and Lemma
\ref{XZmart}, we easily deduce that%
\begin{align*}
\mathbb{E}\left[  \left.  \tilde{U}\left(  T,\hat{\eta}M_{T}^{t,q}\right)
+\hat{\eta}M_{T}^{t,q}P\right \vert \mathcal{F}_{t}\right]  \geq &
\  \mathbb{E}\left[  \left.  U\left(  T,X_{T}^{P}+P\right)  -\hat{\eta}%
M_{T}^{t,q}X_{T}^{P}\right \vert \mathcal{F}_{t}\right] \\
=  &  \  \mathbb{E}\left[  \left.  U\left(  T,X_{T}^{P}+P\right)  \right \vert
\mathcal{F}_{t}\right]  -\xi \hat{\eta},
\end{align*}
which together with (\ref{0301}) proves (\ref{Ob1}).

\textbf{Part (ii).} From the bidual relation (\ref{bidr}) and Lemma
\ref{XZmart}, we have, for any $\pi \in \mathcal{A}_{[t,T]}$, $q\in
\mathcal{Q}_{[t,T]}$, $\xi \in \cap_{p\geq1}L^{p}(\mathcal{F}_{t})$ and $\eta \in
L^{0,+}(\mathcal{F}_{t})$,
\begin{align}
&  \mathbb{E}\left[  \left.  U\left(  T,\xi+\int_{t}^{T}\pi_{r}\left(
\theta_{r}dr+dW_{r}^{1}\right)  +P\right)  \right \vert \mathcal{F}_{t}\right]
\nonumber \\
\leq &  \  \mathbb{E}\left[  \left.  \tilde{U}\left(  T,\eta M_{T}%
^{t,q}\right)  +\eta M_{T}^{t,q}\left(  \xi+\int_{t}^{T}\pi_{r}\left(
\theta_{r}dr+dW_{r}^{1}\right)  +P\right)  \right \vert \mathcal{F}_{t}\right]
\nonumber \\
=  &  \  \mathbb{E}\left[  \left.  \tilde{U}\left(  T,\eta M_{T}^{t,q}\right)
+\eta M_{T}^{t,q}P\right \vert \mathcal{F}_{t}\right]  +\xi \eta. \label{0401}%
\end{align}
Thus, due to the arbitrariness of $\pi,q$ and $\eta$, we obtain
\begin{equation}
u^{P}\left(  t,\xi;T\right)  \leq \operatorname*{essinf}_{\eta \in
L^{0,+}\left(  \mathcal{F}_{t}\right)  }\left(  \tilde{u}^{P}\left(
t,\eta;T\right)  +\xi \eta \right)  . \label{0402}%
\end{equation}
On the other hand, for $q^{\ast,P}$ defined in (\ref{qstar}), using Theorem
\ref{FBSDE1} (ii), equality (\ref{0301}) and part (i) of this theorem gives%
\begin{align}
u^{P}\left(  t,\xi;T\right)   &  =\mathbb{E}\left[  \left.  U\left(
T,X_{T}^{P}+P\right)  \right \vert \mathcal{F}_{t}\right] \nonumber \\
&  =\mathbb{E}\left[  \left.  \tilde{U}\left(  T,\hat{\eta}M_{T}^{t,q^{\ast
,P}}\right)  +\hat{\eta}M_{T}^{t,q^{\ast,P}}P\right \vert \mathcal{F}%
_{t}\right]  +\xi \hat{\eta}=\tilde{u}^{P}\left(  t,\hat{\eta};T\right)
+\xi \hat{\eta}, \label{xh}%
\end{align}
which together with (\ref{0402}) proves the desired result.

\subsubsection{Proof of Theorem \ref{DFBSDE1}}

In this proof, we omit the superscript $P$ for the dual FBSDE for notational simplicity.

\textbf{Part (i).} \underline{Step 1}. Since $q^{\ast,P}$ is optimal, for any
bounded $q\in \mathcal{Q}_{[t,T]}$ and $\varepsilon \in(0,1)$, we have
\begin{equation}
\lim_{\varepsilon \rightarrow0}\frac{1}{\varepsilon}\mathbb{E}\left[  \left.
\left(  \tilde{U}\left(  T,M_{T}^{\varepsilon,q}\right)  +M_{T}^{\varepsilon
,q}P\right)  -\left(  \tilde{U}\left(  T,M_{T}^{\ast}\right)  +M_{T}^{\ast
}P\right)  \right \vert \mathcal{F}_{t}\right]  \geq0, \label{2001}%
\end{equation}
where $M^{\varepsilon,q}:=\eta M^{t,q^{\ast,P}+\varepsilon q}$ and $M^{\ast
}:=\eta M^{t,q^{\ast,P}}$. Note that%
\[
\frac{1}{\varepsilon}\left(  \left(  \tilde{U}\left(  T,M_{T}^{\varepsilon
,q}\right)  +M_{T}^{\varepsilon,q}P\right)  -\left(  \tilde{U}\left(
T,M_{T}^{\ast}\right)  +M_{T}^{\ast}P\right)  \right)  =\frac{1}{\varepsilon
}\int_{0}^{1}\left(  \tilde{U}_{z}\left(  T,M_{T}^{\delta \varepsilon
,q}\right)  +P\right)  \partial_{\delta}M_{T}^{\delta \varepsilon,q}d\delta
\]
and%
\[
\partial_{\delta}M_{T}^{\delta \varepsilon,q}=M_{T}^{\delta \varepsilon
,q}\left(  -\int_{t}^{T}\varepsilon q_{r}dW_{r}^{2}-\int_{t}^{T}\left(
q_{r}^{\ast,P}+\delta \varepsilon q_{r}\right)  \varepsilon q_{r}dr\right)  .
\]
Therefore, by the uniformly integrability assumption and dominated
convergence, (\ref{2001}) implies that%
\begin{align*}
0  &  \leq \lim_{\varepsilon \rightarrow0}\mathbb{E}\left[  \left.  \int_{0}%
^{1}\left(  \tilde{U}_{z}\left(  T,M_{T}^{\delta \varepsilon,q}\right)
+P\right)  M_{T}^{\delta \varepsilon,q}\left(  -\int_{t}^{T}q_{r}dW_{r}%
^{2}-\int_{t}^{T}\left(  q_{r}^{\ast,P}+\delta \varepsilon q_{r}\right)
q_{r}dr\right)  d\delta \right \vert \mathcal{F}_{t}\right] \\
&  =\mathbb{E}\left[  \left.  \left(  \tilde{U}_{z}\left(  T,M_{T}^{\ast
}\right)  +P\right)  M_{T}^{\ast}\left(  -\int_{t}^{T}q_{r}dW_{r}^{2}-\int
_{t}^{T}q_{r}^{\ast,P}q_{r}dr\right)  \right \vert \mathcal{F}_{t}\right]  .
\end{align*}
Replacing $q$ by $-q$, we obtain that%
\begin{equation}
0=\mathbb{E}\left[  \left.  \left(  \tilde{U}_{z}\left(  T,M_{T}^{\ast
}\right)  +P\right)  M_{T}^{\ast}\left(  \int_{t}^{T}q_{r}dW_{r}^{2}+\int
_{t}^{T}q_{r}^{\ast,P}q_{r}dr\right)  \right \vert \mathcal{F}_{t}\right]
=\mathbb{E}\left[  \left.  \Gamma_{T}H_{T}^{q}\right \vert \mathcal{F}%
_{t}\right]  , \label{MH0}%
\end{equation}
where, for $0\leq t\leq s\leq T$, the processes $\Gamma$ and $H^{q}$ are
defined as%
\begin{equation}
\Gamma_{s}:=\mathbb{E}\left[  \left.  \left(  \tilde{U}_{z}\left(
T,M_{T}^{\ast}\right)  +P\right)  M_{T}^{\ast}\right \vert \mathcal{F}%
_{s}\right]  \quad \text{and}\quad H_{s}^{q}:=\int_{t}^{s}q_{r}dW_{r}^{2}%
+\int_{t}^{s}q_{r}^{\ast,P}q_{r}dr. \label{M}%
\end{equation}
From the martingale representation theorem, there exists a square integrable
density process $\psi=(\psi^{1},\psi^{2})$ on $[t,T]$ such that%
\begin{equation}
d\Gamma_{s}=\psi_{s}^{\top}dW_{s},\quad0\leq t\leq s\leq T. \label{Mr}%
\end{equation}

\underline{Step 2}. We introduce an $\mathbb{F}$-progressively measurable
process $\tilde{Y}$ on $[t,T]$ such that%
\[
\left(  \tilde{U}_{z}\left(  s,M_{s}^{\ast}\right)  +\tilde{Y}_{s}\right)
M_{s}^{\ast}=\Gamma_{s}=\mathbb{E}\left[  \left.  \left(  \tilde{U}_{z}\left(
T,M_{T}^{\ast}\right)  +P\right)  M_{T}^{\ast}\right \vert \mathcal{F}%
_{s}\right]  ,\quad t\leq s\leq T.
\]
In other words, $\tilde{Y}_{s}:=\frac{\Gamma_{s}}{M_{s}^{\ast}}-\tilde{U}%
_{z}\left(  s,M_{s}^{\ast}\right)  ,\ t\leq s\leq T, \text{ and } \tilde
{Y}_{T}=P.$ Applying It\^{o}'s formula to $\tilde{U}_{z}$ in (\ref{DUSPDE}),
$\Gamma$ in (\ref{Mr}) and to $M^{\ast}$, we obtain%
\begin{align}
d\tilde{Y}_{s}=  &  \left(  \frac{\Gamma_{s}}{M_{s}^{\ast}}\left(  \left \vert
\theta_{s}\right \vert ^{2}+\left \vert q_{s}^{\ast,P}\right \vert ^{2}\right)
+\frac{\psi_{s}^{1}\theta_{s}+\psi_{s}^{2}q_{s}^{\ast,P}}{M_{s}^{\ast}}%
-\tilde{\alpha}_{z}^{1}\left(  s,M_{s}^{\ast}\right)  \theta_{s}-\frac
{\tilde{\alpha}_{z}^{2}\left(  s,M_{s}^{\ast}\right)  \tilde{\alpha}_{zz}%
^{2}\left(  s,M_{s}^{\ast}\right)  }{\tilde{U}_{zz}\left(  s,M_{s}^{\ast
}\right)  }\right. \nonumber \\
&  +\tilde{U}_{zz}\left(  s,M_{s}^{\ast}\right)  M_{s}^{\ast}\left \vert
\theta_{s}\right \vert ^{2}+\frac{1}{2}\tilde{U}_{zzz}\left(  s,M_{s}^{\ast
}\right)  \frac{\left \vert \tilde{\alpha}_{z}^{2}\left(  s,M_{s}^{\ast
}\right)  \right \vert ^{2}}{\left \vert \tilde{U}_{zz}\left(  s,M_{s}^{\ast
}\right)  \right \vert ^{2}}-\frac{1}{2}\tilde{U}_{zzz}\left(  s,M_{s}^{\ast
}\right)  \left \vert M_{s}^{\ast}q_{s}^{\ast,P}\right \vert ^{2}\nonumber \\
&  \bigg.+\tilde{\alpha}_{zz}^{2}\left(  s,M_{s}^{\ast}\right)  M_{s}^{\ast
}q_{s}^{\ast,P}\bigg)ds+\left(  \frac{\Gamma_{s}\theta_{s}+\psi_{s}^{1}}%
{M_{s}^{\ast}}-\tilde{\alpha}_{z}^{1}\left(  s,M_{s}^{\ast}\right)  +\tilde
{U}_{zz}\left(  s,M_{s}^{\ast}\right)  M_{s}^{\ast}\theta_{s}\right)
dW_{s}^{1}\nonumber \\
&  +\left(  \frac{\Gamma_{s}q_{s}^{\ast,P}+\psi_{s}^{2}}{M_{s}^{\ast}}%
-\tilde{\alpha}_{z}^{2}\left(  s,M_{s}^{\ast}\right)  +\tilde{U}_{zz}\left(
s,M_{s}^{\ast}\right)  M_{s}^{\ast}q_{s}^{\ast,P}\right)  dW_{s}^{2}.
\label{Ytil}%
\end{align}
Next, let, for $0\leq t\leq s\leq T$,%
\begin{equation}%
\begin{array}
[c]{l}%
\tilde{Z}_{s}^{1}:=\dfrac{\Gamma_{s}\theta_{s}+\psi_{s}^{1}}{M_{s}^{\ast}%
}-\tilde{\alpha}_{z}^{1}\left(  s,M_{s}^{\ast}\right)  +\tilde{U}_{zz}\left(
s,M_{s}^{\ast}\right)  M_{s}^{\ast}\theta_{s},\\
\tilde{Z}_{s}^{2}:=\dfrac{\Gamma_{s}q_{s}^{\ast,P}+\psi_{s}^{2}}{M_{s}^{\ast}%
}-\tilde{\alpha}_{z}^{2}\left(  s,M_{s}^{\ast}\right)  +\tilde{U}_{zz}\left(
s,M_{s}^{\ast}\right)  M_{s}^{\ast}q_{s}^{\ast,P}.
\end{array}
\label{DZ}%
\end{equation}
It, then, follows that%
\begin{equation}%
\begin{array}
[c]{l}%
\psi_{s}^{1}=\left(  \tilde{Z}_{s}^{1}+\tilde{\alpha}_{z}^{1}\left(
s,M_{s}^{\ast}\right)  -\tilde{U}_{zz}\left(  s,M_{s}^{\ast}\right)
M_{s}^{\ast}\theta_{s}\right)  M_{s}^{\ast}-\Gamma_{s}\theta_{s},\\
\psi_{s}^{2}=\left(  \tilde{Z}_{s}^{2}+\tilde{\alpha}_{z}^{2}\left(
s,M_{s}^{\ast}\right)  -\tilde{U}_{zz}\left(  s,M_{s}^{\ast}\right)
M_{s}^{\ast}q_{s}^{\ast,P}\right)  M_{s}^{\ast}-\Gamma_{s}q_{s}^{\ast,P}.
\end{array}
\label{psi}%
\end{equation}
Substituting (\ref{psi}) into (\ref{Ytil}) yields%
\begin{equation}
d\tilde{Y}_{s}=-\tilde{f}\left(  s,M_{s}^{\ast},\tilde{Y}_{s},\tilde{Z}%
_{s}\right)  ds+\tilde{Z}_{s}^{\top}dW_{s}, \quad0\leq t\leq s\leq T,
\label{yy}%
\end{equation}
where%
\begin{align}
\tilde{f}\left(  s,d,y,z\right)  :=  &  -z^{1}\theta_{s}-z^{2}q_{s}^{\ast
,P}-\tilde{\alpha}_{z}^{2}\left(  s,d\right)  q_{s}^{\ast,P}+\tilde{U}%
_{zz}\left(  s,d\right)  d\left \vert q_{s}^{\ast,P}\right \vert ^{2}%
+\frac{\tilde{\alpha}_{z}^{2}\left(  s,d\right)  \tilde{\alpha}_{zz}%
^{2}\left(  s,d\right)  }{\tilde{U}_{zz}\left(  s,d\right)  }\nonumber \\
&  -\frac{1}{2}\tilde{U}_{zzz}\left(  s,d\right)  \frac{\left \vert
\tilde{\alpha}_{z}^{2}\left(  s,d\right)  \right \vert ^{2}}{\left \vert
\tilde{U}_{zz}\left(  s,d\right)  \right \vert ^{2}}+\frac{1}{2}\tilde{U}%
_{zzz}\left(  s,d\right)  \left \vert dq_{s}^{\ast,P}\right \vert ^{2}%
-\tilde{\alpha}_{zz}^{2}\left(  s,d\right)  dq_{s}^{\ast,P}. \label{ff}%
\end{align}

\underline{Step 3}. Applying It\^{o}'s formula to $\Gamma$ in (\ref{Mr}) with
$\psi$ as in (\ref{psi}) and $H^{q}$ as in (\ref{M}), we obtain that, for any
bounded $q\in \mathcal{Q}_{[t,T]}$,
\[
d\left(  \Gamma_{s}H_{s}^{q}\right)  =\ q_{s}\left(  \tilde{Z}_{s}^{2}%
+\tilde{\alpha}_{z}^{2}\left(  s,M_{s}^{\ast}\right)  -\tilde{U}_{zz}\left(
s,M_{s}^{\ast}\right)  M_{s}^{\ast}q_{s}^{\ast,P}\right)  M_{s}^{\ast}%
ds+H_{s}^{q}\psi_{s}^{1}dW_{s}^{1}+\left(  \Gamma_{s}q_{s}+H_{s}^{q}\psi
_{s}^{2}\right)  dW_{s}^{2}.
\]

From the definition of $H^{q}$ in (\ref{M}) and the Burkholder-Davis-Gundy
inequality, we have that $\mathbb{E}[\sup_{t\leq s\leq T}|H_{s}^{q}%
|^{2}]<\infty$. In turn, applying similar arguments to the ones used to
establish Lemma \ref{IM}, we deduce that $\int_{t}^{s}H_{r}^{q}\psi_{r}%
^{1}dW_{r}^{1}+\int_{t}^{s}\left(  \Gamma_{r}q_{r}+H_{r}^{q}\psi_{r}%
^{2}\right)  dW_{r}^{2}$, $0\leq t\leq s\leq T$, is a true $\mathbb{F}$-martingale.

Thus, (\ref{MH0}) gives%
\[
\mathbb{E}\left[  \left.  \int_{t}^{T}q_{s}\left(  \tilde{Z}_{s}^{2}%
+\tilde{\alpha}_{z}^{2}\left(  s,M_{s}^{\ast}\right)  -\tilde{U}_{zz}\left(
s,M_{s}^{\ast}\right)  M_{s}^{\ast}q_{s}^{\ast,P}\right)  M_{s}^{\ast
}ds\right \vert \mathcal{F}_{t}\right]  =0,
\]
which together with the arbitrariness of $q$ implies that, for $0\leq t\leq
s\leq T$,%
\begin{equation}
q_{s}^{\ast,P}=\frac{\tilde{Z}_{s}^{2}+\tilde{\alpha}_{z}^{2}\left(
s,M_{s}^{\ast}\right)  }{\tilde{U}_{zz}\left(  s,M_{s}^{\ast}\right)
M_{s}^{\ast}}. \label{qq}%
\end{equation}

\underline{Step 4}. Combining (\ref{qq}) and (\ref{ff}) yields that $\tilde
{f}$ in (\ref{ff}) is given by
\[
\tilde{f}\left(  s,d,y,z\right)  =-z^{1}\theta_{s}+\frac{1}{2}\tilde{U}%
_{zzz}\left(  s,d\right)  \frac{\left \vert z^{2}\right \vert ^{2}}{\left \vert
\tilde{U}_{zz}\left(  s,d\right)  \right \vert ^{2}}+\tilde{U}_{zzz}\left(
s,d\right)  \frac{\tilde{\alpha}_{z}^{2}\left(  s,d\right)  z^{2}}{\left \vert
\tilde{U}_{zz}\left(  s,d\right)  \right \vert ^{2}}-\frac{\tilde{\alpha}%
_{zz}^{2}\left(  s,d\right)  z^{2}}{\tilde{U}_{zz}\left(  s,d\right)  }.
\]
Thus,
\[
d\tilde{Y}_{s}=-\tilde{f}\left(  s,D_{s},\tilde{Y}_{s},\tilde{Z}_{s}\right)
ds+\tilde{Z}_{s}^{\top}dW_{s},\quad \tilde{Y}_{T}=P\text{,}%
\]
and the process $M^{\ast} = \eta M^{t, q^{\ast,P}} $, with $q^{\ast,P} $ given
in (\ref{qq}), provides the forward component of the FBSDE (\ref{DFBSDEsystem}).

\textbf{Part (ii).} Using the duality relations (\ref{dr1}) and (\ref{dr5}),
we rewrite $q^{\ast,P}$ as
\[
q_{s}^{\ast,P}=-\frac{U_{xx}\left(  s,-\tilde{U}_{z}\left(  s,D_{s}\right)
\right)  }{U_{x}\left(  s,-\tilde{U}_{z}\left(  s,D_{s}\right)  \right)
}\left(  \frac{\alpha_{x}^{2}\left(  s,-\tilde{U}_{z}\left(  s,D_{s}\right)
\right)  }{U_{xx}\left(  s,-\tilde{U}_{z}\left(  s,D_{s}\right)  \right)
}+\tilde{Z}_{s}^{2}\right)  ,
\]
which, together with Assumption \ref{Assump} (ii) yields that $q^{\ast,P}%
\in \mathcal{Q}_{[t,T]}$.

Next, we show that, for any $q\in \mathcal{Q}_{[t,T]}$,
\[
\mathbb{E}\left[  \left.  \tilde{U}\left(  T,D_{T}\right)  +D_{T}P\right \vert
\mathcal{F}_{t}\right]  \leq \mathbb{E}\left[  \left.  \tilde{U}\left(  T,\eta
M_{T}^{t,q}\right)  +\eta M_{T}^{t,q}P\right \vert \mathcal{F}_{t}\right]  .
\]
Indeed, from the definition of $q^{\ast,P}$, we have that $\eta M^{t,q^{\ast
,P}}=D$. Furthermore, from Lemma \ref{UZ} and Lemma \ref{XZmart}, we have that
$(\tilde{U}_{z}(s,D_{s})+\tilde{Y}_{s})D_{s}$ is a true $\mathbb{F}%
$-martingale. Thus, (\ref{dr0}) implies
\begin{align}
\mathbb{E}\left[  \left.  \tilde{U}\left(  T,D_{T}\right)  +D_{T}P\right \vert
\mathcal{F}_{t}\right]   &  =\mathbb{E}\left[  \left.  U\left(  T,-\tilde
{U}_{z}\left(  T,D_{T}\right)  \right)  +D_{T}\left(  \tilde{U}_{z}\left(
T,D_{T}\right)  +P\right)  \right \vert \mathcal{F}_{t}\right] \nonumber \\
&  =\mathbb{E}\left[  \left.  U\left(  T,-\tilde{U}_{z}\left(  T,D_{T}\right)
\right)  \right \vert \mathcal{F}_{t}\right]  +\eta \left(  \tilde{U}_{z}\left(
t,\eta \right)  +\tilde{Y}_{t}\right)  . \label{1901}%
\end{align}
On the other hand, from (\ref{dr}) and (\ref{UZpi}), we have
\begin{align}
&  \mathbb{E}\left[  \left.  \tilde{U}\left(  T,\eta M_{T}^{t,q}\right)  +\eta
M_{T}^{t,q}P\right \vert \mathcal{F}_{t}\right] \nonumber \\
\geq &  \  \mathbb{E}\left[  \left.  U\left(  T,-\tilde{U}_{z}\left(
T,D_{T}\right)  \right)  +\eta M_{T}^{t,q}\tilde{U}_{z}\left(  T,D_{T}\right)
+\eta M_{T}^{t,q}P\right \vert \mathcal{F}_{t}\right] \nonumber \\
=  &  \  \mathbb{E}\left[  \left.  U\left(  T,-\tilde{U}_{z}\left(
T,D_{T}\right)  \right)  +\eta M_{T}^{t,q}\left(  \tilde{U}_{z}\left(
t,\eta \right)  +\tilde{Y}_{t}-\int_{t}^{T}\pi_{r}^{\ast,P}\left(  \theta
_{r}dr+dW_{r}^{1}\right)  \right)  \right \vert \mathcal{F}_{t}\right]
\nonumber \\
=  &  \  \mathbb{E}\left[  \left.  U\left(  T,-\tilde{U}_{z}\left(
T,D_{T}\right)  \right)  \right \vert \mathcal{F}_{t}\right]  +\eta \left(
\tilde{U}_{z}\left(  t,\eta \right)  +\tilde{Y}_{t}\right)  , \label{1902}%
\end{align}
where, in the last equality, we used Lemma~\ref{XZmart} with $\pi$ replaced by
$\pi^{\ast,P} $, and $q \in \mathbb{L}_{\mathrm{BMO}}^{2}[t,T] $.  Combining
(\ref{1901}) and (\ref{1902}) we conclude.

\subsubsection{Proof of Theorem \ref{DOP1}}

\textbf{Part (i).} From Lemma \ref{UZ} we have that $\pi^{\ast,P}%
\in \mathcal{A}_{[t,T]}$ and
\[
\mathbb{E}\left[  \left.  U\left(  T,\hat{\xi}+\int_{t}^{T}\pi_{r}^{\ast
,P}\left(  \theta_{r}dr+dW_{r}^{1}\right)  +P\right)  \right \vert
\mathcal{F}_{t}\right]  =\mathbb{E}\left[  \left.  U\left(  T,-\tilde{U}%
_{z}\left(  T,D_{T}^{P}\right)  \right)  \right \vert \mathcal{F}_{t}\right]
.
\]
Next, we verify that, for any $\pi \in \mathcal{A}_{[t,T]}$,
\begin{equation}
\mathbb{E}\left[  \left.  U\left(  T,\hat{\xi}+\int_{t}^{T}\pi_{r}\left(
\theta_{r}dr+dW_{r}^{1}\right)  +P\right)  \right \vert \mathcal{F}_{t}\right]
\leq \mathbb{E}\left[  \left.  U\left(  T,\hat{\xi}+\int_{t}^{T}\pi_{r}%
^{\ast,P}\left(  \theta_{r}dr+dW_{r}^{1}\right)  +P\right)  \right \vert
\mathcal{F}_{t}\right]  . \label{1905}%
\end{equation}
To this end, by the concavity of $U$, the equality (\ref{UZpi}) and the
duality equalities (\ref{dr1}), we obtain that%
\begin{align}
&  \mathbb{E}\left[  \left.  U\left(  T,\hat{\xi}+\int_{t}^{T}\pi_{r}\left(
\theta_{r}dr+dW_{r}^{1}\right)  +P\right)  \right \vert \mathcal{F}_{t}\right]
-\mathbb{E}\left[  \left.  U\left(  T,\hat{\xi}+\int_{t}^{T}\pi_{r}^{\ast
,P}\left(  \theta_{r}dr+dW_{r}^{1}\right)  +P\right)  \right \vert
\mathcal{F}_{t}\right] \nonumber \\
\leq &  \  \mathbb{E}\left[  \left.  U_{x}\left(  T,\hat{\xi}+\int_{t}^{T}%
\pi_{r}^{\ast,P}\left(  \theta_{r}dr+dW_{r}^{1}\right)  +P\right)  \int
_{t}^{T}\left(  \pi_{r}-\pi_{r}^{\ast,P}\right)  \left(  \theta_{r}%
dr+dW_{r}^{1}\right)  \right \vert \mathcal{F}_{t}\right] \nonumber \\
=  &  \  \mathbb{E}\left[  \left.  U_{x}\left(  T,-\tilde{U}_{z}\left(
T,D_{T}^{P}\right)  \right)  \int_{t}^{T}\left(  \pi_{r}-\pi_{r}^{\ast
,P}\right)  \left(  \theta_{r}dr+dW_{r}^{1}\right)  \right \vert \mathcal{F}%
_{t}\right] \nonumber \\
=  &  \  \mathbb{E}\left[  \left.  D_{T}^{P}\int_{t}^{T}\left(  \pi_{r}-\pi
_{r}^{\ast,P}\right)  \left(  \theta_{r}dr+dW_{r}^{1}\right)  \right \vert
\mathcal{F}_{t}\right]  . \label{1904}%
\end{align}
Note that, by Theorem \ref{DFBSDE1} (ii), we have $D^{P}=\eta M^{t,q^{\ast,P}%
}$ with $q^{\ast,P}\in \mathcal{Q}_{[t,T]}$, and $\pi-\pi^{\ast,P}%
\in \mathcal{A}_{[t,T]}$. Thus, Lemma \ref{XZmart} gives $\mathbb{E}[D_{T}^{P}%
{\textstyle \int_{t}^{T}}
\left(  \pi_{r}-\pi_{r}^{\ast,P}\right)  \left(  \theta_{r}dr+dW_{r}%
^{1}\right)  |\mathcal{F}_{t}]=0$, which, together with (\ref{1904}),
establishes (\ref{1905}).

\textbf{Part (ii).} From the inequality (\ref{0401}), we obtain that
\[
\tilde{u}^{P}\left(  t,\eta;T\right)  \geq \operatorname*{esssup}_{\xi \in
\cap_{p\geq1}L^{p}(\mathcal{F}_{t})}(u^{P}\left(  t,\xi;T\right)  -\xi \eta).
\]
On the other hand, using (i) of the theorem and the equalities (\ref{1901}),
we have that
\[
u^{P}\left(  t,\hat{\xi};T\right)  =\mathbb{E}\left[  \left.  U\left(
T,-\tilde{U}_{z}\left(  T,D_{T}^{P}\right)  \right)  \right \vert
\mathcal{F}_{t}\right]  =\mathbb{E}\left[  \left.  \tilde{U}\left(
T,D_{T}^{P}\right)  +D_{T}^{P}P\right \vert \mathcal{F}_{t}\right]  -\eta
\hat{\xi},
\]
which, combined with Theorem \ref{DFBSDE1} (ii), yields that $u^{P}(t,\hat
{\xi};T)=\tilde{u}^{P}(t,\eta;T)+\hat{\xi}\eta.$

\subsubsection{Proof of Proposition \ref{MI}}

The key idea behind the proof of the maturity independence of the value
function $\tilde{u}^{P} $ is to exploit the self-generation property of
$\tilde{U} $.

For $0\leq t\leq T\leq T^{\prime}$ and any $q\in \mathcal{Q}_{[t,T^{\prime}]}$,
we have
\[
M_{T^{\prime}}^{t,q}=M_{T}^{t,q}M_{T^{\prime}}^{T,q}\quad \text{and}%
\quad \mathbb{E}\mathcal{[}M_{T^{\prime}}^{T,q}|\mathcal{F}_{T}]=1,
\]
since $q\in \mathbb{L}_{BMO}^{2}[t,T^{\prime}]$. Then, from (\ref{CP}) and the
tower property of conditional expectation, we obtain%
\begin{align}
\tilde{u}^{P}\left(  t,\eta;T^{\prime}\right)  =  &  \  \operatorname*{essinf}%
_{q\in \mathcal{Q}_{[t,T^{\prime}]}}\mathbb{E}\left[  \left.  \mathbb{E}\left[
\left.  \tilde{U}\left(  T^{\prime},\eta M_{T}^{t,q}M_{T^{\prime}}%
^{T,q}\right)  +\eta M_{T}^{t,q}M_{T^{\prime}}^{T,q}P\right \vert
\mathcal{F}_{T}\right]  \right \vert \mathcal{F}_{t}\right] \nonumber \\
=  &  \  \operatorname*{essinf}_{q\in \mathcal{Q}_{[t,T^{\prime}]}}%
\mathbb{E}\left[  \left.  \mathbb{E}\left[  \left.  \tilde{U}\left(
T^{\prime},\eta M_{T}^{t,q}M_{T^{\prime}}^{T,q}\right)  \right \vert
\mathcal{F}_{T}\right]  +\eta M_{T}^{t,q}P\right \vert \mathcal{F}_{t}\right]
. \label{0501}%
\end{align}
On the other hand, from Proposition \ref{SG}, we have that, for any
$q\in \mathcal{Q}_{\left[  T,T^{\prime}\right]  }$,
\[
\mathbb{E}\left[  \left.  \tilde{U}\left(  T^{\prime},\eta M_{T}%
^{t,q}M_{T^{\prime}}^{T,q}\right)  \right \vert \mathcal{F}_{T}\right]
\geq \tilde{U}\left(  T,\eta M_{T}^{t,q}\right)  ,
\]
and, furthermore, that there exists $q^{\ast}\in \mathcal{Q}[T,T^{\prime}]$
such that%
\[
\tilde{U}\left(  t,\eta M_{T}^{t,q}\right)  =\mathbb{E}\left[  \left.
\tilde{U}\left(  T^{\prime},\eta M_{T}^{t,q}M_{T^{\prime}}^{T,q^{\ast}%
}\right)  \right \vert \mathcal{F}_{T}\right]  .
\]
It, then, follows that%
\[
\tilde{u}^{P}\left(  t,\eta;T^{\prime}\right)  \geq \operatorname*{essinf}%
_{q\in \mathcal{Q}_{[t,T]}}\mathbb{E}\left[  \left.  \tilde{U}\left(  T,\eta
M_{T}^{t,q}\right)  +\eta M_{T}^{t,q}P\right \vert \mathcal{F}_{t}\right]  ,
\]
and%
\begin{align*}
\tilde{u}^{P}\left(  t,\eta;T^{\prime}\right)   &  \leq \operatorname*{essinf}%
_{q\in \mathcal{Q}_{[t,T]}}\mathbb{E}\left[  \left.  \mathbb{E}\left[  \left.
\tilde{U}\left(  T^{\prime},\eta M_{T}^{t,q}M_{T^{\prime}}^{T,q^{\ast}%
}\right)  \right \vert \mathcal{F}_{T}\right]  +\eta M_{T}^{t,q}P\right \vert
\mathcal{F}_{t}\right] \\
&  =\operatorname*{essinf}_{q\in \mathcal{Q}_{[t,T]}}\mathbb{E}\left[  \left.
\tilde{U}\left(  t,\eta M_{T}^{t,q}\right)  +\eta M_{T}^{t,q}P\right \vert
\mathcal{F}_{t}\right]  .
\end{align*}
Combining the above yields $\tilde{u}^{P}(t,\eta;T^{\prime})=\tilde{u}%
^{P}(t,\eta;T)$.

To show the maturity independence of the value function $u^{P}$, we first
observe that FBSDE (\ref{FBSDEsystem}) on $[t,T]$ with initial-terminal
condition $(\xi,P)$ can be extended to $[t,T^{\prime}]$. Indeed, if
$(X^{P},Y^{P},Z^{P})$ is a solution to FBSDE (\ref{FBSDEsystem}) on $[t,T]$
with initial-terminal condition $(\xi,P)$, then $(\bar{X}^{P},\bar{Y}^{P}%
,\bar{Z}^{P})$ defined by
\begin{align*}
\bar{X}_{s}^{P}  &  :=X_{s}^{P}\mathbf{1}_{\left[  t,T\right]  }%
(s)+X_{s}^{T^{\prime},P}\mathbf{1}_{\left(  T,T^{\prime}\right]  }(s),\\
\bar{Y}_{s}^{P}  &  :=Y_{s}^{P}\mathbf{1}_{\left[  t,T\right]  }%
(s)+P\mathbf{1}_{\left(  T,T^{\prime}\right]  }(s),\\
\bar{Z}_{s}^{P}  &  :=Z_{s}^{P}\mathbf{1}_{\left[  t,T\right]  }%
(s)+0\cdot \mathbf{1}_{\left(  T,T^{\prime}\right]  }(s),
\end{align*}
on $0\leq t\leq s\leq T^{\prime}$, is a solution to FBSDE (\ref{FBSDEsystem})
on $[t,T^{\prime}]$ with initial-terminal condition $(\xi,P)$, where
$X_{s}^{T^{\prime},P}:=\tilde{X}_{s}^{P}-P$, $T\leq s\leq T^{\prime}$, with
$\tilde{X}^{P}$ satisfying SDE (\ref{SDE}) on $[T,T^{\prime}]$ with initial
condition $\tilde{X}_{T}^{P}=X_{T}^{P}+P$. Thus, using the first part of the
proposition and Theorem \ref{DO1} (ii), we obtain that%
\[
u^{P}\left(  t,\xi;T^{\prime}\right)  =\operatorname*{essinf}_{\eta \in
L^{0,+}\left(  \mathcal{F}_{t}\right)  }\left(  \tilde{u}^{P}\left(
t,\eta;T^{\prime}\right)  +\xi \eta \right)  =\operatorname*{essinf}_{\eta \in
L^{0,+}\left(  \mathcal{F}_{t}\right)  }\left(  \tilde{u}^{P}\left(
t,\eta;T\right)  +\xi \eta \right)  =u^{P}\left(  t,\xi;T\right)  ,
\]
and we easily conclude.

\subsubsection{\emph{Proof of Theorem \ref{Mar}}}

In line with \cite[Theorem 3.4]{FI2020}, the main steps are to verify all the
conditions in Theorem \ref{3.3}, which in turn yields the global existence and
uniqueness of the FBSDE \eqref{MarFBSDE}. For notational simplicity, we omit
the superscript $P $ from the FBSDE in the following.

Denote%
\[
B_{s}^{1}:=W_{s}^{1}+\int_{0}^{s}\theta \left(  V_{r}\right)  dr\quad \text{and}
\quad B_{s}^{2}:=W_{s}^{2}\text{,}\quad0\leq s\leq T\text{.}%
\]
Then, $B$ is a Brownian motion under a probability measure $\mathbb{Q}$,
equivalent to $\mathbb{P}$, by Girsanov's theorem. We note that $B$ and $W$
generate the same augmented filtration $\mathbb{F}=(\mathcal{F}_{t})_{t\geq0}$.

Therefore, the aim is now to solve FBSDE (\ref{MarFBSDE}) including the SDE
for stochastic factor $V$ in (\ref{sf}). To this end, we consider a FBSDE
satisfying MLLC in Definition \ref{Def_MLLC}. Namely, for sufficient small
$\varepsilon>0$ and any $0\leq t\leq s\leq T$,%
\begin{equation}
\left \{
\begin{array}
[c]{l}%
\tilde{V}_{s}=\tilde{v}+%
{\displaystyle \int_{t}^{s}}
\dfrac{1}{\varepsilon}\left(  l\left(  \varepsilon \tilde{V}_{r}\right)
-\rho \theta \left(  \varepsilon \tilde{V}_{r}\right)  \right)  dr+\dfrac
{1}{\varepsilon}%
{\displaystyle \int_{t}^{s}}
\left(  \rho dB_{r}^{1}+\sqrt{1-\rho^{2}}dB_{r}^{2}\right)  ,\text{ }\rho
\in(0,1),\\
X_{s}=x-%
{\displaystyle \int_{t}^{s}}
\left(  \psi \left(  r,\varepsilon \tilde{V}_{r},X_{r}+Y_{r}\right)  +Z_{r}%
^{1}\right)  dB_{r}^{1},\\
Y_{s}=P\left(  \varepsilon \tilde{V}_{T},X_{T}\right)  +%
{\displaystyle \int_{s}^{T}}
\left(  \dfrac{1}{2}\phi^{1}\left(  r,\varepsilon \tilde{V}_{r},X_{r}%
+Y_{r}\right)  \left \vert Z_{r}^{2}\right \vert ^{2}+\sqrt{1-\rho^{2}}\phi
^{2}\left(  r,\varepsilon \tilde{V}_{r},X_{r}+Y_{r}\right)  Z_{r}^{2}\right)
dr\\
\text{ \  \  \  \ }-%
{\displaystyle \int_{s}^{T}}
Z_{r}^{1}dB_{r}^{1}-%
{\displaystyle \int_{s}^{T}}
Z_{r}^{2}dB_{r}^{2}.
\end{array}
\right.  \label{FBSDEn}%
\end{equation}
This FBSDE is equivalent to the original one in that, if $(\tilde{V},X,Y,Z)$
solves FBSDE (\ref{FBSDEn}), then $V=\varepsilon \tilde{V}$ solves SDE
(\ref{sf}) with a modified initial condition and $(X,Y,Z)$ solves the original
FBSDE (\ref{MarFBSDE}). Therefore, it suffices to study FBSDE (\ref{FBSDEn})
which satisfies MLLC in Definition \ref{Def_MLLC}. To align with the form of
FBSDE (\ref{abMLLC}), we use the notation%
\[%
\begin{array}
[c]{l}%
H\left(  v,x\right)  =P\left(  v,x\right)  ,\quad \tilde{\mu}\left(
r,v\right)  =l\left(  v\right)  -\rho \theta \left(  v\right)  ,\\
\bar{\sigma}\left(  r,v,x+y,z^{1}\right)  =-\psi \left(  r,v,x+y\right)
-z^{1},\\
f\left(  r,v,x+y,z^{2}\right)  =-\dfrac{1}{2}\phi^{1}\left(  r,v,x+y\right)
\left \vert z^{2}\right \vert ^{2}-\sqrt{1-\rho^{2}}\phi^{2}\left(
r,v,x+y\right)  z^{2},
\end{array}
\]
and $B,\tilde{V},X$, respectively, play the role of $W,\tilde{X},\bar{X}$.

Next, we verify that all the conditions of Theorem \ref{3.3} hold which will,
in turn, yield that $I_{\max}^{M}=[0,T]$ and the global existence of the
solution would follow.

To this end, we first note that all conditions listed before Theorem \ref{3.3}
are satisfied. Indeed, the conditions for $H$ and $\tilde{\mu}$ follow from
Assumption \ref{A.2} (ii) and (iii), and one can easily verify the conditions
for $\bar{\sigma}$ and $f$ by taking derivatives and using Assumption
\ref{A.1} (ii) and Assumption \ref{A.2} (i).

Thus, the main task is to deduce a uniform bound of $w_{x}$ for any weakly
regular decoupling field $w:[t,T]\times \mathbb{R}\times \mathbb{R}%
\rightarrow \mathbb{R}$ of (\ref{FBSDEn}). Similarly to the proof of
\cite[Theorem 3.4]{FI2020}, the main tool is to use the dynamics of the
process%
\[
\Phi_{s}:=w_{x}\left(  s,\tilde{V}_{s},X_{s}\right)  \text{,}%
\]
which is bounded by $L_{\bar{\sigma},z^{1}}^{-1}=1$ (but not uniformly) due to
the weakly regularity. Note, however, that $w$ is also strongly regular by
Theorem \ref{MLLC}. Thus, differentiating FBSDE (\ref{FBSDEn}) with respect to
$x$, using the chain rule (see \cite[Lemma A.3.1]{Fthesis}), we obtain that
the weak derivatives of $X$ and $Y$ with respect to $x$ are
\begin{align*}
\partial_{x}X_{s}  &  =1+\int_{t}^{s}\left(  \delta_{x+y}^{\bar{\sigma}%
}\left(  \partial_{x}X_{r}+\partial_{x}Y_{r}\right)  -\partial_{x}Z_{r}%
^{1}\right)  dB_{r}^{1},\\
\partial_{x}Y_{s}  &  =\partial_{x}P\left(  \varepsilon \tilde{V}_{T}%
,X_{T}\right)  -\int_{s}^{T}\left(  \delta_{x+y}^{f}\left(  \partial_{x}%
X_{r}+\partial_{x}Y_{r}\right)  +\delta_{z}^{f}\partial_{x}Z_{r}^{2}\right)
dr-\int_{s}^{T}\partial_{x}Z_{r}^{1}dB_{r}^{1}-\int_{s}^{T}\partial_{x}%
Z_{r}^{2}dB_{r}^{2},
\end{align*}
where
\begin{align}
&  \delta_{x+y}^{\bar{\sigma}}:=-\psi_{x}\left(  r,\varepsilon \tilde{V}%
_{r},X_{r}+Y_{r}\right)  ,\nonumber \\
&  \delta_{x+y}^{f}:=-\dfrac{1}{2}\phi_{x}^{1}\left(  r,\varepsilon \tilde
{V}_{r},X_{r}+Y_{r}\right)  \left \vert Z_{r}^{2}\right \vert ^{2}-\sqrt
{1-\rho^{2}}\phi_{x}^{2}\left(  r,\varepsilon \tilde{V}_{r},X_{r}+Y_{r}\right)
Z_{r}^{2},\label{df}\\
&  \delta_{z}^{f}:=-\phi^{1}\left(  r,\varepsilon \tilde{V}_{r},X_{r}%
+Y_{r}\right)  Z_{r}^{2}-\sqrt{1-\rho^{2}}\phi^{2}\left(  r,\varepsilon
\tilde{V}_{r},X_{r}+Y_{r}\right)  .\nonumber
\end{align}
Next, we observe that, by Assumption \ref{A.1}, Assumption \ref{A.2} and the
uniform boundedness of $Z$ (see Remark \ref{ReZ}), $\delta_{x+y}^{\bar{\sigma
}}$, $\delta_{x+y}^{f}$ and $\delta_{z}^{f}$ are all uniformly bounded
independently of $t$ and $\varepsilon$.

Let $\tau$, $t\leq \tau \leq T$, be any stopping time such that $\partial
_{x}X>0$ a.e. on $[t,\tau]$. Using the chain rule in \cite[Lemma
A.3.1]{Fthesis}, we have%
\[
\partial_{x}Y_{s}=\partial_{x}\left(  w\left(  s,\tilde{V}_{s},X_{s}\right)
\right)  =\Phi_{s}\partial_{x}X_{s}\text{,}\quad0\leq t\leq s\leq \tau,
\]
i.e., $\Phi_{s}=\frac{\partial_{x}Y_{s}}{\partial_{x}X_{s}}$. Applying
It\^{o}'s formula yields that%
\[
d(\frac{1}{\partial_{x}X_{r}})=-\frac{1}{\partial_{x}X_{r}}\left(
\delta_{x+y}^{\bar{\sigma}}\left(  1+\Phi_{r}\right)  -\frac{\partial_{x}%
Z_{r}^{1}}{\partial_{x}X_{r}}\right)  dB_{r}^{1}+\frac{1}{\partial_{x}X_{r}%
}\left \vert \delta_{x+y}^{\bar{\sigma}}\left(  1+\Phi_{r}\right)
-\frac{\partial_{x}Z_{r}^{1}}{\partial_{x}X_{r}}\right \vert ^{2}dr,
\]
and, thus,%
\begin{align*}
&  d\Phi_{r}=\left(  \frac{\partial_{x}Z_{r}^{1}}{\partial_{x}X_{r}}-\Phi
_{r}\delta_{x+y}^{\bar{\sigma}}\right)  \left(  1+\Phi_{r}\right)  dB_{r}%
^{1}+\frac{\partial_{x}Z_{r}^{2}}{\partial_{x}X_{r}}dB_{r}^{2}\\
&  +\left(  \Phi_{r}\left \vert \delta_{x+y}^{\bar{\sigma}}\left(  1+\Phi
_{r}\right)  -\frac{\partial_{x}Z_{r}^{1}}{\partial_{x}X_{r}}\right \vert
^{2}+\delta_{x+y}^{f}\left(  1+\Phi_{r}\right)  +\delta_{z}^{f}\frac
{\partial_{x}Z_{r}^{2}}{\partial_{x}X_{r}}-\frac{\partial_{x}Z_{r}^{1}%
}{\partial_{x}X_{r}}\left(  \delta_{x+y}^{\bar{\sigma}}\left(  1+\Phi
_{r}\right)  -\frac{\partial_{x}Z_{r}^{1}}{\partial_{x}X_{r}}\right)  \right)
dr.
\end{align*}
Let,
\[
\hat{Z}_{r}^{1}:=\left(  \frac{\partial_{x}Z_{r}^{1}}{\partial_{x}X_{r}}%
-\Phi_{r}\delta_{x+y}^{\bar{\sigma}}\right)  \left(  1+\Phi_{r}\right)
\quad \text{and}\quad \hat{Z}_{r}^{2}:=\frac{\partial_{x}Z_{r}^{2}}{\partial
_{x}X_{r}},\quad s\leq r\leq \tau.
\]
Then, we obtain that $(\Phi,\hat{Z})$ solves
\begin{align}
\Phi_{s}=  &  \Phi_{\tau}-\int_{s}^{\tau}\left(  \delta_{x+y}^{f}\left(
1+\Phi_{r}\right)  +\delta_{z}^{f}\hat{Z}_{r}^{2}-\hat{Z}_{r}^{1}\left(
\delta_{x+y}^{\bar{\sigma}}\left(  1+\Phi_{r}\right)  -\frac{\hat{Z}_{r}%
^{1}+\Phi_{r}\delta_{x+y}^{\bar{\sigma}}\left(  1+\Phi_{r}\right)  }%
{1+\Phi_{r}}\right)  \right)  dr\label{2202}\\
&  -\int_{s}^{\tau}\hat{Z}_{r}^{\top}dB_{r},\nonumber
\end{align}
which is a quadratic BSDE. According to \cite[Theorem A.1.11]{Fthesis}, and
using that $\Phi$ is uniformly bounded and $\hat{Z}$ is well-defined in
$L^{2}$ by strong regularity, we have that $\hat{Z}$ is a BMO process. Using a
similar argument as in \cite[Lemma A.1]{FI2020}, it can be shown that $\tau=T$
in (\ref{2202}). In turn, applying It\^{o}'s formula gives
\begin{equation}
\ln \left(  1+\Phi_{s}\right)  =\ln \left(  1+\Phi_{T}\right)  -\int_{s}%
^{T}\left(  \delta_{z^{2}}^{f}\bar{Z}_{r}^{2}+\delta_{x+y}^{f}-\delta
_{x+y}^{\bar{\sigma}}\bar{Z}_{r}^{1}+\frac{1}{2}\left \vert \bar{Z}_{r}%
^{1}\right \vert ^{2}-\frac{1}{2}\left \vert \bar{Z}_{r}^{2}\right \vert
^{2}\right)  dr-\int_{s}^{T}\bar{Z}_{r}^{\top}dB_{r}, \label{12.4}%
\end{equation}
where $\bar{Z}^{i}:=\frac{\hat{Z}^{i}}{1+\Phi},i=1,2$. Using once more
\cite[Theorem A.1.11]{Fthesis}, we obtain that $\bar{Z}^{i}$ is also a BMO process.

Next, we use (\ref{12.4}) to obtain the upper and lower bounds for $w_{x}$.

\underline{Upper bound:} Rewrite (\ref{12.4}) as%
\begin{align*}
\ln \left(  1+\Phi_{s}\right)  =  &  \ln \left(  1+\Phi_{T}\right)  -\int
_{s}^{T}\delta_{x+y}^{f}dr\\
&  -\int_{s}^{T}\bar{Z}_{r}^{1}\left(  dB_{r}^{1}+\left(  \frac{1}{2}\bar
{Z}_{r}^{1}-\delta_{x+y}^{\bar{\sigma}}\right)  dr\right)  -\int_{s}^{T}%
\bar{Z}_{r}^{2}\left(  dB_{r}^{2}+\left(  \delta_{z^{2}}^{f}-\frac{1}{2}%
\bar{Z}_{r}^{2}\right)  dr\right) \\
:=  &  \ln \left(  1+\Phi_{T}\right)  -\int_{s}^{T}\delta_{x+y}^{f}dr-\int
_{s}^{T}\bar{Z}_{r}^{\top}d\tilde{B}_{r},
\end{align*}
where%
\[
d\tilde{B}_{r}^{1}:=dB_{r}^{1}+\left(  \frac{1}{2}\bar{Z}_{r}^{1}-\delta
_{x+y}^{\bar{\sigma}}\right)  dr\text{, \  \ }d\tilde{B}_{r}^{2}:=dB_{r}%
^{2}+\left(  \delta_{z^{2}}^{f}-\frac{1}{2}\bar{Z}_{r}^{2}\right)  dr\text{,}%
\]
and $\frac{1}{2}\bar{Z}_{r}^{1}-\delta_{x+y}^{\bar{\sigma}}$ and
$\delta_{z^{2}}^{f}-\frac{1}{2}\bar{Z}_{r}^{2}$ are both BMO process. Using
Girsanov's theorem for BMO processes, we have that $\tilde{B}$ is Brownian
motion under an equivalent probability measure $\mathbb{\tilde{Q}}$.
Therefore,
\begin{equation}
\mathbb{E}^{\mathbb{\tilde{Q}}}\left[  \ln \left(  1+\Phi_{s}\right)  \right]
=\mathbb{E}^{\mathbb{\tilde{Q}}}\left[  \ln \left(  1+\Phi_{T}\right)
-\int_{s}^{T}\delta_{x+y}^{f}dr\right]  . \label{12.6}%
\end{equation}
Recall the definition of $\delta_{x+y}^{f}$ in (\ref{df}), where $\phi_{x}%
^{1}<0$ from Assumption \ref{A.2} (i). Noting that $\delta_{x+y}^{f}$ can be
in fact seen as a quadratic function of $Z^{2}$, we have%
\begin{equation}
\delta_{x+y}^{f}\geq \frac{1-\rho^{2}}{2}\frac{\left \vert \phi_{x}^{2}\left(
r,\varepsilon \tilde{V}_{r},X_{r}+Y_{r}\right)  \right \vert ^{2}}{\phi_{x}%
^{1}\left(  r,\varepsilon \tilde{V}_{r},X_{r}+Y_{r}\right)  }. \label{CFA}%
\end{equation}
Next, we look at the following distinct cases.

Case 1. $\phi_{x}^{2}=0$. Then, it is obvious that $\delta_{x+y}^{f}\geq0$
and, thus,
\[
\mathbb{E}^{\mathbb{\tilde{Q}}}\left[  \ln \left(  1+\Phi_{s}\right)  \right]
\leq \mathbb{E}^{\mathbb{\tilde{Q}}}\left[  \ln \left(  1+\Phi_{T}\right)
\right]  \leq \ln \left(  1+L_{P,x}\right) .
\]
In particular, $\ln(1+w_{x}(t,\tilde{v},x))\leq \ln \left(  1+L_{P,x}\right)  $.
It follows that $w_{x}(t,\tilde{v},x)\leq L_{P,x}<1$.

Case 2. $\phi_{x}^{2}\neq0$. Since $\frac{1-\rho^{2}}{2}\frac{|\phi_{x}%
^{2}|^{2}}{\phi_{x}^{1}}\geq-K$, we obtain from (\ref{CFA}) and (\ref{12.6})
that%
\[
\mathbb{E}^{\mathbb{\tilde{Q}}}\left[  \ln \left(  1+\Phi_{s}\right)  \right]
\leq \mathbb{E}^{\mathbb{\tilde{Q}}}\left[  \ln \left(  1+\Phi_{T}\right)
\right]  +KT\leq \ln \left(  1+L_{P,x}\right)  +KT,
\]
and, in particular,
\[
\ln(1+w_{x}(t,\tilde{v},x))\leq \ln \left(  1+L_{P,x}\right)  +KT.
\]
Therefore,%
\[
w_{x}(t,\tilde{v},x)\leq \left(  1+L_{P,x}\right)  e^{KT}-1<1.
\]

In summary, under either Case 1 or Case 2, $w_{x}(t,\tilde{v},x)$ is uniformly
bounded by $1$ from above.

\underline{Lower bound:} Using (\ref{12.4}) yields%
\begin{align*}
\ln \left(  1+\Phi_{s}\right)  =  &  \ln \left(  1+\Phi_{T}\right)  -\int
_{s}^{T}\left(  \delta_{x+y}^{f}-\frac{1}{2}\left \vert \bar{Z}_{r}%
^{1}\right \vert ^{2}-\frac{1}{2}\left \vert \bar{Z}_{r}^{2}\right \vert
^{2}\right)  dr\\
&  -\int_{s}^{T}\bar{Z}_{r}^{1}\left(  dB_{r}^{1}+\left(  \bar{Z}_{r}%
^{1}-\delta_{x+y}^{\bar{\sigma}}\right)  dr\right)  -\int_{s}^{T}\bar{Z}%
_{r}^{2}\left(  dB_{r}^{2}+\delta_{z^{2}}^{f}dr\right) \\
:=  &  \ln \left(  1+\Phi_{T}\right)  -\int_{s}^{T}\left(  \delta_{x+y}%
^{f}-\frac{1}{2}\left \vert \bar{Z}_{r}^{1}\right \vert ^{2}-\frac{1}%
{2}\left \vert \bar{Z}_{r}^{2}\right \vert ^{2}\right)  dr-\int_{s}^{T}\bar
{Z}_{r}^{\top}d\hat{B}_{r},
\end{align*}
where $\hat{B}$ defined by
\begin{align*}
d\hat{B}_{r}^{1}  &  :=dB_{r}^{1}+\left(  \bar{Z}_{r}^{1}-\delta_{x+y}%
^{\bar{\sigma}}\right)  dr=dB_{r}^{1}+\left(  \bar{Z}_{r}^{1}+\theta
\varphi_{x}+\rho \left(  \phi^{2}-\frac{u_{xv}}{u_{xx}}\phi^{1}\right)
\right)  dr,\\
d\hat{B}_{r}^{2}  &  :=dB_{r}^{2}+\delta_{z^{2}}^{f}dr=dB_{r}^{2}-\left(
\phi^{1}Z_{r}^{2}+\sqrt{1-\rho^{2}}\phi^{2}\right)  dr,
\end{align*}
is a Brownian motion under an equivalent probability measure $\mathbb{\hat{Q}%
}$ due to the fact that $\bar{Z}^{1}-\delta_{x+y}^{\bar{\sigma}}$ and
$\delta_{z^{2}}^{f}$ are BMO processes. It, then, follows that
\begin{equation}
\mathbb{E}^{\mathbb{\hat{Q}}}\left[  \ln \left(  1+\Phi_{s}\right)  \right]
=\mathbb{E}^{\mathbb{\hat{Q}}}\left[  \ln \left(  1+\Phi_{T}\right)  -\int
_{s}^{T}\left(  \delta_{x+y}^{f}-\frac{1}{2}\left \vert \bar{Z}_{r}%
^{1}\right \vert ^{2}-\frac{1}{2}\left \vert \bar{Z}_{r}^{2}\right \vert
^{2}\right)  dr\right]  . \label{13.1}%
\end{equation}
Next, we obtain a uniform upper bound for $\delta_{x+y}^{f}$ in (\ref{df}). To
this end, applying It\^{o}'s formula to $\varphi(r,\varepsilon \tilde{V}%
_{r},X_{r}+Y_{r})$, we deduce, by using the definitions of $\hat{B}$ and of
$\varphi,\phi^{i}$ and $\psi$, that,
\begin{align*}
&  d\varphi \left(  r,\varepsilon \tilde{V}_{r},X_{r}+Y_{r}\right) \\
=  &  \left(  -\frac{1}{2}\varphi \phi_{x}^{1}\left \vert Z_{r}^{2}\right \vert
^{2}-\sqrt{1-\rho^{2}}\varphi \phi_{x}^{2}Z_{r}^{2}+\left(  \varphi_{x}%
\psi-\rho \varphi_{v}\right)  \bar{Z}_{r}^{1}\right. \\
&  \left.  +\varphi_{t}+\varphi_{v}\left(  l-\rho \theta \right)  +\frac{1}%
{2}\varphi_{vv}+\frac{1}{2}\varphi_{xx}\left \vert \psi \right \vert ^{2}%
+\psi \left(  \varphi_{x}\psi_{x}-\rho \varphi_{xv}\right)  +\varphi_{v}\left(
-\rho \psi_{x}+\left(  1-\rho^{2}\right)  \phi^{2}\right)  \right)  dr\\
&  +\left(  \rho \varphi_{v}-\varphi_{x}\psi \right)  d\hat{B}_{r}^{1}+\left(
\sqrt{1-\rho^{2}}\varphi_{v}+\varphi_{x}Z_{r}^{2}\right)  d\hat{B}_{r}^{2}\\
=:  &  \left(  -\frac{1}{2}\varphi \phi_{x}^{1}\left \vert Z_{r}^{2}\right \vert
^{2}-\sqrt{1-\rho^{2}}\varphi \phi_{x}^{2}Z_{r}^{2}+\beta_{r}\bar{Z}_{r}%
^{1}+\alpha_{r}\right)  dr+\left(  \rho \varphi_{v}-\varphi_{x}\psi \right)
d\hat{B}_{r}^{1}+\left(  \sqrt{1-\rho^{2}}\varphi_{v}+\varphi_{x}Z_{r}%
^{2}\right)  d\hat{B}_{r}^{2},
\end{align*}
where%
\begin{align*}
\beta_{r}  &  :=\varphi_{x}\psi-\rho \varphi_{v},\\
\alpha_{r}  &  :=\varphi_{t}+\varphi_{v}\left(  l-\rho \theta \right)  +\frac
{1}{2}\varphi_{vv}+\frac{1}{2}\varphi_{xx}\left \vert \psi \right \vert ^{2}%
+\psi \left(  \varphi_{x}\psi_{x}-\rho \varphi_{xv}\right)  +\varphi_{v}\left(
-\rho \psi_{x}+\left(  1-\rho^{2}\right)  \phi^{2}\right)  .
\end{align*}
Above, we omitted arguments $(r,\varepsilon \tilde{V}_{r},X_{r}+Y_{r})$ for
simplicity. From Assumptions \ref{A.1} and \ref{A.2}, we deduce that
$\alpha,\beta,\varphi \phi_{x}^{1},\varphi \phi_{x}^{2}$ and $\varphi
_{v},\varphi_{x},\psi$ are uniformly bounded. Thus, taking expectation on both
sides above yields
\begin{align}
&  \mathbb{E}^{\mathbb{\hat{Q}}}\left[  \varphi \left(  s,\varepsilon \tilde
{V}_{s},X_{s}+Y_{s}\right)  \right] \nonumber \\
=  &  \mathbb{E}^{\mathbb{\hat{Q}}}\left[  \varphi \left(  T,\varepsilon
\tilde{V}_{T},X_{T}+Y_{T}\right)  \right]  -\mathbb{E}^{\mathbb{\hat{Q}}%
}\left[  \int_{s}^{T}\left(  -\frac{1}{2}\varphi \phi_{x}^{1}\left \vert
Z_{r}^{2}\right \vert ^{2}-\sqrt{1-\rho^{2}}\varphi \phi_{x}^{2}Z_{r}^{2}%
+\beta_{r}\bar{Z}_{r}^{1}+\alpha_{r}\right)  dr\right] \nonumber \\
=  &  \mathbb{E}^{\mathbb{\hat{Q}}}\left[  \varphi \left(  T,\varepsilon
\tilde{V}_{T},X_{T}+Y_{T}\right)  \right]  +\mathbb{E}^{\mathbb{\hat{Q}}%
}\left[  \int_{s}^{T}\left(  -\varphi \right)  \delta_{x+y}^{f}dr\right]
-\mathbb{E}^{\mathbb{\hat{Q}}}\left[  \int_{s}^{T}\left(  \beta_{r}\bar{Z}%
_{r}^{1}+\alpha_{r}\right)  dr\right]  . \label{BB}%
\end{align}
It follows by Assumption \ref{A.1} (ii) that
\begin{align*}
\mathbb{E}^{\mathbb{\hat{Q}}}\left[  \int_{s}^{T}\delta_{x+y}^{f}dr\right]
&  =\mathbb{E}^{\mathbb{\hat{Q}}}\left[  \int_{s}^{T}\delta_{x+y}%
^{f}\mathbf{1}_{\delta_{x+y}^{f}<0}dr\right]  +\mathbb{E}^{\mathbb{\hat{Q}}%
}\left[  \int_{s}^{T}\delta_{x+y}^{f}\mathbf{1}_{\delta_{x+y}^{f}\geq
0}dr\right] \\
&  \leq \frac{1}{C_{u}}\mathbb{E}^{\mathbb{\hat{Q}}}\left[  \int_{s}^{T}\left(
-\varphi \right)  \delta_{x+y}^{f}\mathbf{1}_{\delta_{x+y}^{f}<0}dr\right]
+\frac{1}{C_{l}}\mathbb{E}^{\mathbb{\hat{Q}}}\left[  \int_{s}^{T}\left(
-\varphi \right)  \delta_{x+y}^{f}\mathbf{1}_{\delta_{x+y}^{f}\geq0}dr\right]
\\
&  =\frac{1}{C_{l}}\mathbb{E}^{\mathbb{\hat{Q}}}\left[  \int_{s}^{T}\left(
-\varphi \right)  \delta_{x+y}^{f}dr\right]  -\frac{C_{u}-C_{l}}{C_{l}C_{u}%
}\mathbb{E}^{\mathbb{\hat{Q}}}\left[  \int_{s}^{T}\left(  -\varphi \right)
\delta_{x+y}^{f}\mathbf{1}_{\delta_{x+y}^{f}<0}dr\right]  ,
\end{align*}
which, together with (\ref{CFA}), gives
\[
\mathbb{E}^{\mathbb{\hat{Q}}}\left[  \int_{s}^{T}\delta_{x+y}^{f}dr\right]
\leq \frac{1}{C_{l}}\mathbb{E}^{\mathbb{\hat{Q}}}\left[  \int_{s}^{T}\left(
-\varphi \right)  \delta_{x+y}^{f}dr\right]  +\frac{C_{u}-C_{l}}{C_{l}}KT,
\]
for both Case 1 and Case 2. Combining the above with (\ref{BB}), we deduce%
\begin{align*}
\mathbb{E}^{\mathbb{\hat{Q}}}\left[  \int_{s}^{T}\delta_{x+y}^{f}dr\right]
&  \leq \frac{1}{C_{l}}\left(  C_{u}+T\left \Vert \alpha \right \Vert _{\infty
}\right)  +\mathbb{E}^{\mathbb{\hat{Q}}}\left[  \int_{s}^{T}\frac{\beta_{r}%
}{C_{l}}\bar{Z}_{r}^{1}dr\right]  +\frac{C_{u}-C_{l}}{C_{l}}KT\\
&  :=\mathbb{E}^{\mathbb{\hat{Q}}}\left[  \int_{s}^{T}\frac{\beta_{r}}{C_{l}%
}\bar{Z}_{r}^{1}dr\right]  +K_{b}.
\end{align*}
Thus (\ref{13.1}) implies%
\begin{align*}
\mathbb{E}^{\mathbb{\hat{Q}}}\left[  \ln \left(  1+\Phi_{s}\right)  \right]
&  \geq \mathbb{E}^{\mathbb{\hat{Q}}}\left[  \ln \left(  1+\Phi_{T}\right)
\right]  -\mathbb{E}^{\mathbb{\hat{Q}}}\left[  \int_{s}^{T}\delta_{x+y}%
^{f}dr\right]  +\mathbb{E}^{\mathbb{\hat{Q}}}\left[  \int_{s}^{T}\frac{1}%
{2}\left \vert \bar{Z}_{r}^{1}\right \vert ^{2}dr\right] \\
&  \geq \ln \left(  1-L_{H,x}\right)  -\mathbb{E}^{\mathbb{\hat{Q}}}\left[
\int_{s}^{T}\frac{\beta_{r}}{C_{l}}\bar{Z}_{r}^{1}dr\right]  +\mathbb{E}%
^{\mathbb{\hat{Q}}}\left[  \int_{s}^{T}\frac{1}{2}\left \vert \bar{Z}_{r}%
^{1}\right \vert ^{2}dr\right]  -K_{b}\\
&  \geq \ln \left(  1-L_{H,x}\right)  -\frac{1}{2}\mathbb{E}^{\mathbb{\hat{Q}}%
}\left[  \int_{s}^{T}\left \vert \frac{\beta_{r}}{C_{l}}\right \vert
^{2}dr\right]  -K_{b}\\
&  \geq \ln \left(  1-L_{H,x}\right)  -\frac{T}{2\left \vert C_{l}\right \vert
^{2}}\left \Vert \beta \right \Vert _{\infty}^{2}-K_{b}=:-C_{b},
\end{align*}
where $C_{b}>0$ is independent of both $\varepsilon$ and time. It, then,
follows that $\ln \left(  1+w_{x}\left(  t,\tilde{v},x\right)  \right)
\geq-C_{b}$, and, in turn,
\[
w_{x}\left(  t,\tilde{v},x\right)  \geq e^{-C_{b}}-1>-1\text{.}%
\]

To conclude, we note that $w_{x}$ can be bounded uniformly by $1$ and, thus,
we can apply Theorem \ref{3.3} to obtain the global existence result. The
uniqueness can be proved by Lemma 2.5 in \cite{FI2020} and is omitted.


\begin{thebibliography}{99}                                                                                               %


\bibitem {Angoshtari}{B. Angoshtari, Predictable forward performance processes
in complete markets, Probability, Uncertainty and Quantitative Risk, 8(2)
(2023), pp. 141-176}.

\bibitem {Angoshtari2020}{B. Angoshtari, T. Zariphopoulou and X. Zhou,
Predictable forward performance processes: The binomial case, SIAM J. Control
Optim., 58(1) (2020), pp. 327-347}.

\bibitem {A2014}\emph{M. Anthropelos, Forward exponential performances:
Pricing and optimal risk sharing, SIAM J. Financial Math., 5(1) (2014), pp.
626-655.}

\bibitem {Anthropelos2022}{M. Anthropelos, T. Geng and T. Zariphopoulou,
Competition in fund management and forward relative performance criteria,
{SIAM J. Financial Math.,} 13(4) (2022), pp. 1271-1301}.

\bibitem {Rufloff2017}\emph{C. Ararat, A. H. Hamel and B. Rudloff, Set-valued
shortfall and divergence risk measures, Int. J. Theor. Appl. Finance, 20(05),
(2017), pp. 1-48.}

\bibitem {Avanesyan_2018}{L. Avanesyan, M. Shkolnikov and R. Sicar,
Construction of a class of forward performance processes in stochastic factor
models, and an extension of Widder's theorem, {Finance Stoch., } 24 (2020),
pp. 981-1011.}

\bibitem {Bartl2020}\emph{D. Bartl, S. Drapeau and L. Tangpi, Computational
aspects of robust optimized certainty equivalents and option pricing. Math.
Finance, 30(1), (2020), pp. 287-309.}

\bibitem {BT1986}A. Ben-Tal and M. Teboulle, Expected utility, penalty
functions and duality in stochastic nonlinear programming, Manage. Sci., 32
(1986), pp. 1445-1466.

\bibitem {BT2007}A. Ben-Tal and M. Teboulle, An old-new concept of convex risk
measures: The optimized certainty equivalent, Math. Finance, 17(3) (2007), pp. 449-476.

\bibitem {BT2020}J. Backhoff-Veraguas and L. Tangpi, On the dynamic
representation of some time-inconsistent risk measures in a Brownian
filtration, Math. Finan. Econ., 14 (2020), pp. 433-460.

\bibitem {BRT2022}J. Backhoff-Veraguas, A. M. Reppen and L. Tangpi, Stochastic
control of optimized certainty equivalents, SIAM J. Financial Math., 13
(2022), pp. 745-772.

\bibitem {Bo_2023}{L. Bo, A. Capponi and C. Zhou, Power forward performance in
semimartingale markets with stochastic integrated factors, {Math. Oper. Res.},
48(1) (2023), pp. 288-312}.

\bibitem {BH2006}\emph{P. Briand and Y. Hu, BSDE with quadratic growth and
unbounded terminal value. Probab. Theory Related Fields, 136 (2006) 604-618.}

\bibitem {Carmona}{R. Carmona (Ed.), Indifference pricing: theory and
applications, Princeton University Press, Princeton, 2008}.

\bibitem {CK2007}A. S. Cherny and M. Kupper, Divergence utilities, Working
paper, (2007), Available at SSRN 1023525.

\bibitem {Chong_2018}{W. F. Chong, Pricing and hedging equity-linked life
insurance contracts beyond the classical paradigm: the principle of equivalent
forward preferences, {Insurance: Mathematics and Economics}, 88 (2019), pp.
93-107.}

\bibitem {CHLZ2019}W. F. Chong, Y. Hu, G. Liang and T. Zariphopoulou, An
ergodic BSDE approach to forward entropic risk measures: Representation and
large-maturity behavior, Finance Stoch., 23 (2019), pp. 239-273.

\bibitem {CM2017}\emph{T. Choulli and J. Ma, Explicit description of HARA
forward utilities and their optimal portfolios, Theory. Probab. Appl., 61(1)
(2017) pp. 57-93.}

\bibitem {Choulli_2007}{T. Choulli, C. Stricker and J. Li, Minimal Hellinger
martingale measures of order $q$, Finance Stoch., 11 (2007), pp. 399-427.}

\bibitem {DPR2010}F. Delbaen, S. Peng, and E. Rosazza Gianin, Representation
of the penalty term of dynamic concave utilities, Finance Stoch., 14 (2010),
pp. 449-472.

\bibitem {El_Karoui_2018}{N. El Karoui, C. Hillairet and M. Mrad, Consistent
utility of investment and consumption: a forward/backward SPDE viewpoint,
{Stochastics}, 90(6) (2018), pp. 927-954.}

\bibitem {El_Karoui_2022}{N. El Karoui, C. Hillairet and M. Mrad, Ramsey rule
with forward/backward utility for long-term yield curves modeling, {Decisions
in Economics and Finance}, 45 (2022), pp. 375-414.}

\bibitem {EM2014}N. El Karoui and M. Mrad, An exact connection between two
solvable SDEs and a nonlinear utility stochastic PDEs, SIAM J. Financial
Math., 4 (2013), pp. 697-736.

\bibitem {El_Karoui_2021}{N. El Karoui and M. Mrad, Recover dynamic utility
from observable process: application to the economic equilibrium, SIAM J.
Financial Math., 12(1) (2022), pp. 189-225.}

\bibitem {Fthesis}\emph{A. Fromm, Theory and applications of decoupling fields
for forward-backward stochastic differential equations. Ph.D. thesis, Humboldt
Univ., Berlin (2015).}

\bibitem {FI2020}\emph{A. Fromm and P. Imkeller, Utility maximization via
decoupling fields, Ann. Appl. Probab. 30 (6) (2020), pp. 2665-2694.}



\bibitem {He_2021}{X. D. He, M. S. Strub and T. Zariphopoulou, Forward
rank-dependent performance criteria: time-consistent investment under
probability distortion, {Math. Finance,} 31(2) (2021), pp. 683-721.}

\bibitem {HHIRZ}U. Horst, Y. Hu, P. Imkeller, A. Reveillac and J. Zhang,
Forward-backward systems for expected utility maximization, Stochastic
Process. Appl. 124 (5) (2014), pp. 1813-1848.

\bibitem {HIM2005}Y. Hu, P. Imkeller and M. M{\"{u}}ller, Utility maximization
in incomplete markets, {Ann. Appl. Probab.}, 15(3) (2005), pp. 1691-1712.

\bibitem {Hu2020}{ Y. Hu, G. Liang and S. Tang, Systems of ergodic BSDEs
arising in regime switching forward performance processes, {SIAM J. Control
Optim.}, 58(4) (2020), pp. 2503-2534}.

\bibitem {HLT2024}\emph{ Y. Hu, G. Liang and S. Tang, Utility maximization in
constrained and unbounded financial markets: Applications to indifference
valuation, regime switching, consumption and Epstein-Zin recursive utility,
arXiv preprint, arXiv:1707.00199.}

\bibitem{KMT2022} \emph{S. Kaakai, A. Matoussi and A. Tamtalini. Multivariate optimized certainty equivalent risk measures and their numerical computation. arXiv preprint arXiv:2210.13825.}

\bibitem {Kallblad_2016}{S. K{\"{a}}llblad, Black's inverse investment problem
and forward criteria with consumption, SIAM J. Financial Math., 11(2) (2020),
pp. 494-525.}

\bibitem {Kallblad_2018}{S. K{\"{a}}llblad, J. Ob{\l }{\'{o}}j and T.
Zariphopoulou, Dynamically consistent investment under model uncertainty: the
robust forward criteria, Finance Stoch., 22(4) (2018), pp. 879-918.}

\bibitem {KSBOOK}I. Karatzas and S. Shreve, Brownian Motion and Stochastic
Calculus, Springer, New York, 1988.

\bibitem {KBOOK}N. Kazamaki, Continuous exponential martingales and BMO,
Springer, Berlin, 2006.

\bibitem {K1997}\emph{H. Kunita, Stochastic Flows and Stochastic Differential
Equations, Cambridge Stud. Adv. Math. 24, Cambridge University Press,
Cambridge, UK, 1997.}

\bibitem {LSZ}T. Leung, R. Sircar, and T. Zariphopoulou, Forward indifference
valuation of American options, Stochastics, 84 (2012), pp. 741-770.

\bibitem {Liang2023}{G. Liang, M. Strub and Y. Wang, Predictable forward
performance processes: Infrequent evaluation and applications to human-machine
interactions, {Math. Finance}, 33(4) (2023), 1248-1286}.

\bibitem {LZ2017}G. Liang and T. Zariphopoulou, Representation of homothetic
forward performance processes in stochastic factor models via ergodic and
infinite horizon BSDE, SIAM J. Financial Math., 8 (2017), pp. 344-372.



\bibitem {MZ0}{M. Musiela and T. Zariphopoulou, Investment and valuation under
backward and forward dynamic exponential utilities in a stochastic factor
model, {Advances in Mathematical Finance}, (2007), pp. 303-334.}

\bibitem {MZ-Kurtz}{M. Musiela and T. Zariphopoulou, Optimal asset allocation
under forward exponential performance criteria, {Markov Processes and Related
Topics: A Festschrift for T. G. Kurtz}, IMS Lecture Notes-Monograph Series 4,
(2008), pp. 285-300.}

\bibitem {MZ1}{M. Musiela and T. Zariphopoulou, Portfolio choice under dynamic
investment performance criteria, {Quantitative Finance}, 9(2) (2009), pp.
161-170.}

\bibitem {MZ2}{M. Musiela and T. Zariphopoulou, Portfolio choice under
space-time monotone performance criteria, SIAM J. Financial Math., 1 (2010),
pp. 326-365.}

\bibitem {MZ2010}M. Musiela and T. Zariphopoulou, Stochastic partial
differential equations and portfolio choice, C. Chiarella and A. Novikov,
eds., Contemp. Quant. Financ., Springer, Berlin, (2010), pp. 195-215.

\bibitem {Nadtochiy_2017}{S. Nadtochiy and T. Tehranchi, Optimal investment
for all time horizons and {M}artin boundary of space-time diffusions, Math.
Finance, 27(2) (2017), pp. 438-470.}



\bibitem {NZ2014}\emph{S. Nadtochiy and T. Zariphopoulou, A class of
homothetic forward investment performance processes with non-zero volatility,
Inspired by Finance, Y. Kabanov et al., eds., Springer, Berlin, 2013, pp.
475-505.}

\bibitem {Strub2021}{M. Strub and X. Zhou, Evolution of the Arrow-Pratt
measure of risk-tolerance for predictable forward utility processes. Finance
Stoch., 25(2) (2021), pp. 331-358}.

\bibitem {Wang}\emph{H. Wang, Forward indifference valuation for dynamically
incoming projects. Probability, Uncertainty and Quantitative Risk, 9(2)
(2024), pp. 219-234}.

\bibitem {Zariphopoulou_2024}\emph{T. Zariphopoulou, Mean field and $n$-player
games in Ito-diffusion markets under forward performance criteria,
\textit{Probability, Uncertainty and Quantitative Risk} 9(2), (2024), pp.
123-148.}

\bibitem {ZBOOK}J. Zhang, Backward stochastic differential equations,
Springer, New York, 2017.

\bibitem {Z2009}G. \v{Z}itkovi\'{c}, A dual characterization of
self-generation and exponential forward performances, Ann. Appl. Probab., 19
(2009), pp. 2176-2210.

\bibitem {ZZ2010}T. Zariphopoulou and G. \v{Z}itkovi\'{c},
Maturity-independent risk measures. SIAM J. Financial Math., 1 (2010), pp. 266-288.
\end{thebibliography}
\end{document}